\documentclass[11pt,fleqn,reqno]{amsart} 
\usepackage[numbers]{natbib}
\usepackage{hyperref}
\usepackage{amsfonts,amssymb,amsthm,amsmath}
\usepackage{enumitem} 
\usepackage{xspace}
\usepackage[margin=1.37in]{geometry}

\usepackage{todonotes}
\usepackage{soul}

\theoremstyle{plain}
\newtheorem{theorem}{Theorem}
\newtheorem{proposition}[theorem]{Proposition}
\newtheorem{lemma}[theorem]{Lemma}
\newtheorem{corollary}[theorem]{Corollary}

\newtheorem{definition}[theorem]{Definition}

\numberwithin{theorem}{section}
\numberwithin{equation}{section}

\newcommand{\nc}{\newcommand}

\nc{\be}{\begin{equation}}
\nc{\la}{\label}
\nc{\ba}{\begin{array}}
\nc{\ea}{\end{array}}
\nc{\bs}{\begin{split}}
\nc{\es}{\end{split}}

\newcommand{\R}{\mathbb{R}}
\newcommand{\C}{\mathbb{C}}
\newcommand{\Z}{\mathbb{Z}}

\newcommand{\cC}{\mathcal{C}}

\newcommand{\cL}{\mathcal{L}}
 
\nc{\al}{\alpha}
\nc{\del}{\delta}
\nc{\h}{\delta}
\nc{\Gam}{\Gamma}
\nc{\G}{\Gamma}
\nc{\g}{\gamma}
\nc{\gam}{\gamma}
\nc{\ka}{\kappa}
\nc{\lam}{\lambda}
\nc{\Lam}{\Lambda}
\nc{\Om}{\Omega}
\nc{\om}{\omega}

\nc{\ta}{\tau}
\nc{\w}{\omega}
\nc{\io}{\iota}
\nc{\z}{\zeta}
\nc{\s}{\sigma}
\nc{\Si}{\Sigma}
\nc{\vphi}{\varphi}

\nc{\e}{\epsilon}

\nc{\bP}{\bar{P}}
\nc{\bQ}{\bar{Q}}

\nc{\ran}{\rangle}
\nc{\lan}{\langle}

\newcommand{\ra}{\rightarrow}

\newcommand{\ls}{\lesssim}
\newcommand{\gs}{\gtrsim}
\newcommand{\one}{\mathbf{1}}

\newcommand{\Ran}{\operatorname{Ran}}

\renewcommand{\Re}{\mathrm{Re}} 
\renewcommand{\Im}{\mathrm{Im}} 
\newcommand{\Tr}{\mathrm{Tr}}
\newcommand{\tr}{\mathrm{Tr}}

\newcommand{\dist}{\mathrm{dist}}

\nc{\bfone}{{\bf 1}}
\nc{\1}{{\bf 1}}

\newcommand{\p}{\partial}

\newcommand{\n}{\nabla}

\def\eqn{\begin{align}}
\def\eeqn{\end{align}}

\newcommand{\etal}{\eta_0}

\newcommand{\cbeta}{s_\beta}
\newcommand{\cbet}{c_T}

\newcommand{\UBF}{U_{\rm BF}}
\newcommand{\Udel}{U_{\delta}}
\newcommand{\absoint}{\left| \oint \right|}

\nc{\ex}{\text{xc}}
\nc{\Ex}{\text{Xc}}
\newcommand{\kp}{{\kappa'}}

\newcommand{\lat}{\mathcal{L}}  

\nc{\latd}{\mathcal{L}_\delta}
\nc{\latde}{\mathcal{L}_\delta}

\newcommand{\Lper}{L_{\rm per}^2}
\newcommand{\LperB}{L^2_{{\rm per}, \delta}}

\newcommand{\Omd}{\Om}
\newcommand{\Omde}{\Om_\delta}

\newcommand{\den}{\operatorname{den}}

\newcommand{\ft}{f_T}

\newcommand{\cutOff}{r}
\newcommand{\Br}{B(\cutOff)}

\newcommand{\LinMOp}{M}
\newcommand{\Proj}{P_r}
\newcommand{\hper}{h_{\rm per}}

\newcommand{\hperDelta}{h^\delta_{\rm per}}

\newcommand{\hperkDelta}{h^\delta_{{\rm per}, k}}

\newcommand{\hphi}[1]{h^{#1}}
\newcommand{\Kfib}[1]{K_{ #1}}
\newcommand{\Mfib}[1]{M_{ #1}}
\newcommand{\barKfib}[1]{\bar K_{ #1}}

\newcommand{\rperDelta}{r^\delta_{\rm per}}

\newcommand{\rperk}{r_{{\rm per}, k}}
\newcommand{\rpero}{r_{{\rm per}, 0}}
\newcommand{\rperkDelta}{r^\delta_{{\rm per}, k}}

\newcommand{\rperDeltak}{r_{{\rm per}, \delta k}}

\newcommand{\Xom}{X}

\newcommand{\DETAILS}[1]{}

\begin{document}

\begin{quote}
\qquad \qquad \qquad  \qquad \qquad \qquad   \qquad \qquad \qquad 	{\it J Stat Phys (2020) 180:954-1001 \\ }
\end{quote}

\title[PBfromRHF April 29, 2020]  {On derivation of the Poisson-Boltzmann equation} 
\author{Ilias (Li) Chenn and I.M. Sigal}

\address{IC: Dept of Mathematics, MIT, Cambridge, MA, USA} 

\address{IMS: Dept of Mathematics, Univ of Toronto, Toronto, Canada} 

\date{April 15, 2020} 

\maketitle

\begin{quote}
\qquad \qquad \qquad  	{\it To Joel with friendship and admiration. \\ }
\end{quote}

\begin{abstract} Starting from the microscopic reduced Hartree-Fock equation, we derive the nanoscopic linearized Poisson-Boltzmann equation for the electrostatic potential associated with the electron density.

Mathematics Subject Classification (MSC): 81V74, 81V45, 81Q80, 35Q70

\end{abstract}


\section{Introduction}

\subsection{The reduced Hartree-Fock equation}   
The 
 success of the  Hartree-Fock and density functional theories in revealing the electronic structure of matter warrants 
their use as a starting point in the derivation of emergent macroscopic properties of quantum matter. 

Here, one of the central problems is the derivation of  macroscopic Maxwell's equations 
  in dielectrics. 
\DETAILS{from 
the many-body Quantum Mechanics (Schr\"odinger equation) is one of the key problems in theoretical and mathematical 
 physics.}
The first attack on such a derivation 
 was made in the pioneering works of E. Canc\`es, M. Lewin and G. Stoltz 
 and W. E and J. Lu and their collaborators  
 (\cite{CS, CDL, CDL2, CL, CLS, ELu, ELu2, ELu3, ELu}). 
These works deal with  the reduced Hartree-Fock equation (REHF)\footnote{The REHF obtained from the Hartree-Fock equation (HFE) by omitting the exchange term, see below.} 
  and the Kohn-Sham equation (KSE) of the density functional theory (DFT) at zero temperature.
 The first treatment of the positive temperature REHF was given by Levitt (\cite{Lev}) (see also \cite{CS2}). 

In this paper, we consider the REHF at positive temperature, 
which is also a simplified DFT equation, and derive from it
 the linearized effective Poisson-Boltzmann equation of electrostatics, widely used in  molecular and structural biology (see e.g. \cite{FogBM}). 

For a positive temperature $T$ and with the electron charge set to $e=-1$, REHF can be written in terms of  the one-particle negative charge (or probability)  density $\rho(x)$ of the electron (or generally any Fermi) gas,  as
 \begin{align} \label{DFT-rho-eq} 
	& \rho= \den[\ft(h_{ \rho} - \mu)], 
\end{align}
where $\den: A\ra \rho_A$ is the map  from operators, $A$, to functions 
  $\rho_A(x):=A(x, x)$ (here $A(x, y)$ stands for the integral kernel of an operator $A$), 
  $\ft(\lam)$ is the Fermi-Dirac distribution, 
\begin{align}\label{fT} 
\ft(\lam):=	f_{FD}(\lam/T),\ \quad	f_{FD}(\lam) = \frac{1}{e^\lam+1} \, 
\end{align}
(due to the Fermi-Dirac statistic),
  $\mu$ is the chemical potential 
 and $h_{\rho}$ is a self-adjoint one-particle Hamiltonian 
  depending on the density $\rho$ (self-consistency). Since  $h_{\rho}$ is self-adjoint  the r.h.s. of \eqref{DFT-rho-eq} is well defined. Assuming the electrons are 
  subject to an external potential due to a positive charge distribution $\ka$ (say, due to positive ions), $h_{\rho}$ is given by 
\begin{align}\label{hrho} 
	h_{\rho}:=-\Delta - v*(\ka-\rho) \, ,  
\end{align} 
where 
$v$ is an inter-particle pair potential. It is taken to be the electrostatic potential, as specified below. 

 Let $L^2_{\rm loc}\equiv L^2_{\rm loc}(\R^d)$ denote the space of locally square integrable functions.  We fix  a lattice $\lat\subset \R^d$ and let $L^2_{\rm per}$ be  the space of  $L^2_{\rm loc}$, periodic w.r.t. $\lat$ functions. 
 Finally, let $(L^2_{\rm per})^\perp$ be  the orthogonal complement of the constant functions in $L^2_{\rm per}$. In what follows,  we assume $\rho-\ka\in L^2+(L^2_{\rm per})^\perp$ 
  and $v$ is the electrostatic potential, $v(x)= \frac{1}{4\pi |x|}$ in $3$D, or, generally, 
\[v*f=(-\Delta)^{-1}f,\] 
 for $f\in L^2+(L^2_{\rm per})^\perp$, 
so that $\Delta^{-1}$ is well-defined.\footnote{The decomposition $L^2+L^2_{\rm per}$ is unique: if $f\in L^2+L^2_{\rm per}$, then the periodic part, $f_{\rm per}$, of $f$ is given by the Fourier coefficients \[\hat f_{\rm per}(k):=\lim_{n\ra \infty} \frac{1}{|\Lam_n|} (2\pi)^{-d/2}\int_{\Lam_n}e^{ik\cdot x}f(x) dx,\ k\in \lat^*,\] where  $\Lam_n:=\cup_{\lam \in \lat_n}(\Om+\lam)$, with $\lat_n:=\lat\cap[-n, n]^d$ and $\Om$  an arbitrary fundamental cell of $\lat$, and $\lat^*$ is the reciprocal lattice. Hence $L^2+L^2_{\rm per}$ is a Hilbert space with the inner product which is sum of the inner products in $L^2$ and $L^2_{\rm per}$. The operator $\Delta$ on  $L^2+L^2_{\rm per}$ is self-adjoint on the natural domain 
 (i.e. $H^2+H^2_{\rm per}$) and is invertible on the subspace $L^2+(L^2_{\rm per})^\perp$.} 
  For  $\rho$'s and $\ka$'s specified above, the operator $h_{ \rho}$ is self-adjoint. 

The {\it positive temperature, reduced Hartree-Fock equation} \eqref{DFT-rho-eq} will be abbreviated, with the view to readability, as the {\it TREHF}.

For $T=0$, function \eqref{fT}  becomes the characteristic function of the interval $(-\infty, 0)$ and Eq. \eqref{DFT-rho-eq} becomes just the REHF.

\subsection{Electrostatic potential} \label{sec:electro-pot}
Due to the choice $v*f=(-\Delta)^{-1}f$,  the electrostatic potential  $\phi=v*(\ka-\rho)$ satisfies the Poisson equation 
\begin{align} 
	\label{DFTphi-eq} 
	 - &\Delta \phi= \ka-\rho.
\end{align}
Plugging $\rho$ from  \eqref{DFT-rho-eq} into this equation and  taking $ v*(\ka-\rho)=\phi$ in \eqref{hrho}, we find the equation for $\phi$ 
\begin{align} \label{phi-only-eq}
	-\Delta \phi = \ka - \den[\ft(\hphi{\phi}-\mu)],
\end{align}
where 
\begin{align} \label{h-phi}	\hphi{\phi} = -\Delta - \phi. 
\end{align}
We can recover $\rho$ from $\phi$ via equation \eqref{DFTphi-eq} or the equation
\begin{align} \label{rho-phi-expr}\rho = \den[\ft(\hphi{\phi} - \mu)].\end{align}

Let $H^s$ and $H^s_{\rm per}$ be the Sobolev spaces corresponding to $L^2$ and  $L^2_{\rm per}$.   If $\ka, \rho\in L^2 + L^2_{\rm per}$, $\phi\in H^2 + H^2_{\rm per}$ 
 and $\phi$  and  $\rho-\ka$ satisfy \eqref{DFTphi-eq}, then 
\begin{align}\label{solv-cond}
& \int_\Om \rho_{\rm per}=\int_\Om \ka_{\rm per},
\end{align}
 where $\Om$ is an arbitrary fundamental cell of $\lat$ and the subindex `per' denotes the periodic part of the corresponding function ($\in L^2 + L^2_{\rm per}$). Indeed, let $\Lam_n:=\cup_{\lam \in \lat_n}(\Om+\lam)$, where $\lat_n:=\lat\cap[-n, n]^d$. Integrating \eqref{DFTphi-eq} over the domain $\Lam_n$ 
  and using the Stokes' theorem, 
we find \[\int_{\p\Lam_n}\n\phi= \int_{\Lam_n} (\ka-\rho).\] Since $\lim_{n\ra \infty} \frac{1}{|\Lam_n|}\int_{\p\Lam_n}\n\phi=0$ and $\lim_{n\ra \infty} \frac{1}{|\Lam_n|}\int_{\Lam_n}(\ka-\rho)=  \int_\Om (\rho_\g - \ka)_{\rm per}$, the last relation gives \eqref{solv-cond}.

Eq \eqref{solv-cond} shows that $\rho-\ka\in L^2+(L^2_{\rm per})^\perp$, i.e. it satisfies the conditions mentioned in the paragraph after \eqref{hrho}.

Eq \eqref{solv-cond} determines the chemical potential $\mu$ and expresses the conservation of the charge per fundamental cell of $\lat$.  It is considered as the solvability condition and should be added to \eqref{DFT-rho-eq} 
  in the periodic case.\\

In what follows we 
 associate with a solution $\rho$ of \eqref{DFT-rho-eq} the electrostatic potential 
\begin{align} \label{phi-rho-expr}
	\phi_\rho=  (- \Delta)^{-1} (\ka-\rho) \, , 
\end{align}
and with a solution $\phi$ of Eq. \eqref{phi-only-eq}, the charge density $\rho$ according to \eqref{DFTphi-eq}, or \eqref{rho-phi-expr}.

\subsection{Relation to the TEHF and KSE}  
 The key 
positive temperature  HFE is given by 
\begin{align} \label{HF-gam-eq}
	& \g=\ft(h_{\g} - \mu), 
\end{align}
where $f_{T}(\lam)$ is as above and, for an external charge distribution $\ka$,  
\begin{align} \label{h-gam} 
	h_{\g}:=-\Delta - v*(\ka - \rho_\g)  + ex(\g) \, . 
\end{align} 
Here, recall, $\rho_\g(x):=\g(x, x)$ and $v*f=(-\Delta)^{-1}f$, and 
$ex(\g)$ (the exchange term) is the operator with the integral kernel 
$ex(\g)(x, y):=-v(x-y)\g(x, y)$, where $\g(x, y)$ is the integral kernel of $\g$. 
Observing that $h_{\g}\big|_{ex(\g)=0}=h_{\rho_\g}$, where $h_{\rho}$ is given in  \eqref{hrho},  one  sees that \eqref{HF-gam-eq} with $ex(\g)=0$ implies the equation
\begin{align} \label{DFTgam-eq} 
	& \g=\ft(h_{\rho_\g} - \mu).
\end{align}
Eq. \eqref{DFTgam-eq} is equivalent to Eq.  \eqref{DFT-rho-eq}. Indeed, applying the map $\den$ to equation  \eqref{DFTgam-eq} gives \eqref{DFT-rho-eq}. 
In the opposite direction,  
 if $\rho$ solves \eqref{DFT-rho-eq}, then the density operator 
\begin{align} \label{rho-gam} \g = \ft(h_\rho - \mu),\end{align}
 acting on $L^2(\R^d)$, solves \eqref{DFTgam-eq}.  Thus, \eqref{DFTgam-eq} is the TREHF in terms of the density operator $\g$.


 {By replacing  $ex(\g)$ in 
\eqref{h-gam} 
 by a local exchange-correlation term $\ex(\rho)$ and then applying, as above, the map $\den$ to the resulting equation, one  obtains the natural extension of the original Kohn-Sham equation to positive temperatures:} 
 \begin{align} \label{DFT-rho-eq'} 
	& \rho= \den[\ft(h_{ \rho}^{\rm KS} - \mu)],\\ 
\label{hrho'} 
	& h_{\rho}^{\rm KS}:=-\Delta - v*(\ka-\rho) +\ex(\rho) \, .  \end{align}

\subsection{The origin of the TEHF/TREHF equations} 
As the TEHF and TREHF arise in the same way, in order to avoid repetitions, we consider here only the later. 

Equation \eqref{DFTgam-eq} 
 originates from the static version
\begin{align} \label{DFT-gam-eq'} 
	& [h_{\rho_\g}, \g]=0 
\end{align}
of the time dependent RHF equation (
see e.g. 
\cite{CS1} for a review)
\begin{align}
	\partial_t \g = i[h_{\rho_\g}, \g] \,.  \label{eqn:time-dependent-KS}
\end{align}
Indeed, ignoring symmetries and accidental divergence, $\g$ solves \eqref{DFT-gam-eq'} if and only if  $\g$ solves $\g=f((h_{\rho_\g}-\mu)/T)$ for some reasonable function $f$. (The parameters $T$ and $\mu$ are of no significance at this stage; they are introduced for future reference.) 

The selection of $f$ is done on physics grounds, either bringing the system in question in contact with a thermal reservoir at temperature $T$ and the chemical potential $\mu$, or passing to the thermodynamic limit. This leads to Eq. \eqref{DFTgam-eq}.

As we discuss below, Eq. \eqref{DFTgam-eq} is the Euler-Lagrange equation for the natural free energy.

\medskip

\paragraph{\bf Remark.} 
 If the particles in question were bosons, then $f_{FD}$ would be replaced by the Bose-Einstein distribution
\begin{align}
	f_{BE}(\lam) = \frac{1}{e^\lam - 1} \, .
\end{align}

\subsection{Results} \label{sec:results} 
\DETAILS{Derivation of  macroscopic Maxwell's equations for the electro-magnetic fields in dielectrics from 
the many-body Quantum Mechanics (Schr\"odinger equation) is one of the key problems in theoretical and mathematical 
 physics. In the full generality, this problem is far from our reach. 
 Thence there is a great interest in derivation of these equations 
  from reliable microscopic models.
 
\DETAILS{Derivation of macroscopic theories from microscopic ones is one of the main challenges of theoretical and mathematical physics, with 
 very few  very limited successes. 

Thence are 
{Problem}: Derive macro 
theories from reliable micro 
 models }


With success in describing the electronic structure of matter, the DFT and HFT ares presently the {\it leading such microscopic theories}. Hence it is an interesting proposition to derive the macroscopic Maxwell's equations  from the HFT/DFT microscopic theory.}

We are interested in the dielectric response in a medium subjected to a local deformation of the crystalline structure.  
To formulate our results we introduce some notation and definition.

 In what follows, we assume that $d=3$ and let $\lat $ be a (crystalline) Bravais lattice in $\R^3$. 
 We also define the Hilbert space of $\lat$-periodic functions
\begin{align} \label{L2-per} 
	L^2_{\rm per}\equiv L^2_{\rm per}(\R^3) = \{ f \in L^2_{\rm loc}(\R^3) : \text{$f$ is $\lat$-periodic } \},
\end{align}
 with the inner product $\lan f, g\ran = \int_\Om \bar f g$ and the norm $\|f\|_{L^2_{\rm per}}^2 = \int_\Om |f|^2$ for some arbitrary fundamental domain $\Om$ of $\lat$. We denote by $H^s_{\rm per}\equiv H^s_{\rm per}(\R^3)$ and $\|\cdot\|_{H^s_{\rm per}}$ the associated Sobolev spaces and their norms, while the standard Sobolev spaces and their norms are denoted by $H^s\equiv H^s(\R^3)$ and $\|\cdot\|_{H^s}$.
 

\medskip
\paragraph{\textbf{Crystals.}} 
We consider 
 a  background	charge distribution, $\ka(y)\equiv \ka_{\rm per}(y)$, periodic with respect to the lattice $\cL$ (crystal). Here  $y$ stands for the microscopic coordinate.

We think of  $\lat$ and $\ka_{\rm per}$ as a crystal lattice and  the ionic charge distribution of $\lat$. An example of $\ka_{\rm per}$ is
\begin{align}
	\ka_{\rm per}(y) = \sum_{l \in \lat} \ka_a(y-l) \, ,
\end{align}
where $\ka_a$ denotes an ionic (``atomic'') charge distribution.   

\medskip

\paragraph{\bf Dielectrics.} Next, we describe a model of the (crystalline) dielectric. 
\begin{definition}\label{def:diel}	
We say that 
an $\lat$-periodic background charge  density $\ka_{\rm per} \in L^2_{\rm per}$ is 
 dielectric, 
 if TREHF \eqref{DFT-rho-eq}, with $\ka=\ka_{\rm per}$, 
  has an $\lat$-periodic solution $(\rho_{\rm per}, \mu_{\rm per})$, with the following properties:

(a) the periodic one-particle Schr\"odinger operator 
\begin{align}\label{h-per-ex}	
& h_{\rm per}:=\hphi{\phi_{\rm per}} =-\Delta - \phi_{\rm per}, 
\text{ with}\\
\label{phi-per}	& \phi_{\rm per}:= 4\pi (-\Delta)^{-1}(\ka_{\rm per}-\rho_{\rm per}), 
\end{align}
 acting on $L^2 \equiv L^2(\R^3)$ is self-adjoint and 
has a gap in its spectrum; 

(b) $\mu_{\rm per}$ is in this gap; 

 (c)  $\phi_{\rm per} \in H^2_{\rm per} $ and $ \|\phi_{\rm per}\|_{H^2_{\rm per}} + 
 |\mu_{\rm per}|\le \Lambda_{\rm per}$, for some constant $\Lambda_{\rm per}$ independent of $T$. 
\end{definition}

An existence result for the 
dielectrics is discussed in Remarks 6 and 7 after the next theorem. In particular, Proposition \ref{prop:diel} shows that the set of dielectric charge  densities $\ka_{\rm per}$ is robust. Moreover, \eqref{phi-only-eq} can be reformulated so that only $\phi_{\rm per}$ and $\mu_{\rm per}$, but not $\ka_{\rm per}$, enter it explicitly, see \eqref{psi-eq}. So these are the only inputs of our analysis. 

\medskip

\textbf{Dielectric response.} 
We consider a macroscopically deformed microscopic crystal charge distribution, 
\begin{align}\label{m-del}
	\ka_\delta(y) = \ka_{\rm per} (y) +  \delta^{3}\ka'(\del y), 
\end{align}
where $\del$ is a small parameter which stands for the ratio of microscopic and macroscopic scale and $\ka'(x) \in L^2$ is a small local perturbation living on the macroscopic scale. By $y$ and $x=\del y$, we denote the microscopic and macroscopic coordinates, respectively. Thus,  the microscopic scale is $y\sim 1$ and $x\sim \del$ and the macroscopic one, $y\sim 1/\del$ and $x\sim 1$.

We formulate the conditions for our main result.
  We introduce the homogeneous Sobolev spaces
\begin{align}\label{dotHs}
	\dot{H}^s\equiv \dot{H}^s(\R^3) = \left\{ f \text{ measurable on $\R^3$} : \|f\|_{\dot{H}^s} < \infty \right\}
\end{align}
for $s\ge 0$ with the associated norm 
\begin{align}
	\|f\|_{\dot{H}^s}^2 = \int |(-\Delta)^{s/2}  f(k)|^2 \,. \label{eqn:dot-H-s-norm-def}
\end{align}
\DETAILS{\begin{enumerate}[label={[A\arabic*]}]	
	\item\label{A:kappa} (Periodicity) 
	$\ka_{\rm per}\ \text{ is $\lat$-periodic and } 	\ka_{\rm per} \in H^2_{\rm per}(\R^3). 
	$ 
\end{enumerate}}

\begin{enumerate}[label={[A\arabic*]}]
	\item\label{A:diel} 
	 (Dielectricity) 
$\ka_{\rm per}$	  is dielectric.\end{enumerate}

\DETAILS{For the operator  $h_{\rm per}$ acting on $L^2(\R^3)$ and defined in \eqref{h-per-ex}, we let
\begin{align} \label{eta-R3}	&\eta 
 := \text{dist}(\mu_{\rm per}, \sigma(h_{\rm per})). 
\end{align}}

Let $h_{\rm per}$ and $h_{\rm per, 0}$ denote operators given by expression \eqref{h-per-ex} acting on $L^2(\R^3)$ and  $L^2_{\rm per}(\R^3)$, respectively. These operators are self-adjoint and the latter has a purely discrete spectrum.   By  Assumptions \ref{A:diel}, $\mu_{\rm per}$ is in a gap of $h_{\rm per}$. For notational convenience, we rescale our problem so that
\begin{align} \label{eta-R3}	&
\eta := \text{dist}(\mu_{\rm per}, \sigma(h_{\rm per})) =1. 
\end{align}
   It follows the Bloch-Floquet decomposition results in Section \ref{sec:BF-decomp} below that the gaps of $h_{\rm per}$ are contained 
in the resolvent set of $h_{\rm per, 0}$, so that
\begin{align}   \label{eta-Om}	&\etal :=   \text{dist}(\mu_{\rm per}, \sigma(h_{{\rm per}, 0}))\ge 1.
\end{align}

\DETAILS{
\begin{enumerate}[label={[A\arabic*]}]
\setcounter{enumi}{1}
	\item\label{A:spec-gap-ineq} (Spectral gap inequality) 
		\begin{align} 
	\eta_0 \le 3\eta =3.		
			\label{eta'-gap-def} 
		\end{align}
\end{enumerate} }

\begin{enumerate}[label={[A\arabic*]}]
	\setcounter{enumi}{1}	
	\item\label{A:kappa-reg} (Perturbation $\kp$) 
	\begin{align*}
		\kp \in H^1 \cap \dot H^{-1} . 
	\end{align*}
\end{enumerate}

In what follows, the inequalities $A\ls B$ and $A\gs B$ mean that there are constants $C$ and $c$ independent of $T$ and $\del$, s.t. $A\le C B$ and $A\ge c B$ and similarly for $A\ll B$ and $A\gg B$.
\DETAILS{Denote $\cbeta:= \beta  e^{-\eta(\Om)\beta}$, where  $\beta = 1/T$ is the inverse temperature. 
By  Assumption \ref{A:scaling}, 
  $\cbeta\ll 1$.
We require also that
\begin{align}  \label{beta-del-cond} \delta^{9/8}\ll  \cbeta \ll   \delta^{9/11}. \end{align}}
  Our main result is 
	 	  \begin{theorem} \label{thm:PE-macro} 
  Let Assumptions \ref{A:diel} - \ref{A:kappa-reg} hold  and let 
  $(\phi_{\rm per}, \mu_{\rm per})$ be the  electrostatic and chemical potentials 
 associated with $\ka_{\rm per}$ (entering \ref{A:diel}) as per 
  Definition \ref{def:diel}. 
 \DETAILS{ {\bf solving TREHF \eqref{phi-only-eq}, with $\ka =\ka_{\rm per}$ from  \eqref{m-del}}.} There is $\al$ $=\al(\Lambda_{\rm per})>0$ 
  sufficiently small, 
   s.t., if 
\begin{enumerate}[label={[A\arabic*]}]
\setcounter{enumi}{2}
	\item\label{A:scaling} (Regime) 
	\DETAILS{{\bf(Replace  Assumption \ref{A:scaling} by a condition at the beginning of Theorem \ref{thm:PE-macro}? add \eqref{beta-del-cond}?)} 
		\begin{align*}
			\delta \ll 1\ \text{ and }\ |T \ln T| \ll \eta(\R^3)\ \quad  (T=1/\beta). 
		\end{align*}}
The parameters $T > 0$ and $\del>0$ satisfy 
\begin{align}  \label{beta-del-cond} 
c_{T}:=T^{-1}  e^{-\etal/T} \le \al,\ \quad   
  c_{T}^{- 8/9} \delta  \le  \al,  \end{align}  
\end{enumerate}
  then    the following statements are true 
 \begin{enumerate}
	\item  Electrostatic TREHF \eqref{phi-only-eq}, with $\ka=\ka_\delta$ given in 
	 \eqref{m-del} and $\mu = \mu_{\rm per}$, has a unique solution  $\phi_\delta \in H^2_{\rm per} + H^1$;  
	\item The potential $\phi_\delta(y)$  
	 is of the form 
\begin{align} \label{phi-del-exp}
\phi_\delta(y) = \phi_{\rm per}(y) + \delta\psi (\delta y) + \vphi_{\rm rem}(\delta y), 
\end{align}
where  
$\vphi_{{\rm rem}}(x) \in H^{1}$ and obeys the estimates (with $\dot{H}^0=L^2$) 
\begin{align}\label{phidel-rem-est} \|\vphi_{{\rm rem}}\|_{\dot{H}^i}\ll  
\al^{\frac14-\frac12 i} (c_T^{-1/2} \delta)^{2-i},
\end{align}
 with $\al$ given in \eqref{beta-del-cond}, and $\psi (x) \in H^1$ and satisfies the equation 
\begin{align}
	(\nu-\n\cdot \epsilon \nabla) \psi = \ka', \label{eqn:macro-eff-eqn}
\end{align}
with a positive number  $\nu>0$ and
 a constant real, symmetric $3 \times 3$ matrix, $\epsilon \geq 1- O(\cbet^2)$; 
 \item   $\epsilon\equiv \epsilon(T)$   and  $\nu\equiv \nu(T, \del)$ are  given explicitly by \eqref{eps} - \eqref{eps''} and \eqref{nu} - \eqref{m}, below. Moreover, $\nu\ra 0$, as either $T\ra 0$ or $T\ra \infty$.
   \end{enumerate}
\end{theorem}
We discuss Assumptions \ref{A:diel} and \ref{A:scaling} in Remarks 6 and 10 and the statements of the theorem, in Remarks 1-4, below. 
\subsection{Discussion}
 1)   Theorem \ref{thm:PE-macro}(1) and equation 
 \eqref{phi-rho-expr} connecting the charge density $\rho$ with $\phi$ imply that RHF equation \eqref{DFT-rho-eq}, with 
	 \eqref{m-del} and $\mu = \mu_{\rm per}$, has a unique solution  $\rho_\delta \in L^2_{\rm per} + \dot H^{-1}$. 

\medskip

2) The quantity $\nu\equiv \nu(T, \del)$ is defined as
 \begin{align}\label{nu}	& \nu=\delta^{- 2}|\Om|^{-1}(m+  O(\cbet^2)),\\
 \label{m}	& 
 m = -\tr_\Om\left[ f_{T}'(h_{{\rm per}, 0}-\mu) \right]>0, 
\end{align}   
where $h_{{\rm per}, 0}$ denotes the restriction of the self-adjoint operator $h_{\rm per}:= h^{\phi_{\rm per}}=  -\Delta +\phi_{\rm per}$ to $L^2_{\rm per}$.  
 Lemmas \ref{lem:V1-lower-bound} and \ref{lem:Vp-uppper-bound} of Appendix \ref{app:bounds-on-V} imply the estimates
\begin{align} \label{m-est}	&0<\cbet \ls m  \ls  \cbet. 
\end{align}
By \eqref{m-est} and \eqref{beta-del-cond}, $m$ is the leading term in \eqref{nu} and $\nu \gg  \delta^{-7/8}$.   
 A careful examination of the formulae above shows that $\nu\ra 0$, as either $T\ra 0$ or $T\ra \infty$.

3) 
The $3 \times 3$ matrix, $\epsilon$, in  \eqref{eqn:macro-eff-eqn} is given explicitly by
\begin{align}
\label{eps}	&\epsilon := \one +\epsilon' - \epsilon'',\\ 
 \label{eps'}	&\epsilon' = -\frac{1}{|\Om|}\Tr_{L^2_{\rm per}} \oint  \rpero^2(z) (-i\nabla) \rpero(z) 
		 (-i\nabla) \rpero(z) , \\
\label{eps''}	&\epsilon'' = \frac{1}{|\Om|} \lan \rho',\barKfib{0}^{-1} \rho'\ran_{L^2_{\rm per}} \, , 
\end{align} 
 where 
 $	\oint := \frac{1}{2\pi i} \int_\Gamma dz \ft(z-\mu), $ 
with $\Gamma$ the contour given in Figure \ref{fig:cauchy-int-contour} below,  with the positive orientation, $\rpero(z):=(z-h_{{\rm per}, 0})^{-1}$, 
  $\barKfib{0}$ is the operator defined in \eqref{eqn:bar-K-0-def}, 
and  
\begin{align} \label{rho'}
	\rho' =  2 \den \oint \rpero^2(z)(-i\n)\rpero(z) \, . 
\end{align}

4) Eqs. \eqref{eqn:macro-eff-eqn}, \eqref{nu} and \eqref{m-est}  imply that 
\[\|\psi\|_{\dot{H}^{i}} =O([\delta|\Om|^{1/2} m^{-1/2}]^{2-i}),\ i=0, 1,\] and therefore, by  \eqref{beta-del-cond}, \eqref{phidel-rem-est} and \eqref{m-est}, we have
\[\|\vphi_{{\rm rem}}\|_{L^2}\ll 
\al^{1/4}(m^{-1/2} \delta)^2\ll \|\psi\|_{L^2}.\] Hence $\psi$ is a subleading term in \eqref{phi-del-exp} in the $L^2$-norm. 

\medskip

5)  
  \eqref{eqn:macro-eff-eqn} is 
 the linearized  {\it 
  Poisson-Boltzmann equation} used extensively in physical chemistry and molecular biology (see e.g. \cite{FogBM}). 
$\epsilon$ is an effective {\it permittivity matrix} and $\sqrt\nu$ and $1/\sqrt\nu$ are  the {\it  Debye-H\"uckel parameter} and the {\it  Debye length}, respectively.

 The screening term $\nu$ in \eqref{eqn:macro-eff-eqn} is due to the electrons at the tail of the {\it Fermi-Dirac distribution} being in the {\it conduction} band. (In the macroscopic regime, the  Fermi-Dirac distribution becomes the (Maxwell-) Boltzmann distribution.)

\medskip


6)  (Existence of crystalline dielectrics)
We say that the potential $\phi$ is gapped if 
the Schr\"odinger operator $-\Delta - \phi$  has a gap in its continuous spectrum.  
\begin{proposition} \label{prop:diel}
For any $\lat$-periodic, gapped potential $\phi_{\rm per} \in H^k_{\rm per}$, $k\ge 2,$ and any real number $\mu_{\rm per}$ in a gap of $h_{\rm per} := -\Delta - \phi_{\rm per}$, there is  
  $\ka_{\rm per} \in H^{k-2}_{\rm per}$ such that Eq. \eqref{DFT-rho-eq}, with $\ka=\ka_{\rm per}$, 
  has the solution $(\rho=\rho_{\rm per} \in H^{k-2}_{\rm per}, \mu=\mu_{\rm per})$ 
   with the associated (according to \eqref{DFTphi-eq}) electrostatic potential 
exactly  $\phi_{\rm per}$. Moreover, the pair $(\phi_{\rm per}, \mu_{\rm per})$ satisfies the electrostatic Eq \eqref{phi-only-eq} with this $\ka_{\rm per}$. 
\end{proposition}

\begin{proof} 
Let $\phi_{\rm per}$ be such that $h_{\rm per} := -\Delta - \phi_{\rm per}$ has a gap. We choose $\mu_{\rm per}$ to be in this gap and define (see \eqref{h-phi}-\eqref{rho-phi-expr})
\begin{align} \label{eqn:rho-def-by-V}
	\rho_{\rm per} := \den[\ft(h_{\rm per} - \mu_{\rm per})].
\end{align}
Next, we define 
\begin{align}\label{phi-ka-def-by-V}	
	\ka_{\rm per} := -\Delta \phi_{\rm per}  + \rho_{\rm per}. 
\end{align}
Then, it is straightforward to check that $(\rho_{\rm per}, \mu_{\rm per})$ is a solution of Eq.  \eqref{DFT-rho-eq} 
 with background potential $\ka_{\rm per}$. By construction, $h_{\rm per}$ has a gap and $\mu_{\rm per}$ is in this gap.
 \end{proof}

One can extend Proposition \ref{prop:diel} to construct  pairs $(h_{\rm per} = -\Delta - \phi_{\rm per}, \mu_{\rm per})$ having any desired property {\bf P}. 
Following Proposition \ref{prop:diel}, we construct $\rho_{\rm per}$, $\phi_{\rm per}$, and $\ka_{\rm per}$ via \eqref{eqn:rho-def-by-V} and \eqref{phi-ka-def-by-V} in this order. Then $(\rho_{\rm per}, \mu_{\rm per})$ is a solution of Eq. \eqref{DFT-rho-eq} with background potential $\ka_{\rm per}$. By construction, $h_{\rm per}$ has property {\bf P}.

 The proposition above shows that for any positive $\eta$ and $T$, we can find  
  $\ka_{\rm per} \in H^{k-2}_{\rm per}$ such that the solution of Eq. \eqref{DFT-rho-eq} with $\ka=\ka_{\rm per}$ and $T$ gives the gap $\eta$.

\medskip

7)  (General dielectrics) 
We say that a background charge  density $\ka$ is dielectric  if 
 Eq. \eqref{DFT-rho-eq} with background charge distribution $\ka$ has a solution $(\rho, \mu)$, with $\rho$ in an appropriate space, say, $H^2_{\rm loc}\cap L^\infty$, 
and having the following properties:

(a) the one-particle Schr\"odinger operator, defined for this solution, 
\begin{align}
	& h^\phi:=-\Delta - \phi, 
	\text{ with }\ 
	\phi:=  (-\Delta)^{-1}(\ka-\rho), 
\end{align}
acting on $L^2 $, is self-adjoint and 
has a gap in its spectrum; 

(b) $\mu$ is in this gap.

By the remark at the end of the previous item we have
\begin{proposition}[Existence of general dielectrics] \label{prop:diel-gen}
For any 
gapped potential $\phi \in H^2_{\rm loc}\cap L^\infty$ and any  number $\mu$ in a gap of $h^{\phi} := -\Delta - \phi $, there is   $\ka \in L^2_{\rm loc}\cap L^\infty$ 
s.t. Eq. \eqref{DFT-rho-eq}, with these $\ka \in L^2_{\rm loc}\cap L^\infty$ and $\mu$, has the solution $\rho$, whose 
 the electrostatic potential  (according to \eqref{DFTphi-eq}) is $\phi$. 
\end{proposition}

\medskip

8) (Existence of ideal crystals)
The existence of periodic solutions  to Eq. \eqref{DFT-rho-eq} 
(equilibrium crystalline structures exists at $T>0$) is shown in the following:

\begin{theorem} 
\label{thm:ideal-cryst-exist} 
Let $d=3$ and $\ka_{\rm per} \in H^2_{\rm per}$. 
 Then Eq. \eqref{DFT-rho-eq},  with the $\cL-$periodic background charge density $\ka=\ka_{\rm per}$ 
  has a solution $(\rho_{\rm per}, \mu_{\rm per})$,  with $\rho_{\rm per}$ periodic and satisfying $\sqrt{\rho_{\rm per}} \in H^1_{\rm per}$.
\end{theorem}
We give references to the proof of this theorem below.

\medskip

9)  
In the limit $T \rightarrow 0$, our expression for the dielectric constant $\epsilon$ agrees with \cite{CL} (see Appendix \ref{sec:epsilon-T-to-zero} below).

\medskip

10)  (Physical dimensions) 
 The physical cell size of common crystals is on the order of $10^{-10}$ (\cite{SB}). This gives  $\del\sim 10^{-10}$. The gap size, $\etal$, is on the order $1 eV$ \cite{SB}. Since the Boltzmann constant, $k_B$, is of the order $10^{-4} eV/K$, this gives $\etal/k_B\sim 10^{4} K$.  Thus, though we do not compute actual constants in our estimate, we expect that the allowed values of $\del$ and $T$ are within physically interesting ranges. 

\medskip

11) (Energy)
The evolution \eqref{eqn:time-dependent-KS} conserves the number of particles $N_{\Xom}(\g):=\Tr_\Xom (\gamma)$  and the energy  
 \begin{align} 
	E_{\Xom}(\g) &:=\Tr_\Xom \big((-\Delta) \gamma \big) + \frac12 \int_\Xom \s_\g v*\s_\g, 		\label{energy} 
\end{align}
 where $\Xom$ is either $\R^d$ or a fundamental cell $\Om$ of $\lat$, with $\Tr_\Xom$ defined accordingly, and $\s_\g:=\ka - \rho_\gamma$.

Eq. \eqref{DFTgam-eq} is 
the Euler-Lagrange equation  for the free energy functional 
\begin{align}\label{FT-def} 
	F_{T}(\g):=E_{\Xom}(\g) -T S_{\Xom}(\g)-\mu N_{\Xom}(\g),
	\end{align}
where $S_{\Xom}(\g) = -\Tr_{\Om} (\g \ln \g+(\one-\g) \ln (\one-\g))$ is the entropy.

To obtain the HF (free) energy functional, one should add to \eqref{energy} (\eqref{FT-def}) the HF exchange energy term $Ex(\g):=\frac12 \int_\Xom\int_\Xom |v(x-y)\g(x, y)|^2$.

\medskip


\bigskip

\paragraph{\bf Literature.}
The relation of the HF theory to the exact quantum many-body problem was established rigorously in \cite{LiebSim}.   

For $T=0$, the existence theory for the RHFE and HFE was developed in \cite{LiebSim, Lions1,
CLeBL2, AnCan, HLS}, see \cite{Lions, Lieb2, LeBL}, for reviews. For the Hartree-Fock equation \eqref{HF-gam-eq} with periodic $\ka = \ka_{\rm per}$, the existence of periodic solutions from certain trace classes was obtained in \cite{CLeBL} and \cite{CLeBL2}.

Results for $T=0$, similar and related to Theorem \ref{thm:PE-macro}, were proven in \cite{CS, CDL, CDL2, CL, CLS, ELu2, ELu3, ELu}.

For the case where $T > 0$, F. Nier \cite{Nier} 
 proved the existence and uniqueness of the TRHF \eqref{DFT-rho-eq} 
  via variational techniques. Later, Prodan and Nordlander \cite{PrNord} provided another existence and uniqueness result with the exchange-correlation term in the case where $\ka = \ka_{\rm per}$ is small. In this case, the associated potential term $\phi_{\rm per} + \ex(\rho_{\rm per})$, where $\ex(\rho)$ is a local exchange-correlation term, see \eqref{hrho'}, is small as well. (As was pointed by A. Levitt, a result for small $\ka = \ka_{\rm per}$ would not work in Theorem \ref{thm:PE-macro} above as Assumption \ref{A:diel} fails for it.) 

The results given in Theorem \ref{thm:ideal-cryst-exist} is taken from \cite{CS2}. Papers \cite{AnCan, CLeBL2, CLeBL, CS2} use variational techniques and did not provide uniqueness results.
A. Levitt \cite{Lev} proved 
the screening of small defects for the TRHFE.

\bigskip

\paragraph{\bf Approach.}  As in \cite{Lev}, our starting point is Eq. \eqref{phi-only-eq} for the electrostatic potential $\phi$. We also use some important ideas from  \cite{CLS}. However, our approach to proving Theorem \ref{thm:PE-macro} is fairly novel. Rather that employing variations-based techniques, we use the Lyapunov-Schmidt reduction, which also allows us to estimate the remainders.
 
 The starting equation of our analysis can be formulated as follows.  Let $(\phi_{\rm per} (x), \mu_{\rm per})$ be the solution of \eqref{phi-only-eq}, with $\ka (x)=\ka_{\rm per} (x)$, and let  $\ka_\delta$ be given in  \eqref{m-del}. 
Define $\psi$ by the equality 
\[
	 \phi = \phi_{\rm per} +  \psi.\] 
Plugging this decomposition into \eqref{phi-only-eq}, with  $\ka=
\ka_\delta$ 
 and $\mu=\mu_{\rm per}$, and using that $h^{\phi}=h^{\phi_{\rm per}}-\psi$, we arrive at the equation for $\psi$: 
\begin{align} \label{psi-eq}
	-\Delta \psi =  \kp_\del - \den[g_{\phi_{\rm per}'}(\psi)],
\end{align}
where $\phi_{\rm per}':=\phi_{\rm per}+\mu_{\rm per}$, $g_{\phi_{\rm per}'}(\psi):=\ft(h^{\phi_{\rm per}'}-\psi)-\ft(h^{\phi_{\rm per}'})$ 
 and $\kp_\delta(y):=\delta^{3}\ka'(\del y)$. 
 
 This is a nonlinear and nonlocal Poisson equation for $\psi$. We see that only $\phi_{\rm per}':=\phi_{\rm per}+\mu_{\rm per}$, but not $\ka_{\rm per}$, enters Eq. \eqref{psi-eq} explicitly.

Though we deal with the simplest microscopic model - the reduced HF equation - our techniques are fairly robust and would work for 
 the full-fledged DFT. Also, we favoured rough estimates to more precise but lengthier ones which produce better bounds on $\beta$ in  \eqref{beta-del-cond}, see 
  Appendix \ref{sec:nonlin-est-improv} below.



\medskip

\medskip

The paper is organized as follows. After presenting preliminary material on charge density estimates and the Bloch-Floquet decomposition in Section \ref{sec:den-BF-deco}, we prove  Theorem \ref{thm:PE-macro} in Sections  \ref{sec:pf:thm:PE-macro} - \ref{sec:nonlin-est-rough}. 
  Section \ref{sec:pf:thm:PE-macro} 
   contains  the main steps of the proof of Theorem \ref{thm:PE-macro}. 
    Section \ref{sec:nonlin-est-rough} covers fairly straightforward 
 technical estimates of the nonlinearity.


\subsection*{Acknowledgements}
The authors thank   Rupert Frank, J\"urg Fr\"ohlich, Gian Michele Graf, Christian Hainzl, Jianfeng Lu and Heinz Siedentop for stimulating discussions, and to anonymous referee, for many pertinent remarks. 
The  correspondence with Antoine Levitt played a crucial role in steering the research at an important junction. The second author is also grateful to  Volker Bach, S\'ebastien  Breteaux, Thomas Chen and J\"urg Fr\"ohlich for enjoyable collaboration on related topics.

 The research on this paper is supported in part by NSERC Grant No. NA7901. The first author is also in part supported by NSERC CGS D graduate scholarship.

\section{Densities and Bloch-Floquet decomposition} \label{sec:den-BF-deco}

\subsection{Locally trace class operators} \label{sec:loc-tr-class}
Let $C_c\equiv C_c(\R^3)$ denote the space of compactly supported continuous functions on $\R^3$. An operator $A$ on $L^2$ is said to be locally trace class if  $fA$  and $Af$ are trace class for all $f \in C_c$. (For the proofs below, it suffices to require that $fA$ is trace class.) 

Let $\lat$ be a Bravais lattice on $\R^3$ and $\Om$ a fundamental domain of $\lat$ as in Section \ref{sec:results}. Denote $|S|$ to be the volume of a measurable set $S \subset \R^3$ and note that $|\Om|$ is independent of the choice of the fundamental cell $\Om$. Let $T_s$ be the translation operator 
\begin{align}
	T_s : f(x) \mapsto f(x-s). \label{eqn:translation-op-def}
\end{align}
We say that a function $f: \R^3\ra \C$  is $\lat$-periodic if and only if it is invariant under the translations action of $T_s$ for all lattice elements $s \in \lat$. We define the space
\begin{align} \label{eqn:Lp-per}
	L^p_{\rm per}\equiv L^p_{\rm per}(\R^3) = \{ f \in L^p_{\rm loc}(\R^3) :  
f \text{ is $\lat$-periodic}	 \},
\end{align}
with the norm of $L^p(\Om)$ for some $\Om$. The norms for $L^p_{\rm per}$ and $L^p\equiv L^p(\R^3)$ are distinguished by the subindices $L^p_{\rm per}$ and $L^p$. 

We say that a bounded operator $A$ on $L^2$ is $\lat$-periodic if and only if $[A, T_s] = 0$ for all $s \in \lat$ where $T_s$ is the translation operator defined in \eqref{eqn:translation-op-def}. 

Let $S^p$ be the standard $p$-Schatten space of bounded operators on $L^2$ with the $p$-Schatten norm
\begin{align} \label{Ip-norm'}
	\|A\|_{S^p}^p 	:=& \Tr_{L^2}( (A^*A)^{p/2} ).
\end{align} 
Next, let $\chi_Q$ denote the characteristic function of a set $Q \subset \R^3$ and let $S^p_{\rm per}$ be the space of bounded, $\lat$-periodic operators $A$ on $L^2$ with $\|A\|_{S^p_{\rm per}} < \infty$ where
\begin{align} 
	\|A\|_{S^p_{\rm per}}^p	:=& \Tr_{\Om}( (A^*A)^{p/2} ) := \frac{1}{|\Om|} \Tr_{L^2}( \chi_\Om (A^*A)^{p/2} \chi_\Om).  \label{eqn:I-p-per-def}
\end{align} 

We remark that the $S_{\rm per}^2$ norm does not depend on the choice of $\Om$ since $A$ is $\lat$-periodic. We have the following estimates for the densities in terms of Schatten norms.

\subsection{Densities} \label{sec:densities}
For a locally trace class operator $A$, we define its density $\den[A]$ to be a regular countably additive complex Borel measure satisfying 
\begin{align} 
 \label{den-def1}	& \int \den(A)f= \Tr( fA), 
\end{align}
for every $f \in C_c$. If $\Tr(fA)$ is continuous in $f$ in the $C_c$-topology, then the Riesz representation theorem shows that \eqref{den-def1}, for every $f \in C_c$, 
  define $\den[A]$ uniquely. In our case, we will frequently stipulate stronger regularity assumptions on $A$, 
  implying that $\den[A]$ is actually in a reasonable function space. (e.g. Lemma \ref{lem:rhoA-by-tracial-norm} below).
  
If an operator $A$ has an (distributional) integral kernel, $A(x, y)$, with the diagonal, $A(x, x)$, being  a regular countably additive complex Borel measure, then   
\begin{align}  \label{den-prop1}
	&  \den(A)(x) = A(x, x).
\end{align}
Finally, den is a linear map on  locally trace class operators with the property that for any $f \in C_c$,
\begin{align}  \label{den-prop1}
	&   \den( fA) = f\den(A).
\end{align}

\begin{lemma} \label{lem:rhoA-by-tracial-norm}
Let $A$ be a locally trace class operator on $L^2$ and $\epsilon > 0$. We have the following statements.
\begin{enumerate}
	\item If $(1-\Delta)^{3/4 + \epsilon} A \in S^2$, resp. $S^2_{\rm per}$, then $\den[A] \in L^2 $, resp. $L^2_{\rm per}$. Moreover, respectively,
		\begin{align}
	& \|\den[A]\|_{L^2} \ls \|(1-\Delta)^{3/4 + \epsilon} A\|_{S^2}, \label{rhoA-by-tracial-norm1} \\
			&\|\den[A]\|_{L^2_{\rm per}} \ls |\Om|^{1/2} \|(1-\Delta)^{3/4 + \epsilon} A\|_{S^2_{\rm per}} \label{rhoA-by-tracial-norm2} 
		\end{align}
	\item If $(1-\Delta)^{1/4 + \epsilon} A \in S^{6/5}$, then $\den[A] \in \dot{H}^{-1} $ (where $\dot{H}^s $ is defined in \eqref{dotHs}). Moreover,
		\begin{align}
			\|\den[A]\|_{\dot{H}^{-1}} \ls \|(1-\Delta)^{1/4 + \epsilon} A\|_{S^{6/5}}  \label{rhoA-by-tracial-norm3}	
		\end{align}
\end{enumerate}
\end{lemma}
\begin{proof}
We prove \eqref{rhoA-by-tracial-norm2} and \eqref{rhoA-by-tracial-norm3} only; \eqref{rhoA-by-tracial-norm1} is similar and easier. We begin with \eqref{rhoA-by-tracial-norm2}. Since the operator $(1-\Delta)^{1/4+\epsilon} A$ is $\lat$-periodic, its density, if it exists, is also $\lat$-periodic. By the $L^2_{\rm per}$-$L^2_{\rm per}$ duality, relation \eqref{den-def1},  $\den[A] \in L^2_{\rm per}$ and \eqref{rhoA-by-tracial-norm2} holds if and only if 
\begin{align}
	|\Tr_{\Om}(f A)| \ls  |\Om|^{1/2} \| f\|_{L^2_{\rm per}}\|(1-\Delta)^{3/4 + \epsilon} A\|_{S^{\rm per}_2} \label{eqn:TrAf-need-to-est-per}
\end{align}
for all $f \in L^2(\R^3)$ with support in $\Om$, where we recall $\| f\|_{L^2_{\rm per}} = \|f\chi_\Om\|_{L^2}$. Since the support of $f$ is in $\Om$, by the H\"{o}lder's inequality for the trace-per-volume norm,
\begin{align}
	\frac{1}{|\Om|}|\Tr_{\Om}&(f A)| = \frac{1}{|\Om|}|\Tr( \chi_\Om fA \chi_\Om)| \\
		&\ls \|A(1-\Delta)^{3/4+\epsilon}\|_{S^2_{\rm per}} \|(1-\Delta)^{-3/4-\epsilon}f\|_{S^2_{\rm per}}.
\end{align}
By the Kato-Seiler-Simon inequality
\begin{align}
	\|f(x)g(-i\n) \|_{S^p} \ls \|f\|_{L^p}\|g\|_{L^p} \label{eqn:KSS}
\end{align}
for $2 \leq p < \infty$ (see \cite{SimonTI}; one can also replace $S^p$ and $L^p$ by their periodic versions $S_{\rm per}^p$ and $L^p_{\rm per}$, respectively.), we obtain \eqref{eqn:TrAf-need-to-est-per}. Thus, \eqref{rhoA-by-tracial-norm2} is proved.

Now we prove \eqref{rhoA-by-tracial-norm3} as above. By the $\dot{H}^1$-$\dot{H}^{-1}$ duality, it suffices to show that
\begin{align}
	|\Tr (f A)| \ls  \| f\|_{\dot{H}^1}\|(1-\Delta)^{1/4 + \epsilon} A\|_{S^{6/5}} \label{eqn:TrAf-need-to-est}
\end{align}
for all $f \in \dot{H}^1 \cap C_c $ and for $\epsilon > 0$. So, we estimate $|\Tr (f A)|$. By the non-abelian H\"{o}lder inequality with $1 = \frac{1}{6} + \frac{1}{6/5}$ (\cite{SimonTI}), 
\begin{align}
	|\Tr (f A)| \ls  \|f(1-\Delta)^{-1/4 - \epsilon}\|_{S^6} \|(1-\Delta)^{1/4 + \epsilon} A\|_{S^{6/5}} .
\end{align}
The Kato-Seiler-Simon inequality \eqref{eqn:KSS} shows
\begin{align}
	|\Tr (f A)| \ls  \|f\|_{L^6} \|(1-\Delta)^{1/4 + \epsilon} A\|_{S^{6/5}}.\label{eqn:TrfA-L6-I65} 
\end{align}
Now, applying the Gagliardo-Nirenberg-Sobolev inequality (for $d=3$; see \cite{LiebLoss})
\begin{align}
	\| f \|_{L^6} \ls \|\n f\|_{L^2} \label{eqn:hardy-ineq}
\end{align}
to $\|f\|_{L^6}$ in \eqref{eqn:TrfA-L6-I65}, we obtain \eqref{eqn:TrAf-need-to-est}. The proof of Lemma \ref{lem:rhoA-by-tracial-norm} is completed by the $\dot{H}^1$-$\dot{H}^{-1}$ duality and the fact that $\dot{H}^1  \cap C_c $ is dense in $\dot{H}^1$.
\end{proof}

\subsection{Bloch-Floquet Decomposition} 
 \label{sec:BF-decomp} 
Let $\lat^*$ denote the lattice reciprocal to $\lat$,  with the reciprocity relation between bases for $\lat$  and  $\lat^*$ given by $\om_i\cdot \om_j^*=2\pi\del_{ij}$.
  Define the (fiber integral) space
\begin{align}
	\mathcal{H}_{\lat}^\oplus = \{ f \in L^2_{\rm loc}&(\R^3_k \times \R^3_x) : T_s^x f= f \\
		&\text{ and } T_r^k f = e^{- i r \cdot x} f,\ \,\forall s \in \lat,\ \forall r \in \lat^* \},
\end{align}
where $T_s^k$ is the translation in the $k$-variable by $s$ and $T_r^x$ is the translation in the $x$-variable by $r$ (see \eqref{eqn:translation-op-def}). 
We write $f = f_k(x) \in \mathcal{H}_{\lat}^\oplus$ as
\begin{align}
	f=\int_{\R^3/\lat^*}^\oplus \, f_k \, d\hat k = \int_{\Omd^*}^\oplus \, f_k \, d\hat k, 
\end{align}
 for some choice of a fundamental cell $\Omd^*$ of the reciprocal lattice $\lat^*$ and $d\hat k := |\Omd^*|^{-1} dk$.

We use the Bloch-Floquet decomposition $\UBF$ mapping from $L^2(\R^3)$ into $\mathcal{H}^\oplus_{\lat}$ as
\begin{align} \label{eqn:bloch-def}
	& \UBF f := \int^\oplus_{\Omd^*} d\hat k f_k \, ,\\
\label{fk}	&  f_k(x) := \sum_{t \in  \lat} e^{- ik(x+t)} f(x+t)  
\end{align}
and the inverse Bloch-Floquet transform
\begin{align} \label{eqn:inverse-bloch}	
	& \UBF^{-1}\big( \int^\oplus_{\Omd^*} d\hat k f_k \big)(x) 
		:= \int_{\Omd^*} d\hat k \, e^{ ikx} f_k(x), \, \forall x\in \R^3.
\end{align}

\begin{lemma} \label{lem:hatf-L1}
We have, for any $f \in L^2 $,
\begin{align}
	\int_{\Omd} f_k(x) dx = \hat{f}(k) \label{eqn:int-f-p-is-f-hat}
\end{align}
\end{lemma}
\begin{proof}
By \eqref{fk} and a change of variable, we see that
\begin{align}
	\int f_k(x) dx :=& \int_{\Omd} \sum_{t \in \lat} e^{- ik(x+t)} f(x+t) d{x} \\
	=& \sum_{t \in \lat} \int_{t+ \Omd} e^{- ikx} f(x) d{x} \\
	=& \int e^{- ikx} f(x) dx.
\end{align}
Equation \eqref{eqn:int-f-p-is-f-hat} follows from the definition of the Fourier transform.
\end{proof}

Let  $\lan f\ran_{S} = |S|^{-1} \int_S f(x) dx$, the average of $f$ on a set $S$, and $\chi_S$ be the indicator (characteristic) function of $S$.
\begin{lemma} \label{lem:P-decomposed-form}
Let $f \in L^2 $ and $f_k$ be its $k$-th fiber $\lat$-Bloch-Floquet decomposition. Then for any $S \subset \Omd^*$,
\begin{align}\label{P-decomposed-form}
	\chi_S(-i\n)f = \UBF^{-1}\int^{\oplus}_{S} d\hat k \, \lan f_k \ran_{\Omd}. 
\end{align}
\end{lemma}
\begin{proof}
 Let $f \in L^2 $ with the $k$-th fiber $f_k$. Then Lemma \ref{lem:hatf-L1} shows that 
\begin{align}
	\lan f_k \ran_{\Omd} = |\Omd|^{-1} \hat{f}(k). \label{eqn:lan-fk-ran-Om-delta-eqn}
\end{align}
Using the definition of the inverse Bloch transform in \eqref{eqn:inverse-bloch} and \eqref{eqn:lan-fk-ran-Om-delta-eqn}, we see that
\begin{align}
	\UBF^{-1} \left( \int_{S}^\oplus d\hat k \, \lan f_k \ran_{\Omd} \right) =& \int_{\Omd^*} d\hat k \,  e^{  ikx} \lan f_k \ran_{\Om_\delta} \notag \\
		=& \int_{S} d\hat k \,  |\Omd|^{-1} e^{ ikx} \hat f (k)  
\end{align}
Since $d\hat k = |\Omd^*|^{-1}dk = |\Omd| dk$, the last equation yields
\begin{align}
	\UBF^{-1} \int_{S}^\oplus d\hat k \, \lan f_k \ran_{\Omd}	=& \int_S d k \,   e^{ ikx}\hat{f}(k) = \chi_S(-i\n) f, \label{eqn:UBR-inverse-1}
\end{align}
which gives \eqref{P-decomposed-form}.
\end{proof}

Let $\Proj = \chi_{B(r)}(-i\n)$ where $B(r)$ is the ball of radius $r$ centered at the origin (see \eqref{eqn:choice-of-P}). Lemma \ref{lem:hatf-L1} and \ref{lem:P-decomposed-form} imply
\begin{corollary} \label{cor:Pf-computation}
Let $f \in L^2$ and $\Br \subset \Omd^*$, then
\begin{align} 
(\Proj f)_k = |\Omd|^{-1} \hat f(k)\chi_{\Br}(k).
\end{align}
\end{corollary}

 Any $\lat$-periodic operator $A$ has a Bloch-Floquet decomposition \cite{RS4} in the sense that
\begin{align}\label{BF-deco-opr} 
	A = \UBF^{-1} \int^\oplus_{\Om^*} d\hat k A_k \UBF ,
\end{align}
where $A_k$ are operators (called {\it $k$-fibers} of $A$) on $L^2_{\rm per}$ and the operator $\int^\oplus_{\Om^*} d\hat k \, A_k$ acts on $\int^\oplus_{\Om^*} d\hat k f_k \in \mathcal{H}_{\lat}^\oplus$ as
\begin{align}\label{fib-int-action}
	\int^\oplus_{\Om^*} d\hat k A_k \cdot \int^\oplus_{\Om^*} d\hat k f_k = \int^\oplus_{\Om^*} d\hat k A_kf_k.
\end{align}

Definitions \eqref{BF-deco-opr} and  \eqref{fib-int-action} implies the following relations for  any $\lat$-periodic operators $A$ and $B$ 
\begin{align}\label{Af-fib}
	&(A f)_k= A_k f_k,\\
\label{AB-fib}
	&(A B)_k= A_k B_k,\\
\label{A-Ak-norm}
	&\|A \|= \sup_{k\in\Om^*}\|A_k \| .
\end{align}
Furthermore, we have

\begin{lemma} \label{lem:Ak-repr} Let $A$ be an $\lat$-periodic operator and $A_k$, its  $k$-fibers in its Bloch-Floquet decomposition. Then \[A_k = e^{- ixk}A_{0} e^{ ixk}.\]
\end{lemma}
\begin{proof}
We compute $(Af)_k$. Let $T_s$ denote the translation operator \eqref{eqn:translation-op-def}. Let $A_0$ denote the $0$-th fiber of $A$ in its Bloch-Floquet decomposition. By \eqref{fk} and the periodicity of $A$,
\begin{align}
	(Af)_k =& \sum_{t\in \lat}  e^{- ik(x+t)} T_{-t} Af = 
	\sum_{t\in \lat}  e^{- ikx} A e^{- i kt} T_{-t} f \\
			=& e^{- ikx} A_0  e^{ ikx} \sum_{t\in \lat}  e^{- i k(x+t)} T_{-t} f \\
			=& e^{- ikx} A_0  e^{ ikx} f_k.
\end{align}
\end{proof}

Now, we have the following result.

\begin{lemma} \label{lem:PAP-repr'}  
Let $A$ be an $\lat$-periodic operator and $A_k$, its  $k$-fibers in its Bloch-Floquet decomposition and let $r$ be such that $B(r)\subset \Om^*$. Then
 \begin{align}\label{PAP-explicit-form}	\Proj A  \Proj = b(-i \nabla) \Proj 
\end{align}
where $b (k) = \lan A_k \one\ran_{\Om}$, $1 \in L^2_{\rm per}(\R^3)$ is the constant function $1$.
\end{lemma}
\begin{proof}  Let $f_k$ be the $k$-th fiber of the Bloch-Floquet function $f$. 
\DETAILS{Since $\Udel^* \Proj \Udel = P_{\del r}$, we have $\Proj A_\delta \Proj =\Udel^* P_{\del r} A P_{\del r}\Udel$. Using 
 Lemma \ref{lem:P-decomposed-form}, with $S = B(\del r)$ and $f = A P_{\del r} \varphi$, for the r.h.s. gives
 \begin{align}\label{PAP-repr'}
	\Proj A_\delta \Proj \varphi &= \Udel^* \UBF^{-1} \int^{\oplus}_{\Omd^*} d\hat k \,   \lan (A P_{\del r} \Udel\varphi)_k \ran_{\Om}.
	\end{align}
By Corollary \ref{cor:Pf-computation} and equation \eqref{PAP-repr'}, we find
 \begin{align}
	\Proj A_\delta \Proj \varphi &= |\Omde|^{-1}  \UBF^{-1} \int^{\oplus}_{\Br} d\hat k \,   \lan A_k 1 \ran_{\Omd}  \hat  (\Udel\varphi)(k)\end{align}
where $1 \in \LperB$ is the constant function equal to $1$. Applying the inverse Bloch-Floquet transform \eqref{eqn:inverse-bloch} and using $d \hat k = |\Omd|^{-1} dk$, we see that
\begin{align}
	\Proj A_\delta \Proj \varphi &= |\Omd|^{-1} \UBF^{-1} \int_{\Br}^\oplus d\hat k \,  e^{2\pi i k x} \lan (A_\delta)_k 1 \ran_{\Omd}  \hat \varphi(k)\\
		=& b_\delta(-i\n) \varphi \label{eqn:b-delta-i-n}
\end{align}
where $b_\delta(k) = \lan (A_\delta)_k 1 \ran_{\Omd}$, which gives \eqref{PAP-explicit-form}.}
We apply Lemma \ref{lem:P-decomposed-form} with $S = \Br$ and $f = A \Proj \varphi$ (so that $\chi_S(-i\n) = \Proj$) to obtain
 \begin{align}\label{PAP-repr'}
	\Proj A \Proj \varphi &= \UBF^{-1} \int^{\oplus}_{\Omd^*} d\hat k \,   \lan (A \Proj \varphi)_k \ran_{\Omd}.
	\end{align}
By Corollary \ref{cor:Pf-computation} and equation \eqref{PAP-repr'}, we find
 \begin{align}
	\Proj A \Proj \varphi &= |\Omd|^{-1}  \UBF^{-1} \int^{\oplus}_{\Br} d\hat k \,   \lan A_k 1 \ran_{\Omd}  \hat \varphi(k),\end{align}
where $1 \in \LperB$ is the constant function equal to $1$. Using the definition \eqref{eqn:inverse-bloch} of the inverse Bloch-Floquet transform  and that $d \hat k = |\Omd|^{-1} dk$, we deduce \eqref{PAP-explicit-form}.
\end{proof}

\subsection{Passing to the macroscopic variables} \label{sec:BF-decomp} 
Define the microscopic lattice $\latde:=\del \lat$ and let $\latde^*$ be its reciprocal lattice. Define the rescaling operator 
\begin{align}\label{rescU}
	\Udel: f(x) \mapsto \delta^{-3/2} f(\delta^{-1} x)
\end{align}
mapping from the microscopic to the macroscopic scale. A change of variable in \eqref{den-def1} gives the following

\begin{lemma} \label{lem:U-rho-UAUstar}
For any operator $A$ on $L^2 $, we have 
\begin{align}
	\delta^{-3/2}\Udel\den[A] = \den[\Udel A\Udel^*].
\end{align}
\end{lemma}

Finally, note that  \begin{align}A\ \text{ is  $\lat$-periodic\ 
iff\ $\Udel A \Udel^*$ be  $\latde$-periodic}.\end{align} 
Lemma \ref{lem:PAP-repr'} implies

\begin{lemma} \label{lem:PAP-repr-del}  
Let $A$ be an $\lat$-periodic operator and $A_k$, its  $k$-fibers in its Bloch-Floquet decomposition  and let $r$ be such that $B(\del r)\subset \Om^*$. Then 
\begin{align}\label{PAP-explicit-form}	
	\Proj \Udel A \Udel^* \Proj = b(-i\delta\nabla) \Proj 
\end{align}
where $b (k) = \lan A_k \one\ran_{\Om}$, $1 \in L^2_{\rm per} $ is the constant function $1$.
\end{lemma}
\begin{proof} 
 By $\Udel^* \Proj \Udel = P_{\del r}$ and  Lemma \ref{lem:PAP-repr'}, we have 
\[
	\Proj \Udel A \Udel^* \Proj  =\Udel P_{\del r} A P_{\del r}\Udel^* = \Udel  b(-i \nabla) P_{\del r}\Udel^*.
\]
Relations $\Udel P_{\del r} \Udel^* = \Proj$ and $\Udel b(-i \nabla) \Udel^* = b(-i\delta\nabla)$ yield \eqref{PAP-explicit-form}. 
\end{proof}

\section{Dielectric Response: Proof of Theorem \ref{thm:PE-macro}} \label{sec:pf:thm:PE-macro}
In this section, we prove Theorem \ref{thm:PE-macro} modulo several  technical (though important) statements proved in Sections \ref{sec:oper-ell} and \ref{sec:nonlin-est-rough}. 

\subsection{Linearized Map}
Our starting point is  equation \eqref{phi-only-eq}, which we reproduce here 
\begin{align} \label{eqn:new-phi-only-eqn'}
	-\Delta \phi =  \ka - \den[\ft(\hphi{\phi}-\mu)],
\end{align}
where, recall, $f_{T}(\lam)$ is given in \eqref{fT}  and, recall,  
\begin{align}
	\hphi{\phi} := -\Delta - \phi. 
\end{align} 
We consider \eqref{eqn:new-phi-only-eqn'} on the function space $\phi \in H_{\rm per}^2 + \dot{H}^1 $. For such $\phi$'s, the operator $\hphi{\phi}$ is self-adjoint and bounded below so that functions of $\hphi{\phi}$ above are well-defined by the spectral theory.  

Our first step is to investigate the linearization of the map on the r.h.s. of \eqref{eqn:new-phi-only-eqn'} 
\begin{align} \label{M-def} 
	&M:= d_\phi \den[\ft (\hphi{\phi}  - \mu)]\big|_{\phi=\phi_{\rm per}}.
 \end{align}

To derive basic properties of $M$, we find an explicit formula for it. 
Recalling the relation $\ft(\lam):=f_{FD}(\lam/T)$, 
 see \eqref{fT},  and assuming that $\phi$ is close to $\phi_{\rm per}$, we write $\ft(\hphi{\phi}  - \mu)$ using the Cauchy-integral formula 
\begin{align}
	\ft(\hphi{\phi}  - \mu) = \frac{1}{2\pi i} \int_\Gamma dz \ft(z-\mu) (z-\hphi{\phi})^{-1} \label{fFD-cauchy-formula}
\end{align}
where $\Gamma$ is a positively oriented contour around the spectrum of $\hphi{\phi}$ 
not containing the poles of $\ft$ which are located at $\mu + i\pi (2k+1) T$, $k \in \Z$ (see Figure \ref{fig:cauchy-int-contour'} below), 
 in which $\e$ satisfies 
\begin{align} \label{eps-bet-mu-ineq} \e < T \pi\ \text{ and }\ -1 < \cos(\mu \epsilon). \end{align}
Here we use that $h^\phi$ is bounded from below 
and, due to the definition $\ft(\lam) = \frac{1}{e^{\lam/T}+1}$ (see \eqref{fT}) and the relation $ |\Im z|\le \pi/4 T$,  
\begin{align}
	& |\ft(z-\mu)| \ls \min(1, e^{-(\Re z-\mu)/T})) \label{oint-decay}	
\end{align}
assuring the convergence of the integral.  (Note that we do not use that $h^\phi$  has a gap and that $\mu$ is in the gap.) 

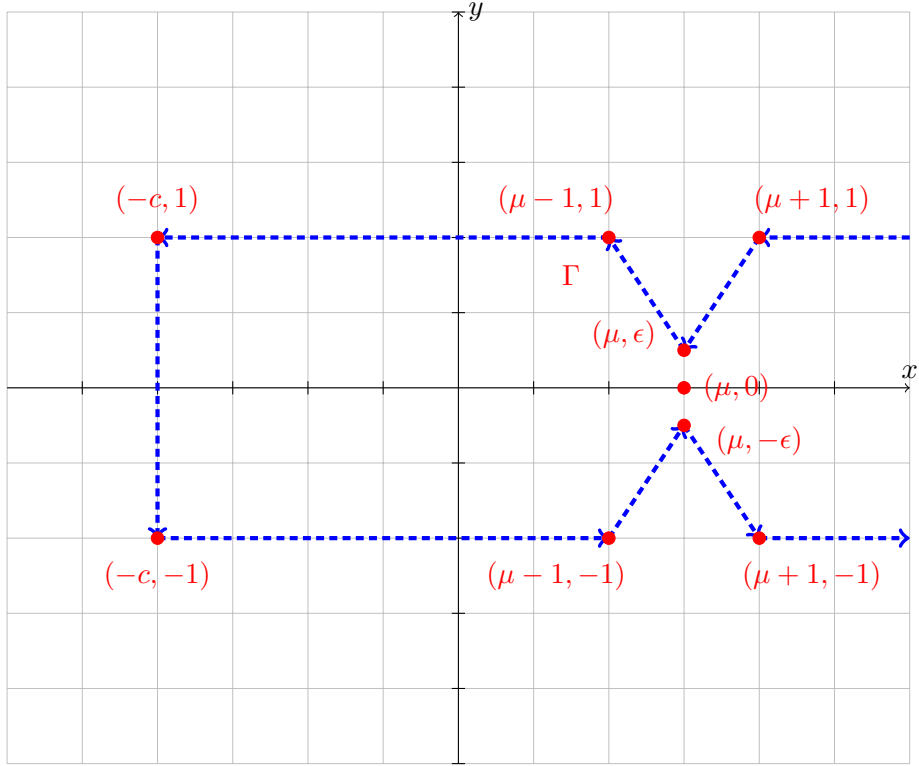
\begin{figure}[ht] 
\begin{center}
\begin{tikzpicture}
    \path [draw, help lines, opacity=.5]  (-6,-5) grid (6,5);
    \foreach \i in {1,...,5} \draw (\i,2.5pt) -- +(0,-5pt) node [anchor=north, font=\small] {} (-\i,2.5pt) -- +(0,-5pt) node [anchor=north, font=\small] {} (2.5pt,\i) -- +(-5pt,0) node [anchor=east, font=\small] {} (2.5pt,-\i) -- +(-5pt,0) node [anchor=east, font=\small] {};
    \draw [->] (-6,0) -- (6,0) node [anchor=south] {$x$};
    \draw [->] (0,-5) -- (0,5) node [anchor=west] {$y$};
		\draw [->, ultra thick, color=blue, densely dashed] (6,2) node {} -- (4,2) node {};
		\draw [->, ultra thick, color=blue, densely dashed] (4,2) node {} -- (3,0.5) node {};
		\draw [->, ultra thick, color=blue, densely dashed] (3,0.5) node {} -- (2,2) node {};
    \draw [->, ultra thick, color=blue, densely dashed] (2,2) node {} -- (-4,2) node {};
		\draw [->, ultra thick, color=blue, densely dashed] (-4,2) node {} -- (-4, -2) node {};
		\draw [->, ultra thick, color=blue, densely dashed] (-4, -2) node {} -- (2, -2) node {};
		\draw [->, ultra thick, color=blue, densely dashed] (2, -2) node {} -- (3, -0.5) node {};
		\draw [->, ultra thick, color=blue, densely dashed] (3, -0.5) node {} -- (4, -2) node {};
		\draw [->, ultra thick, color=blue, densely dashed] (4, -2) node {} -- (6, -2) node {};
		\node[color=red] at (-4,2.5) {$(-c,  1)$};
		\draw[red, fill=red] (-4,2) circle (.5ex);
		\node[color=red] at (-4,-2.5) {$(-c, - 1)$};
		\draw[red, fill=red] (-4,-2) circle (.5ex);
		\node[color=red] at (3.7,0) {$(\mu, 0)$};
		\draw[red, fill=red] (3,0) circle (.5ex);
		\node[color=red] at (2.2, 0.7) {$(\mu, \epsilon)$};
		\draw[red, fill=red] (3,0.5) circle (.5ex);
		\node[color=red] at (4,-0.7) {$(\mu, -\epsilon)$};
		\draw[red, fill=red] (3,-0.5) circle (.5ex);
		\node[color=red] at (1.3,2.5) {$(\mu - 1, 1)$};
		\draw[red, fill=red] (2,2) circle (.5ex);
		\node[color=red] at (4.7,2.5) {$(\mu + 1, 1)$};
		\draw[red, fill=red] (4,2) circle (.5ex);
		\node[color=red] at (1.3,-2.5) {$(\mu - 1, -1)$};
		\draw[red, fill=red] (2,-2) circle (.5ex);
		\node[color=red] at (4.7,-2.5) {$(\mu + 1, -1)$};
		\draw[red, fill=red] (4,-2) circle (.5ex);
		\node[color=red] at (1.5,1.5) {$\G$};
  \end{tikzpicture}
\end{center}
\caption{The contour $\G$. It depends on $\epsilon$ satisfying \eqref{eps-bet-mu-ineq}. } \label{fig:cauchy-int-contour}
\end{figure}

To simplify the expressions below, we will introduce the following notation
\begin{align}
	\oint :=& \frac{1}{2\pi i} \int_\Gamma dz \ft(z-\mu) \label{eqn:oint-adhoc} 
\end{align}
where $\Gamma$ is the contour given in Figure \ref{fig:cauchy-int-contour},  with the positive orientation.

Recall the notation for the $\lat$-periodic Hamiltonian and introduce one for the $\lat$-periodic resolvent:
 \begin{align} \label{rper-def}	
h_{\rm per}:= h^{\phi_{\rm per}}=  -\Delta -\phi^{\rm per}, \quad r_{\rm per}(z) = (z-h_{\rm per})^{-1}.
\end{align}
By Theorem \ref{thm:ideal-cryst-exist}, the electrostatic potential, $\phi_{\rm per}(y)$ associated with the solution $\rho_{\rm per}(y)$ (c.f. \eqref{phi-rho-expr}) 
  satisfies
   \begin{align}\label{A:crystal-structure} 	\phi_{\rm per} \in H^2_{\rm per}. 
  \end{align}
  Hence the operator $\hper$ is self-adjoint and the operator functions above are well-defined.  Moreover, under Assumption  \ref{A:diel}, 
\begin{align}\label{eqn:resolvent-bound-on-G}
	\sup_{z\in \G} \|(z-h_{\rm per})^{-1}\|_\infty = O(1). 
\end{align}
Finally, for any operator $h$, we denote $h^L: \al \ra h\al$ and $h^R: \al \ra \al h$. 

 The next proposition gives an explicit form for $M$ and states its  properties
  (also see \cite{CL}).

\begin{proposition} \label{pro:M-expl}
Let Assumption \ref{A:diel} hold. Then 
\begin{enumerate}
\item The operator $M$ has the following explicit representation  
\begin{align} 
\label{Mexpl} & M f = - \den \big[ \oint  r_{\rm per}(z)   f  r_{\rm per}(z)   \big] \\ 
 \label{Mexpl2'}=- &\frac12   \den \big[ \frac{\tanh(\frac{1}{2 T}(h_{\rm per}^L-\mu)) - \tanh(\frac{1}{2 T}(h_{\rm per}^R-\mu))}{h_{\rm per}^L- h_{\rm per}^R}f \big],	\end{align}
where $f \in L^2 $ on the right hand side is considered as a multiplication operator. 

\item  The operator $M$  is bounded,
 self-adjoint, positive on $L^2 $ and $\lat$-periodic (c.f. Section \ref{sec:densities}) and satisfies
 \begin{align}\label{M-bnd}
	\|M\| \ls 1. 
\end{align} 
\end{enumerate}
\end{proposition}

\begin{proof}[Proof of Proposition \ref{pro:M-expl}]
 In this proof, we {\it omit the subscript} ``per'' in $h_{\rm per}$ and $r_{\rm per}(z)$. 
We begin with item (1). Equation \eqref{Mexpl} follows from definition \eqref{M-def}, the Cauchy formula  \eqref{fFD-cauchy-formula} 
and a simple differentiation of the resolvent. 

Now, we use \eqref{Mexpl} to derive  \eqref{Mexpl2'}. By the definition of $h^L$ and $h^R$ and the second resolvent identity, we have, for any operator $\alpha$,
\begin{align}\notag
    (z-h)^{-1}& \alpha (z-h)^{-1} =  (z-h^L)^{-1} (z-h^R)^{-1} \alpha  \\
  &=(h^L-h^R)^{-1} ((z-h^L)^{-1} - (z-h^R)^{-1}) \alpha . 
\end{align}
  Using the Cauchy integral formula and the definition \eqref{eqn:oint-adhoc} and the choice of  the contour $\Gamma$ i(see Figure \ref{fig:cauchy-int-contour}), we observe that 
\begin{align}\notag
    \oint &(z-h)^{-1} \alpha (z-h)^{-1}=  (h^L-h^R)^{-1} \\ &
    \notag \times\frac{1}{2\pi i} \int_\Gamma dz \ft(z-\mu)  ((z-h^L)^{-1} - (z-h^R)^{-1}) \alpha\\ \label{M-expr'}& =  (h^L-h^R)^{-1} (\ft(h^L-\mu)  - \ft(h^R-\mu)) \alpha. 
\end{align}
Now, by definition \eqref{fT}, $\ft(\lam):= \frac{1}{e^{\lam/T}+1}$ and therefore $\ft(\lam)=\frac{1}{2}(1+\tanh(\lam/2T))$. This relation, together with \eqref{M-expr'}, gives  
\begin{align} \label{eqn:formula-for-L-alpha}   
    \oint &(z-h)^{-1} \alpha (z-h)^{-1} \notag\\
    &= \frac{1}{2}\frac{\tanh(\frac{1}{2 T}(h^L-\mu)) - \tanh(\frac{1}{2 T}(h^R-\mu))}{h^L- h^R} \alpha. 
\end{align}
This, together with  \eqref{Mexpl},  gives \eqref{Mexpl2'}. Item (1) is now proved.

Now we prove item (2). 
Since $h = h_{\rm per}$ is self-adjoint and bounded below, we can pick $c>0$ 
 sufficiently large, s.t. $h\ge - c+1$. Then, in particular, $h+c$ is invertible and, for each function $f \in L^2(\R^3)$, we define the operator
\begin{align}
    \alpha_f := (c+h)^{-1/2}  f (c+h)^{-1/2}.
\end{align}
The Kato-Seiler-Simon inequality \eqref{eqn:KSS} shows that $\alpha_f$ is Hilbert-Schmidt and 
\begin{align}
	\|\alpha_f\|_{S^2} \ls \|f\|_{L^2} \label{eqn:I2-by-L2}
\end{align}
(the $S^2$  norm  is given in \eqref{Ip-norm'}). 
 Using \eqref{Mexpl}, together with \eqref{den-def1}, we write $\lan f, \LinMOp g \ran$\\ $ = -\oint \Tr( \bar f r (z)   g  r(z))$, which can be transformed to
\begin{align}
    \lan f, \LinMOp g \ran =& - \oint \Tr( \alpha_{f}^* (c+h)  r(z)  \alpha_{g}  r(z)  (c+h)) \, .
\end{align}
Moreover, by \eqref{eqn:formula-for-L-alpha}, we have that
\begin{align}
  \label{eqn:fMg-HS-form}   & \lan f, M  g \ran = \Tr\left( \alpha_{f}^* G(h^L, h^R) \alpha_{g} \right),\\
\label{eqn:temp-f-x-y-tanh} &G(x,y) := -\frac12 \frac{\tanh(\frac{1}{2 T}(x-\mu))-\tanh(\frac{1}{2 T}(y-\mu))}{(x+c)^{-1}-(y+c)^{-1}}. 
\end{align}
Since the function $G : \R^2 \rightarrow \R$
is bounded on the set $x,y \geq -c+1$, we see that $M$ is bounded due to \eqref{eqn:I2-by-L2} and \eqref{eqn:fMg-HS-form}. 

Moreover, we can also see from expressions \eqref{eqn:fMg-HS-form} - \eqref{eqn:temp-f-x-y-tanh} that $\LinMOp$ is symmetric since $G$ is real and $h^L$ and $h^R$ are self-adjoint in the space $S^2$. Since $M$ is bounded, it is self-adjoint. Since the function $G$ in \eqref{eqn:temp-f-x-y-tanh} is positive for $x,y \geq -c+1$, equation \eqref{eqn:fMg-HS-form} and spectral theorem on $S^2$ show that $\lan f, \LinMOp f \ran =\Tr\left( \alpha_{f}^* G(h^L, h^R) \alpha_{f} \right)> 0$ for any nonzero $f \in L^2(\R^3)$. This shows that $M$ is positive. 

 Finally, formula \eqref{Mexpl} and the fact $h = h_{\rm per}$ and $r = r_{\rm per}(z)$ are $\lat$-periodic show that $M$ is $\lat$-periodic.

 To prove bound \eqref{M-bnd}, we  use \eqref{Mexpl} and \eqref{rhoA-by-tracial-norm1} to find
\begin{align} 
\label{Mest1}  \|M f\|_{L^2}& \ls \|(1-\Delta)  \oint  r (z)   f  r (z)\|_{S^2}    \\ 
 \label{Mest2}&\ls \big|\oint\big|\|(1-\Delta)    r (z)\|^2  \| f (1-\Delta)^{-1}\|_{S^2}    \, .	\end{align}
Now,  writing $-\Delta=h-z+\phi_{\rm per}+z,$ for $z\in \Gam$,  and using the uniform boundedness of $\|\phi_{\rm per}\|_{H^2}$ which follows from Assumption \ref{A:diel}, we derive the estimate 
\begin{align}\label{Del-r-est}\|(1-\Delta) r(z)\|\ls\ & 1,\end{align}
which, together with \eqref{Mest2} and  the Kato-Seiler-Simon inequality \eqref{eqn:KSS}, gives  bound \eqref{M-bnd}. 
  The proof of Proposition \ref{pro:M-expl} is now complete. \end{proof}

\subsection{Scaling and splitting}\label{sec:scal}
This step is to pass   from the microscopic coordinate $y$ to the macroscopic one, $ x= \delta y$ passing to the macroscopic quantities (with superscripts $\delta$) which are related the microscopic quantities (with subscripts $\delta$) as 
\begin{align}\label{ka-phi-macro-micro-scaling}	& \ka^\delta = \delta^{-3/2}\Udel \ka_\delta,\ 
	\phi^\delta(x) = \delta^{1/2} (\Udel\phi_\delta)(x)= \delta^{-1} \phi_\delta(\delta^{-1}x ),\\ 
	&\ka_{\rm per}^\delta(x) := \delta^{-3} \ka_{\rm per}(\delta^{-1}x ) = (\delta^{-3/2}\Udel \ka_{\rm per})(x),
\end{align}
where $\Udel: f(x) \mapsto \delta^{-3/2} f(\delta^{-1} x)$, the $L^2(\R^3)$-unitary scaling map, see \eqref{rescU} (note that the $L^1$-norm, hence total charge, is preserved under this scaling). Let
\begin{align} 
	\ka^\delta(x) = \ka_{\rm per}^\delta(x) + \ka'(x) 
\end{align}
be the macroscopic perturbed background potential. Accordingly, we rescale equation \eqref{eqn:new-phi-only-eqn'}   by applying $\delta^{-3/2}\Udel$ to it. 
Using Lemma \ref{lem:U-rho-UAUstar} and relations $\Udel f_{T}(\hphi{\phi}-\mu)\Udel^*=f_{T}(\Udel \hphi{\phi}\Udel^*-\mu)$ and \[\Udel \hphi{\phi}\Udel^* = -\delta^2 \Delta - \delta \phi^\delta,\] 
we arrive at the rescaled electrostatic potential equation
\begin{align}
	\label{phi-del-eq}- &	\Delta \phi^\delta = 4\pi \ka^\delta - F_\delta(\phi^\delta),\\ 
	\label{F-del-phi}&F_\delta(\phi) =  4\pi \den[ \ft(-\delta^2 \Delta - \delta \phi - \mu)]. 
\end{align}
We will consider \eqref{phi-del-eq} on the space $H_{\rm per}^2 +\dot{H}^1 $. 

Let $\phi_{\rm per}^\delta = \delta^{1/2} \Udel \phi_{\rm per}$, where $\phi_{\rm per}$ is the periodic potential associated to the periodic solution $(\rho_{\rm per}, \mu_{\rm per}$) of \eqref{DFT-rho-eq} with periodic background charge $\ka_{\rm per}$ given in Theorem \ref{thm:ideal-cryst-exist}. We split the solution $\phi^\delta$ into the big part $\phi_{\rm per}^\delta$ and the fluctuation 
\begin{align}\label{varphi}
	\varphi \equiv \varphi^\delta:= \phi^\delta - \phi_{\rm per}^\delta .
\end{align}

We rewrite equation \eqref{phi-del-eq} by expanding the r.h.s. around $\phi_{\rm per}^\delta$ to obtain
\begin{align}\label{eqn:macro-eqn-varphi}
	K_\delta \varphi = 4\pi \kp + N_\delta(\varphi) 
\end{align}
where $N_\delta$ is defined by this expression and 
 \begin{align} 
 \label{KdeltaMdelta-def} K_\delta  =&  -\Delta +  M_\delta \, ,\ \text{ with }\	M_\delta  = d_\phi F_\delta(\phi_{\rm per}^\delta).
\end{align}
 Note that the inputs into this equation are $\phi_{\rm per},\ \mu=\mu_{\rm per}$ and $\kp$ (cf. \eqref{psi-eq}).

As was mentioned in the introduction, we prove Theorem \ref{thm:PE-macro} by decomposing $\varphi$ in \eqref{varphi} in small and  large momentum parts (c.f. \cite{CLS}). We use rough estimates for high momenta while we expand in $\delta$ and use a perturbation argument for low momenta.

We begin with a discussion of the linearized map, $K_\delta$.
Since we rescaled equation \eqref{DFT-rho-eq} by  applying $\delta^{-3/2}\Udel$ to it and rescaled the microscopic potentials via \eqref{ka-phi-macro-micro-scaling}, it follows that
\begin{align}
	F_\delta = \delta^{-3/2} \Udel \circ F \circ (\delta^{-1/2}\Udel^*) \label{eqn:F-delta-rescaling-rel}
\end{align}
where $F = F_{\delta = 1}$. Thus, by the definition of $M_\delta$ in \eqref{KdeltaMdelta-def} and the fact it is linear, it can be written as
\begin{align} \label{Mdelta-M-relation'}
	M_\delta  = \delta^{-2}\Udel M\Udel^* ,\end{align}
where $M:=M_{\del = 1}$ and is given by \eqref{M-def}.

Recall that an operator $A$ on $L^2(\R^3)$ is said to be $\lat$-periodic if and only if it commutes with the translations $T_s$ (see \eqref{eqn:translation-op-def}) by all lattice elements $s \in \lat$. As an immediate consequence of Proposition \ref{pro:M-expl}, representation \eqref{Mexpl}, and the rescaling \eqref{Mdelta-M-relation'}, we have the following result

\begin{proposition} \label{pro:Mdel-prop}
Let Assumption \ref{A:diel} hold. Then $M_\delta$ is $\latd$-periodic, positive (so that $K_\delta = -\Delta + M_\delta > -\Delta$), bounded on $L^2 $ with an $O(\delta^{-2})$ bound, and has the following representation
\begin{align} \label{Mexpl2}	
	M_\delta \varphi = - \delta \den\left[ \oint \rperDelta(z)\varphi \rperDelta(z)\right] \, , 
\end{align}
where the resolvent operator $\rperDelta(z)$ acting on $L^2(\R^3)$ is given by
\begin{align} \label{rperDelta-def} 
&\rperDelta(z) = (z-\hperDelta)^{-1},\ \quad \hperDelta = -\delta^2 \Delta - \delta \phi_{\rm per}^\delta.
\end{align}
\end{proposition} 


\subsection{Lyapunov-Schmidt decomposition}

To separate small and large momenta, we now perform a Lyapunov-Schmidt reduction. 

Let $\chi_{Q}$ be the characteristic function of a set $Q \subset \R^3$. Let $\Omde^*$ denote the fundamental domain of $\latde^*$ as in Subsection \ref{sec:BF-decomp}. We recall the definition of the orthogonal projection onto low momenta (as \cite{CLS})
 \begin{align} \label{eqn:choice-of-P}
	\Proj = \chi_{\Br}(-i\n) \, ,
\end{align}
where $B(\cutOff)$ is the ball of radius $r$ centred at the origin. With $m$ given in \eqref{m} and estimated in \eqref{m-est}, we choose $r$ such that $B(r) \subset \Omde^*$ and  
\begin{align} \label{a-m}	a:=\delta r = O(1) \text{ small, but }\ a^4 \gg m,
\end{align}
 is  independent of $\del$ and $T$ (or $m$) and is fixed. Below, we use the convention that $\ls$ is independent of $r$, $\delta$ and $T$.  \DETAILS{!!!{\bf(In order not to clatter formulae below, we do not display the  $a$ dependence of the constants, even though 
 it could be easily traced.\footnote{In refinements, it is natural to make $a$ dependent on $\delta$ and $T$.})}}
Let 
\begin{align}
	\bar \Proj = 1- \Proj \label{eqn:bar-Proj}
\end{align}
be the orthogonal projection onto the large momenta. We decompose 
\begin{align}
	\varphi = \varphi_s + \varphi_l, \label{eqn:varphi-decomposition}
\end{align}
where $\varphi_s = \Proj \varphi$ and $\varphi_l = \bar \Proj\varphi$. Here $s$ stands for small momentum and $l$ stands for large momenta. We split \eqref{eqn:macro-eqn-varphi} as
\begin{align}
	&\Proj K_\delta (\varphi_s + \varphi_l) = \Proj\kp + \Proj N_\delta(\varphi), \label{eqn:P-eqn}\\
	&\bar \Proj K_\delta  (\varphi_s + \varphi_l) = \bar \Proj \kp + \bar \Proj N_\delta(\varphi) \label{eqn:barP-eqn}\, .
\end{align}
We solve \eqref{eqn:barP-eqn} for $\varphi_l$ in 
 the ball
\begin{align}\label{Bl-def} 
	B_{l, \del} := &\{\varphi \in \bar \Proj  H^1  : 
	 \|\varphi\|_{\dot{H}^1} \le  c_l\}, \end{align}
while keeping $\varphi_s$ fixed in the (deformed) ball 
\begin{align}\label{Bs-def} 
	&B_{s,\delta} := \{ \varphi \in \Proj H^1 :  \|\varphi \|_{\delta} \le c_s 
	 \},\end{align} 
 with the  norm $\|\varphi \|_{\delta}$ given by
\begin{align} \label{s-norm-def}
	\|\varphi \|_{\delta}^2 :=\sum_0^1\z^{2(i-1)} \|\n^i\varphi\|_{L^2}^2 ,\ \z:=\delta m^{-1/2} .
\end{align}
The constants $c_s$ and $c_l$ above {(\it should not be confused with the estimating function $c_T$ which appeared in Theorem \ref{thm:PE-macro})} are chosen to satisfy the conditions  
\begin{align}\label{cs-cl-cond} &\z \ll c_s \ll   c_l  \ll   \theta^{- 3/2}\z,
 \end{align}
where $\theta:= m^{- 8/9}  \delta$  and, recall, $\z:=\delta m^{-1/2}$. 
 The latter condition can be satisfied, provided 
\begin{align}  \label{m-ineq} 
\theta:= m^{-8/9} \delta \ll  1.\end{align}
Due to estimate  \eqref{m-est}, this is equivalent to  condition \eqref{beta-del-cond}. 

 We see that, while our model is parametrized by  $\delta$ and $\beta$ satisfying \eqref{beta-del-cond}, 
our method  is determined by the parameters   $a$, $c_s$ and $c_l$, satisfying \eqref{a-m} and 
\eqref{cs-cl-cond}. 

 The subleading term, $\psi$, in \eqref{phi-del-exp} just fits into $B_{s,\delta}$: $\|\psi \|_{\delta}\sim \z\ll c_s$.
Finally, we note that   since $\n^{-1} \bar \Proj\le r^{-1} \bar \Proj, \n^{-1}:=\n \Delta^{-1}, r=a/\del$, we have
\begin{align} \label{s-Hdot-est}\|\varphi\|_{L^2} \ls   m^{ 1/2}\z \|\varphi\|_{\dot{H}^1}, \quad
 \|\varphi\|_{\delta}  
  \ls   \|\varphi\|_{\dot{H}^1} ,\quad  \forall \varphi\in \Ran \bar\Proj.\end{align}
 Eq. \eqref{s-Hdot-est} shows that, if $m^{ 1/2}\z=\del \ll c_s/c_l$, then, in the $L^2$-norm, $B_{l,\delta}$ is much smaller that $B_{s,\delta}$.

In the proofs below, we will use the convention  $\|\cdot\|_{\dot{H}^{0}}\equiv \|\cdot\|_{L^2}$ 
and the estimates of the nonlinearity $N_\delta$ (defined implicitly through \eqref{eqn:macro-eqn-varphi}) proved in Proposition \ref{prop:nonlinear-rough} in Section \ref{sec:nonlin-est-rough} below, under Assumption \ref{A:diel}:
\begin{align} \label{Ndel-est'}	
 \|N_\delta(\varphi_1) - & N_\delta(\varphi_2)\|_{L^2}  \notag \\& \ls  m^{-\frac13}
   \delta^{-1/2} (\|\varphi_1\|_{{ \delta}} + \|\varphi_2\|_{{ \delta}})  
	\|\varphi_1 - \varphi_2\|_{{ \delta}}.\   
\end{align}\begin{proposition} \label{pro:phir-fixed-pt} 
 Let Assumptions \ref{A:diel} -  \ref{A:scaling} hold. 
 Assume $\varphi_s \in B_{s,\delta}$ and that \eqref{a-m} holds. Then equation \eqref{eqn:barP-eqn} on $B_{l,\delta}$ has a unique solution $\varphi_l = \varphi_l(\varphi_s)\in B_{l,\delta}$. 
\end{proposition}
\begin{proof}[Proof of Proposition \ref{pro:phir-fixed-pt}]We use that,  by Proposition \ref{pro:Mdel-prop}, $\bar K_\delta:=\bar \Proj  K_\delta  \bar \Proj$ is invertible on the range of $\bar \Proj$ (see \eqref{eqn:bar-Proj}) to convert \eqref{eqn:barP-eqn} into a fixed point problem  
\begin{align}\label{fp-phil}	&\varphi_l = \Phi_l(\varphi_l) =\Phi_l' + \Phi_l''(\varphi_l),\end{align}
where 
\begin{align} \label{Phi'} 
	 &\Phi_l':=  \bar K_\delta^{-1} (-M_\delta \varphi_s +\bar \Proj \kp),\\
 \label{Phi''} 
 	  &\Phi_l''(\varphi_l):= \bar K_\delta^{-1}\bar \Proj N_\delta(\varphi_s + \varphi_l). 
\end{align}
Given $\varphi_s$, this is a fixed point problem for $\varphi_l$. We will solve this problem in the ball $B_{l,\delta}$ defined in \eqref{Bl-def}). 
Let $\dot{H}^{0}\equiv L^2$.  We begin with the following simple but key lemma 
\begin{lemma}\label{lem:Kinv-est-k} 
 Let Assumption \ref{A:diel}   
  hold
  and let $c_T:=T^{-1} e^{-\etal/T}\ls 1$ (which is weaker than  Assumption \ref{A:scaling}). 
   Then, for $f \in L^2(\R^3)$, 
\begin{align}\label{Kinv-est} 
	& \|\bar K_\delta^{-1} f\|_{\dot{H}^{k-i}} \ls r^{-2+k} \| f\|_{\dot{H}^{-i}}, \,  i\le k,\ k=0, 1\\
 \label{KinvMP-est}
	&\|\bar K_\delta^{-1} M_\delta \Proj f\|_{L^2} \ls   \|f\|_{L^{2}},\\ 
	\label{KinvMP-est'}
	&\|\bar K_\delta^{-1} M_\delta \Proj f\|_{\dot H^{1}} \ls   \|f\|_{\delta}. 
\end{align}
\end{lemma}
\begin{proof}[Proof of Lemma \ref{lem:Kinv-est-k}]  Since $-\Delta \bar\Proj\ge r^2\bar\Proj$, we have the inequality $r^2 \|f\|^2$ $\ls \lan f, \bar K_\delta f\ran\le \|f\| \|\bar K_\delta f\|$, which gives $r^2 \|f\|\ls  \|K_\delta f\|$, which implies \eqref{Kinv-est} for $k=i=0$.

 Since $  \bar K_\delta\bar\Proj\ge -\Delta\bar\Proj$, we have $ \|\bar\Proj f\|_{\dot{H}^1}^2\ls \lan f, \bar K_\delta f\ran\le \| \bar\Proj f\| \|\bar K_\delta f\|$.  This inequality and $\| \bar\Proj f\|=\|  \n^{-1}\bar\Proj  \n f\|\le r^{-1} \| \n f\|$, where  
\begin{align}\label{nabla-inverse-def}	\n^{-1} := \n (-\Delta)^{-1}, 
\end{align}
 give $ \|\bar\Proj f\|_{\dot{H}^1}\ls r^{-1}  \|\bar K_\delta f\|$, which implies \eqref{Kinv-est} for $k=1$.

 Inequality \eqref{Kinv-est}, with $i=0$, 
  and the bound $\|M_\delta\| \ls  \delta^{-2} $, proven in Proposition \ref{pro:Mdel-prop},
 yield 
 \begin{align}
 \label{KinvMP-estk}
	&\|\bar K_\delta^{-1} M_\delta \Proj f\|_{\dot{H}^k} \ls  r^{k}  \|f\|_{L^{2}},
\end{align}
for $k= 0, 1$, which for $k=0$ implies  
 \eqref{KinvMP-est}. 

Finally, we prove more subtle \eqref{KinvMP-est'}. Using $\n^{-1}$ from \eqref{nabla-inverse-def}, we write  $\nabla\bar K_\delta^{-1} M_\delta f=\nabla \bar K_\delta^{-1}\nabla \cdot (\nabla^{-1} M_\delta ) f$. Proposition \ref{pro:Mdel-prop} shows that $\nabla \bar K_\delta^{-1} \nabla \leq 1$. It follows
\begin{align}
	\|\nabla\bar K_\delta^{-1} M_\delta f\|_{L^2} \ls \|\bar \Proj \n^{-1} M_\delta f\|_{L^2}.
\end{align}
This bound and Proposition \ref{prop:Mdel-bnd} of Appendix \ref{sec:Mdel-bnd} imply  
 \eqref{KinvMP-est'}.  
\end{proof}

Definition \eqref{Phi'} and Eqs \eqref{KinvMP-est'} and  \eqref{Kinv-est}, with $k=1, i=0$,      
 show that  
 \begin{align}
 \label{eqn:Phi-s-bound-in-B-l-delta-1}	\|\Phi'_l\|_{\dot{H}^1}   &\ls  
   \|\varphi_s\|_{\delta}+r^{-1} \| \ka'\|_{L^2}.
\end{align} 

   For the nonlinear term, 
   $\bar \Phi''_l(\varphi_l):= K_\delta^{-1}\bar \Proj N_\delta(\varphi_s + \varphi_l)$ (see \eqref{Phi''}),   Eqs \eqref{Kinv-est}, with $k=1, i=0$,   
   \eqref{Ndel-est-rough} 
 and the inequality $\|\varphi_l\|_{s,\delta} \ls  \|\varphi_l\|_{\dot{H}^1}$ (see \eqref{s-Hdot-est}) give  
\begin{align}
	&\|\Phi''_l(\varphi_l) \|_{\dot{H}^1} \ls  r^{-1} m^{-\frac13}
	 \delta^{- 1/2}(\|\varphi_s\|_{\delta} + 
	 \|\varphi_l\|_{\dot{H}^1})^2. 
	\label{eqn:Phi-prime-bound-in-B-l-delta-1} 
\end{align}
Since $\|\varphi_s\|_{ \delta} \le c_s$ and  $\|\varphi_l\|_{ \delta} \le c_l$ for $\varphi_s \in B_{l,\delta}$ and $\varphi_s \in B_{s,\delta}$ (see \eqref{Bs-def} and \eqref{Bl-def}) and, due to our assumption \eqref{cs-cl-cond},  we have
 \begin{align}  \label{cl-cond''}  \del \| \kp\|_{L^2}&+  c_s 
 +\del^{1/2}   m^{-1/3}
     (c_s + c_l)^2\ll c_l,\end{align}
   \eqref{eqn:Phi-s-bound-in-B-l-delta-1} -\eqref{cl-cond''} 
    show that
$\Phi_l$ maps $B_{l,\delta}$ into itself. 

  Once more, by Eqs \eqref{s-Hdot-est}, \eqref{Kinv-est}, with $k=1, i=0$,    \eqref{Ndel-est'} 
   and  \eqref{KinvMP-est'}, 
 we see that $\Phi_l$ satisfies
\begin{align}
\label{Phi-l-differ-est}	 \|\Phi_l(\varphi_1) &- \Phi_l(\varphi_2)\|_{\dot{H}^{1}}  \notag \\
& \ls r^{-1}   m^{-\frac13}
 \delta^{-1/2} (\|\varphi_1\|_{{ \delta}} + \|\varphi_2\|_{{ \delta}})  
	\|\varphi_1 - \varphi_2\|_{{ \delta}}\, 
\end{align} 
 and therefore, since $r=a/\del$,  is a contraction on $B_{l,\delta}$  for $ m^{-\frac13}  \delta^{1/2}c_l\ll 1$, which follows from \eqref{cs-cl-cond}.  
  Proposition \ref{pro:phir-fixed-pt} now follows by applying the fixed point theorem on $B_{l,\delta}$.
 \end{proof} 


Let $\varphi_l = \varphi_l(\varphi_s)$ be the solution to equation \eqref{eqn:barP-eqn} given in Proposition \ref{pro:phir-fixed-pt} with $\varphi_s \in B_{s,\delta}$. 
Later on we will need a Lipschitz estimate on the solution,  $\varphi_l(\varphi_s) \in B_{l,\delta}$. 

\begin{lemma} \label{lem:vphil-Lip-est}
  If $\varphi,\psi \in B_{s,\delta}$, then  the solution,  $\varphi_l(\varphi_s) \in B_{l,\delta}$, to \eqref{eqn:barP-eqn}  given in Proposition \ref{pro:phir-fixed-pt} satisfies the estimate
    \begin{align} \label{vphil-Lip-est}	\|\varphi_l(\varphi)-\varphi_l(\psi)\|_{\dot{H}^1} \ls \|\varphi - \psi\|_{ \delta}.\end{align}
\end{lemma}
\begin{proof}
Since $\varphi_l(\varphi), \varphi_l(\psi)$ satisfy \eqref{eqn:barP-eqn} (and therefore \eqref{fp-phil}), we see that
\begin{align} \label{vphil-Lip-rel}	 \varphi_l(\varphi) &- \varphi_l(\psi) = -\bar K_\delta^{-1}M_\delta (\varphi - \psi) \notag\\
		&+ \bar K_\delta^{-1} \bar \Proj (N_\delta(\varphi + \varphi_l(\varphi) ) - N_\delta(\psi + \varphi_l(\psi))).
\end{align}
Using Eqs \eqref{vphil-Lip-rel}, \eqref{Kinv-est}, with $k=1, i=0$,   and  \eqref{KinvMP-est'} 
  and nonlinear estimate \eqref{Ndel-est'} 
and going through the same arguments as in the proof of Proposition \ref{pro:phir-fixed-pt}, we show  
\eqref{vphil-Lip-est}.
\end{proof}

We substitute $\varphi_l = \Phi_l(\varphi_l)$ (see \eqref{fp-phil}), with $\Phi_l(\varphi_l)$ given by \eqref{fp-phil}-\eqref{Phi''} 
 into equation \eqref{eqn:P-eqn} 
and use that $\Proj K_\delta \bar \Proj  = \Proj M_\delta \bar \Proj$ to arrive at the following equation
\begin{align}
 \label{phis-eq}	\ell \varphi_s = Q \kp + QN(\varphi(\varphi_s)),
\end{align}
where $\varphi(\varphi_s) = \varphi_s + \varphi_l(\varphi_s)$ with $\varphi_l(\varphi_s)$ being the solution of \eqref{eqn:barP-eqn}, and
\begin{align}\label{ell-def} 
 	&\ell := \Proj K_\delta \Proj  - \Proj M_\delta \bar K_\delta^{-1} M_\delta \Proj ,\\ 
	&Q := \Proj -\Proj M_\delta \bar K_\delta^{-1} \label{Q-def}. 
\end{align}
 Note that $\ell$ is the Feshbach-Schur map of $K_\delta :=-\Delta +  M_\delta$ with projection $\Proj$.

In Subsection \ref{sec:vphis-fp} below, we prove the following  
\begin{proposition} \label{prop:vphis-fp-exist} 
Under Assumption \ref{A:diel}, Eq. \eqref{phis-eq} 
 has a unique solution $\varphi_s \in B_{s,\delta}$. 
\end{proposition}
 As a consequence of Propositions \ref{pro:phir-fixed-pt} and \ref{prop:vphis-fp-exist} and equations \eqref{eqn:varphi-decomposition} and \eqref{s-Hdot-est}, equation \eqref{eqn:macro-eqn-varphi} has the unique solution $\varphi =\vphi_s + \vphi_l \in H^1 (\R^3)$, with the estimate 
\[\|\vphi\|_{\del} \le c_s+c_l.\] 
This proves the existence and uniqueness of the solution $\phi_\delta \in (H_{\rm per}^2 + H^1)(\R^3)$ of \eqref{phi-del-eq} (and therefore of \eqref{phi-only-eq}) with $\ka$ given in \eqref{m-del}.  
This completes the proof of Theorem \ref{thm:PE-macro}(1). \qquad  \qquad  \qquad  \qquad  \qquad  \qquad   \qquad   \qquad  \qquad  \qquad  \qquad   \qquad  \qquad    $\Box$

Now, we address Theorem \ref{thm:PE-macro}(2). 
 Below, we let $\beta = T^{-1}$, so that 
  \[\cbet=\beta  e^{-\beta\etal}=:\cbeta.\] We begin with a result, 
 proven in Section \ref{sec:oper-ell}, which gives a detailed description of the operator $\ell$.

\begin{proposition} \label{pro:ell-expansion-full}
On $\text{ran } \Proj$, the operator $\ell$ in \eqref{ell-def} is a  smooth, real, even function of $-i\n$ and it has the expansion 
\begin{align}
	\ell =& \nu  -\nabla \epsilon 
	\n + O(\delta^2 (-i\nabla)^4)  \label{eqn:ell-expansion-full}
\end{align}
where $\nu = \delta^{-2}|\Om|^{-1} (m +  O(\cbeta^{2}))$,  with $m$ given in \eqref{m}, and  
  $\epsilon$ is a  matrix given explicitly in \eqref{eps} - \eqref{eps''} and satisfies the estimate 
  \begin{align}\label{eps-ineq'} 
\epsilon   \ge \one - O(\cbeta^2) .	
\end{align}
   \end{proposition}

By Proposition \ref{pro:ell-expansion-full} 
 the leading order term in $\ell$ is given by 
\begin{align}
	  &\ell_0 := \nu - \nabla \epsilon 
   \nabla, \label{ell0-def}
\end{align}
where $\nu = \delta^{-2}|\Om| (m +  O(\cbeta^2))$,  
 with $m$ given in \eqref{m},   $\e\ge \one - O(\cbeta^2)$.

To construct   an expansion of $\varphi_s$, we 
let $\psi$  be the solution to the equation
\begin{align}
	\ell_0 \psi = \ka' \, \label{eqn:phi0-def}
\end{align}
(since $\nu>0$ and $\e>0$, this solution exists)  
and  write 
\begin{align}
	\varphi_s = \Proj \psi + \psi_1 \label{eqn:varphis-decomposition}
\end{align}
where $\psi_1$ is defined by this expression. 
In Subsection \ref{sec:psi1-fp} below to prove the following 
\begin{proposition} \label{prop:psi1-est} 
Under Assumption \ref{A:diel},  $\psi_1 \in B_{s,\delta}$  obeys the estimate 
\begin{align} \label{psi1-est}
	\|\psi_1\|_{\delta} \ls ( m^{1/2}+ \theta^{1/2})\z. 
	\end{align}
\end{proposition}

Due to \eqref{eqn:varphi-decomposition} and \eqref{eqn:varphis-decomposition}, the solution $\varphi$ of equation \eqref{eqn:macro-eqn-varphi} can be written as
 \begin{align}\label{vphi-exp}
 	\varphi = \Proj \psi + \psi_1 + \varphi_l \, 
\end{align}
with $\psi_1 \in B_{s,\delta}$, satisfying estimate \eqref{psi1-est},  and $\varphi_l \in B_{l, \delta}$.

To  complete the proof of item (2) of Theorem \ref{thm:PE-macro}, 
we notice that \eqref{varphi}, \eqref{vphi-exp} and the relation $\Proj  \psi = \psi - \bar \Proj  \psi$ 
imply \eqref{phi-del-exp} with 
\begin{align}
 \label{vphi-rem1}	
\varphi_{\rm rem} =& 	\psi_1- \bar \Proj  \psi + 
\varphi_l \, .
\end{align}
Thus it remains to estimate the remainders above (see \eqref{phidel-rem-est}). 

Eq. \eqref{psi1-est} controls $\psi_1$. 
\DETAILS{shows that \[\|\psi_1\|_{\dot{H}^1} \ls 
 m^{- 11/6}\delta^{ 5/2}.\]} 
To control the term $-\bar \Proj  \psi$, we use \eqref{nabla-inverse-def} and  $\ell_0^{-1} \bar \Proj \le r^{-2}$ 
to obtain, for $i=0, 1$,
\begin{align}
\notag	\|  \bar \Proj  \psi \|_{\dot{H}^i} & = \|\nabla^i \bar \Proj \ell_0^{-1} \ka'\|_{L^2} = \|\ell_0^{-1} \bar \Proj \nabla^i \ka'\|_{L^2}\\ 
 \label{barPpsi-est}		&  \ls r^{-2} \| \ka'\|_{\dot{H}^i}. \, 
\end{align} 
 Since, by condition 
 \eqref{beta-del-cond},  $( m^{1/2}+ \theta^{1/2})\z^{2-i}\gg \delta^{2}$, 
Eqs. \eqref{psi1-est} and \eqref{barPpsi-est},  together with \eqref{vphi-rem1}, show that 
\begin{align}  \label{vphi-rem1-est}\|\psi_1- \bar \Proj  \psi\|_{\del} 
\ls  ( m^{1/2}+ \theta^{1/2})\z. \end{align} 

  By Proposition \ref{pro:phir-fixed-pt}, $\varphi_l$ is in the range of $\bar \Proj$ and bounded as $ \|\varphi_{l}\|_{\dot{H}^1} \ls  c_l$. 
Hence, using \eqref{s-Hdot-est} and  taking $c_l =\om^{-1/4}\z,\ \om:=\max(\theta^{2}, m)\ll 1$
 (satisfying \eqref{cs-cl-cond}) and using that $r^{-1}\om\z= \om m^{1/2}  \z^{2}\ll  m^{1/4}  \z^{2}$,  gives 
\begin{align}  \label{vphil-est}	&\|\varphi_{l}\|_{L^2} \ll   m^{1/4}  \z^{2},
	 \quad \|\varphi_{l}\|_{\dot{H}^1} \ls  \om^{-1/4}\z. 
	\end{align} 
 \DETAILS{If we use instead definition \eqref{Bl-def} and inequality \eqref{s-Hdot-est}, then we find the estimate $\|\varphi_{\rm rem, 2}\|_{\del} \ls c_l .$
Taking here $c_l \sim \z=  m^{-1/2}  \delta$ (see condition \eqref{cs-cl-cond}) gives the rougher estimate 
\begin{align} 
&\|\varphi_{\rm rem, 2}\|_{\del} \ls \z=m^{-1/2}  \delta \sim \delta^{7/16},	\label{eqn:stronger-est-for-ex-0-rem-2}
\end{align} 
{\bf which is good enough for us, since it $\ll$ than  $\|\psi\|_{\del} \sim (m^{-1/2}  \delta) \sim \delta^{7/16}$.}}
By \eqref{m-est}, Eqs \eqref{vphi-rem1-est} and \eqref{vphil-est} imply  part (2) of Theorem \ref{thm:PE-macro}. 

Finally,  part (3) of  Theorem \ref{thm:PE-macro} follows from Proposition \ref{pro:ell-expansion-full} and Eqs. \eqref{ell0-def} and \eqref{eqn:phi0-def}.     \quad $\Box$


\subsection{Small quasi-momenta: Proof of Proposition \ref{prop:vphis-fp-exist} }\label{sec:vphis-fp}
Our starting point is equation \eqref{phis-eq}.
By Proposition \ref{pro:ell-expansion-full}, the operator $\ell$ given in \eqref{ell-def} is invertible. Hence
 we can  rewrite \eqref{phis-eq} as the fixed point problem:
\begin{align} \label{phis-fp'}
	 \varphi_s = \Phi_s(\varphi_s),\ \Phi_s(\varphi_s):= -\ell^{-1}Q (\ka - N(\varphi(\varphi_s))),
\end{align}
where $\varphi(\varphi_s) = \varphi_s + \varphi_l(\varphi_s)$ with $\varphi_l(\varphi_s)$ being the solution of \eqref{eqn:barP-eqn}, and $Q$ is given in  \eqref{Q-def}.

  First,
   we estimate the operator $\ell^{-1}Q$ in $\Phi_s$.
 Recall, $m$ is given in \eqref{m}.  

\begin{lemma} \label{lem:ell-inv-Q-est}
 Assume \eqref{a-m} and, recall,  $\z := \delta  m^{- 1/2}$. Then   
\begin{align}
	\|\ell^{-1}Qf\|_{ \delta} 
	 \ls \z \| f\|_{L^2}. 
\end{align}
\end{lemma}
\begin{proof} By the choice $a := \delta r = O(1)$ 
 (see  \eqref{a-m}), we have that that \[O(\delta^2 (-i\nabla)^4)=O(a^2 (-i\nabla)^2)\ \text{ on }\ \Ran \Proj.\] By Proposition \ref{pro:ell-expansion-full}, we have that  $\nu = \delta^{-2}|\Om|^{-1} (m +  O(\cbeta^{2}))$, which, together with the lower bound in \eqref{m-est}, implies  
\begin{align} \label{nu-est'} \nu \gs \delta^{-2}|\Om|^{-1} m=|\Om|^{-1} \z^{-2}.\end{align} These two facts and Eq. \eqref{eqn:ell-expansion-full} imply $\n^k\ell^{-1} 
\ls \z^{2-k},\ k=0, 1, 2$, which gives
\begin{align} 
\label{n-k-ell-1-est}	
 \|\n^k\ell^{-1}  \| \ls  \z^{2-k},\ k=0, 1, 2, 
\end{align}
for the $L^2$-operator norm.  
 Furthermore, we claim the  bound 
 \begin{align} \label{MdelKinv-bnd'}	\|\nabla^k\ell^{-1} \Proj  M_\delta \bar K_\delta^{-1}  \| \ls  \z^{2-k} m^{1/2} .
\end{align}
 Indeed, decomposing $M_\delta$ according to   \eqref{Mdelta-deco}-\eqref{Mdel''-bnd} of Proposition \ref{prop:Mdel-deco} and using bound  Eqs \eqref{n-k-ell-1-est}, 
we find 
\begin{align} \notag 
\|\nabla^k &\ell^{-1} \Proj  M_\delta  \n^{-1}\bar \Proj\|\\
 \notag & \leq  \|\nabla^k\ell^{-1} \Proj  M_\del' \Proj \| + \|\nabla^k\ell^{-1} \n\Proj  \n^{-1} M_\delta''  \n^{-1}\bar \Proj \|\\
\label{Mdel-bnd''}	&\ls   \z^{2-k}\delta^{-1} m^{1/2} + \z^{1-k} , 
\end{align}
where $\|\cdot \|$ is the operator norm in $L^2$. Since $\z := \delta  m^{- 1/2}$, this implies 
\begin{align} \label{Mdel-bnd'}	\|\nabla^k\ell^{-1} \Proj  M_\delta \n^{-1}\bar \Proj f\|_{L^2} \ls a^{-1} \z^{1-k}   \| f\|_{L^2}.
\end{align}
 
Eqs  \eqref{Kinv-est}, with $k=1, i=0$,   and \eqref{Mdel-bnd'}, together with the insertion of $\one =\n^{-1}  \n=\Delta^{-1}\n\n$ between $M_\delta$ and $\bar K_\delta^{-1}$, imply \eqref{MdelKinv-bnd'}.

Using  \eqref{n-k-ell-1-est} and \eqref{MdelKinv-bnd'} 
 and recalling the definition $Q := \Proj -\Proj M_\delta \bar K_\delta^{-1} $ (see \eqref{Q-def}), we find that 
\begin{align}
 \label{nonlinear-est-3'}		& \|\n^k\ell^{-1}Q\|  \ls \z^{2-k}, \ k=0, 1, 2,
\end{align}
which, due to the definition of the norm $\|f\|_{ \delta} $ $\simeq \sum_0^1\z^{k-1} \|\n^k\varphi\|_{L^2}$ 
  in \eqref{s-norm-def}, implies Lemma \ref{lem:ell-inv-Q-est}. 
\end{proof}

 Lemma \ref{lem:ell-inv-Q-est} and nonlinear estimate \eqref{Ndel-est-rough}, 
together with  $\z \delta^{- 1/2}=m^{- 1/2} \delta^{1/2}$, 
 imply that, under Assumption \ref{A:diel}, 
 \begin{align} \notag 
	&\|\ell^{-1}Q[N_\delta(\varphi) - N_\delta(\psi)]\|_{ \delta}\\ 
 \label{nonlin-est-s1'}	& \qquad  \qquad \ls 
 m^{- 5/6}\delta^{1/2}(\|\varphi\|_{ \delta} + \|  \psi\|_{ \delta})\|\varphi - \psi\|_{ \delta}. 
\end{align}

 Eq. \eqref{phis-fp'}, 
  Lemma 
 \ref{lem:vphil-Lip-est}   and 
  estimate \eqref{nonlin-est-s1'} 
  imply, for 
 $\varphi, \psi \in B_{s,\delta}$, 
\begin{align} 
 \label{Phis-est-1} &\|\Phi_s(\varphi_s) \|_{\delta} \ls 
 \z\|\kp\|_{L^{2}}    + m^{- 5/6} \delta^{1/2} c_s^2,\\ 
&\|\Phi_s(\varphi_s) - \Phi_s(\varphi_s')\|_{\delta} \ls 
 m^{- 5/6}\delta^{1/2} c_s\|\varphi_s - \varphi_s'\|_{\delta}. \label{Phis-est-2}
\end{align}
These inequalities, together with the inequality $ m^{5/6}  \delta^{- 1/2}$ 
  $=\theta^{-3/2}\z \gg c_s\gg \z$, which follows from assumption  \eqref{cs-cl-cond}, 
 yield that $\Phi_s(\varphi_s)$ is a contraction on $B_{s,\delta}$ and therefore has a unique fixed point.  This proves Proposition \ref{prop:vphis-fp-exist}. 
 \qquad \qquad  \qquad \qquad  \qquad \qquad  \qquad  \qquad \qquad   \qquad \qquad $\Box$

\subsection{
Proof of Proposition \ref{prop:psi1-est} }\label{sec:psi1-fp}

In view of Proposition \ref{pro:ell-expansion-full}, we write
\begin{align}
	&\ell = \ell_0 + \ell', \label{eqn:ell-decomposition-ell-0-1-2} 
\end{align}
where $ \ell_0$ is defined \eqref{ell0-def}, 
 and $\ell'$ is defined by this expression. 
By  Proposition \ref{pro:ell-expansion-full}, $\ell'=O(\delta^2 (-i\n)^4)$ on the range of $\Proj$.

 Inserting \eqref{eqn:varphis-decomposition} into equation \eqref{phis-eq} and 
 using \eqref{eqn:ell-decomposition-ell-0-1-2} and the relations $\ell  \Proj \psi= \Proj \kp+ \ell'  \Proj \psi$ and $- \Proj \kp+ Q \kp = - \Proj M_\delta \bar K_\delta^{-1} \kp$, we obtain the equivalent equation for $\psi_1$: 
\begin{align} \label{eqn:phi1-eqn}	&\ell \psi_1 =-\ell'\Proj \psi  
- \Proj M_\delta \bar K_\delta^{-1} \kp + QN_\delta(\tilde\varphi) \, ,\\
	&\tilde \varphi = \tilde \varphi(\psi_1) := \Proj \psi + \psi_1 + \varphi_l(\Proj \psi+\psi_1), \label{eqn:tilde-varphi-def}
\end{align}
with $\varphi_l = \varphi_l(f)$  the solution to equation \eqref{eqn:barP-eqn} given by Proposition \ref{pro:phir-fixed-pt} with $\varphi_s$ replaced by $f \in B_{s,\delta}$.

By Proposition \ref{pro:ell-expansion-full}, the operators $\ell$ and $\ell_0$ are invertible. 
We invert $\ell_0$ (see \eqref{ell0-def}) in \eqref{eqn:phi0-def} to obtain
\begin{align}
	\psi := \ell_0^{-1}\kp \label{eqn:phi1-eqn-inverse} \, .
\end{align}
Furthermore, we invert $\ell$ (see \eqref{ell-def}) 
  in equation \eqref{eqn:phi1-eqn} and use \eqref{eqn:phi1-eqn-inverse} to find
\begin{align}
 \label{psi1-fp}	\psi_1 =& \Phi_1(\psi_1) = \Phi_{1}' + \Phi_1''(\tilde \varphi) \, ,
\end{align}
where $\tilde \varphi$ is  given in \eqref{eqn:tilde-varphi-def},  and, with $Q$  given in  \eqref{Q-def},  
\begin{align} \label{Phi-k'}	
&\Phi_{1}' := -\ell^{-1} [\ell' \ell_0^{-1}\kp  +\Proj M_\delta \bar K_\delta^{-1} \kp],\\
 \label{Phi1-nonl} 
 & \Phi_1''(\tilde \varphi):= \ell^{-1}QN(\tilde \varphi).
\end{align}
  \eqref{psi1-fp} is a fixed point equation for $\psi_1$. However, we do not have to solve it since we have already proved the existence of $\psi_1$. We use \eqref{psi1-fp} to estimate $\psi_1$.

Next,  by Proposition \ref{pro:ell-expansion-full}, Eq. \eqref{nu-est'} and the relation $O(\delta^2 (-i\nabla)^4)=O(a^2 (-i\nabla)^2)$ on $\Ran \Proj$ (see the definition of $\Proj$ in  \eqref{eqn:choice-of-P}), valid due to the choice $a := \delta r = O(1)$ 
  (see  \eqref{a-m}), we have that
the operators $\ell, \ell_0, \ell'$ given in \eqref{ell-def}, \eqref{ell0-def}, and \eqref{eqn:ell-decomposition-ell-0-1-2}, respectively, satisfy
\begin{align}
&	|\ell'| \ls  \delta^2 (-i\n)^4 , \label{ell'-est}\\
&	\ell_0 \gs  -\Delta + \z^{-2}, \label{ell0-lower-bnd}\\
&	\ell \gs  -\Delta + \z^{-2}, \label{ell-lower-bnd}
\end{align} 
where, recall,  $\z:=  m^{- 1/2}  \delta$, with $m$ given in \eqref{m}, (cf. \eqref{n-k-ell-1-est}). 

Using \eqref{ell'-est} - \eqref{ell-lower-bnd} 
 and the fact $\ell, \ell_0, \ell'$ are self-adjoint and are functions of $-i\n$ and therefore mutually commute, 
and using  \eqref{MdelKinv-bnd'}, 
we find that  
\begin{align}
\label{eqn:Phikp-L2}	&\| \Phi_{1}' \|_{\dot H^i} \ls  \del^{ 2}\|\kp\|_{H^i}+    m^{1/2}\z^{2-i}\|\kp\|_{L^2},\ i=0,1.
\end{align}
  By the choice of the $B_{s,\delta}$ norm (see \eqref{s-norm-def}) and since $\del^{ 2}\ll m^{1/2}\z^{2-i}$ (by  \eqref{cs-cl-cond}, or \eqref{beta-del-cond}), we see that  
  \begin{align} \label{Phi-ka'-est} 
	\|\Phi_{1}'\|_{\delta} \ls   m^{1/2} \z \|\kp\|_{L^2}. 
\end{align}

Now, we turn our attention to the map $\Phi(\tilde \varphi):=\ell^{-1}QN(\tilde \varphi)$ (see \eqref{Phi1-nonl}).  The definition of $\Phi(\tilde \varphi)$ 
   and Eq. \eqref{nonlin-est-s1'} give 
   \begin{align} \label{Phi-est} 
	\|\Phi_{1}''(\tilde \varphi)\|_{\del} \ls 
	 m^{- 5/6}\delta^{1/2}\|\tilde \varphi\|_{\delta}^2. 
\end{align}
Next, we estimate $\tilde \varphi = \tilde \varphi(\psi_1) := \Proj \psi + \psi_1 + \varphi_l(\Proj \psi+\psi_1)=\varphi_s+\varphi_l$ (see \eqref{eqn:tilde-varphi-def}).   
 By Propositions \ref{pro:phir-fixed-pt} and \ref{prop:vphis-fp-exist},  $\|\varphi_l\|_{ \delta} \le c_l$ 
  and $\| \Proj \psi + \psi_1\|_{ \delta} \ls c_s$ and therefore 
 $\|\tilde \varphi\|_{ \delta} \ls c_s+c_l$. This, together with \eqref{Phi-est} and condition \eqref{cs-cl-cond} and inequality \eqref{s-Hdot-est}, yields
\begin{align} \label{Phi-est'} 
	\|\Phi_{1}''(\tilde \varphi)\|_{ \delta} &\ls   
	 m^{- 5/6}\delta^{ 1/2}c_l^2. 
\end{align}
 Equations \eqref{Phi-ka'-est} and \eqref{Phi-est'} and the relation $ m^{- 5/6}\delta^{ 1/2}=\theta^{3/2}\z^{- 1}$ imply  
\begin{align} \label{Phi1-est'} 
	\|\Phi_1(\psi_1)\|_{ \delta} &\ls  m^{1/2} \z \| \ka'\|_{L^{2}}+ \theta^{3/2}\z^{- 1} c_l^2. 
\end{align}
By condition  \eqref{cs-cl-cond} 
 and 
 our choice $c_l =\om\z,\ \om:=\min(\theta^{- 1/2}, m^{-1/4})$, we see that 
   \eqref{Phi1-est'} implies 
\begin{align} \notag
	\|\Phi_1(\psi_1)\|_{ \delta} 	& \ls   
	  (a^{-2} m^{1/2}+ \theta^{1/2})\z.
\end{align}
 Since $\psi_1=\Phi_1(\psi_1)$ and, the above estimate 
 gives \eqref{psi1-est},  proving  Proposition \ref{prop:psi1-est}. 
\quad \quad \quad $\Box$

\section{Analysis of the operator $\ell$. Proof of Proposition \ref{pro:ell-expansion-full}}   \label{sec:oper-ell}

The goal of this section is to prove Proposition \ref{pro:ell-expansion-full}. 
The proof follows readily from Lemmas \ref{lem:ell-b-rel}, \ref{lem:b-expan} and \ref{lem:dielectric-const-positive} below.  Throughout this section, we suppose 
Assumption \ref{A:diel} 
holds, without mentioning this explicitly.  

Let  $\Mfib{k}$ and $\barKfib{k}$ be the $k$-th Bloch-Floquet fibers of $M\equiv M_{\delta=1}$ and $\bar K \equiv \bar K_{\delta=1}$ (see \eqref{BF-deco-opr}, {\it not to be confused} with $M_{\delta}$ and $\bar K_{\delta}$). 
Since, by Proposition \ref{pro:Mdel-prop}, $\bar K$ is invertible, then so is $\barKfib{k}$ and $\barKfib{k}^{-1}=(\bar K^{-1})_k$ (see \eqref{AB-fib}).
 We have
\begin{lemma}\label{lem:ell-b-rel} The operator $\ell$, defined in \eqref{ell-def}, is of the form    
\begin{align}\label{ell-b-rel}	 \ell =& \delta^{-2} b (-i\delta \nabla) \Proj, \end{align}
where $b(k)$ is a smooth, even function of $-i\delta \n$ given  explicitly as:
\begin{align}  \label{b-expl}
	b (k) =& |\Om|^{-1} \big\lan 1, (|k|^2+\Mfib{k}-\Mfib{k} \barKfib{k}^{-1} \Mfib{k}) 1 \big\ran_{L^2_{\rm per}}. \end{align}
 \end{lemma} 
\begin{proof} Since $M_\delta$ is $\latde$-periodic by  Proposition \ref{pro:Mdel-prop}, Eq \eqref{ell-def} implies that so is $\ell$. Moreover, \eqref{ell-def}  and \eqref{Mdelta-M-relation'} yield
\begin{align}
\ell =& P_r [\delta^{-2}U_\delta (-\Delta + M) U_\delta^*] P_r  
  - P_r [\delta^{-2}U_\delta M U_\delta^*] \notag \\
&\times\bar P_r [\delta^{-2}U_\delta (-\Delta + M) U_\delta^*]^{-1} \bar P_r [\delta^{-2}U_\delta M U_\delta^*] P_r,  \notag   
\end{align}
where, recall, $M\equiv M_\del\big|_{\del = 1}$, which implies that
\begin{align} \label{ell-ell-del1}
	\ell  = \delta^{-2}\Udel \ell\big|_{\del = 1}\Udel^* .\end{align}
The last two properties and Lemma \ref{lem:PAP-repr-del} show that $\ell$ is a function of $-i\delta \n$ of the form \eqref{ell-b-rel},    
where $b(k) = \lan (\ell |_{\delta =1} )_k 1 \ran_\Om$, with  $(\ell\mid_{\delta =1} )_k$ being the Bloch-Floquet fibres of $\ell\mid_{\delta =1} $ and $1$ standing for the constant function, $1 \in L^2_{\rm per}(\R^3)$.   Using equations \eqref{AB-fib}, \eqref{ell-def} and $\Delta_k\one =0$, we find explicit form \eqref{b-expl} of $b(k)$.

 The next proposition gives the Bloch-Floquet decomposition of the operator $M$.
\begin{proposition} \label{prop:M-BF-deco} 
The operator $M$ has a Bloch-Floquet decomposition \eqref{BF-deco-opr} whose $k-$fiber, $M_{k}$, acting on $\Lper$ is given by 
\begin{align} \label{exact-form-Mk}
	M_{k} f =&  -  \den\left[ \oint r_{\rm per, 0}(z)f r_{{\rm per}, k}(z)\right]
\end{align}
where 
 $f \in \Lper$ and, on $\Lper$,
\begin{align}
	&r_{\rm per, k}(z) = (z- h_{\rm per, k})^{-1}, \label{eqn:rperk-adhoc} \\
	&h_{\rm per, k} = (-i\nabla- k)^2 - \phi_{\rm per} \label{eqn:hperk-adhoc} \, .
\end{align}
\end{proposition}
\begin{proof}[Proof of Proposition \ref{prop:M-BF-deco}]
Let $T_s$ be given in \eqref{eqn:translation-op-def} and $\varphi \in L^2 $. To compute $k$-fibers of  $M$,  we note $T_{-t}\den [A] = \den\left[T_t^* A T_t \right] $ and $[ T_t,  r_{\rm per}(z)]=0$ for all $t \in \lat$. Using these relations, the definition of the Bloch-Floquet decomposition \eqref{fk} and equation \eqref{Mexpl}, we obtain
\begin{align}
	 (M \varphi)_k(x) & 
	 = - \sum_{t \in \lat} e^{- ik(x+t)} \oint T_{-t}\den\left[ r_{\rm per}(z)  \varphi r_{\rm per}(z) \right] \notag\\
		= &- \sum_{t \in \lat} e^{- ik(x+t)}\oint \den\left[T_t^* r_{\rm per}(z) \varphi r_{\rm per}(z) T_t \right].  \label{eqn:M-delta-varphi-k-1}
\end{align}
Since $r_{\rm per}(z)$ is $\lat$-periodic, \eqref{eqn:M-delta-varphi-k-1} shows
\begin{align}
	(M \varphi)_k & (x) \notag\\
		  &= - \sum_{t \in \lat} e^{-2\pi ik(x+t)}\oint \den\left[r_{\rm per}(z) (T_{-t}\varphi) r_{\rm per}(z) \right] . \label{eqn:M-delta-varphi-k-2}
\end{align}
Using that $\den[A]f = \den[Af] = \den[fA]$ for any operator $A$ on $L^2(\R^3)$ and any sufficiently regular function $f$ on $\R^3$, we insert the constant factor of $e^{-ikt}$ into $\den$ in \eqref{eqn:M-delta-varphi-k-2}. We obtain
\begin{align}
		  (M &\varphi)_k(x) \notag\\
		  &= -  e^{-ikx} \oint \den\left[r_{\rm per}(z)  \sum_{t \in \lat} e^{- ikt} (T_{-t}\varphi) r_{\rm per}(z) \right]. 
\end{align}
This and the definition of the Bloch-Floquet decomposition of $\varphi$, \eqref{fk}, imply
\begin{align}
(M \varphi)_k(x) 	=&  -  \oint \den\left[r_{\rm per}(z)  \varphi_k e^{ ikx} r_{\rm per}(z) e^{-ikx}  \right]. 
\end{align}
 Since $e^{ ikx} (-i\nabla) e^{- ikx} = -i\nabla - k$, and therefore $e^{ i kx}r_{\rm per}(z)e^{- ikx} = r_{{\rm per}, k}(z)$, this gives \eqref{exact-form-Mk}. 
\end{proof}

Since the resolvents $r_{\rm per, k}(z)$ 
 are smooth in $k$ (see \eqref{eqn:rperk-adhoc} - \eqref{eqn:hperk-adhoc}), then, by \eqref{exact-form-Mk}, $\Mfib{k}$ is also smooth in $k$. Hence, by \eqref{b-expl}, $b(k)$ is smooth in $k$.
 
   Since the operator $\Mfib{k}-\Mfib{k} \barKfib{k}^{-1} \Mfib{k}$ in \eqref{b-expl} is self-adjoint, the function $b(k)$ is real.  By Lemma \ref{prop:M-BF-deco} and the properties $\bar r_{\rm per, k}(z) := \cC r_{\rm per, k}(z)\cC$ $ = r_{\rm per, - k}(\bar z)$, where $\cC$ is the complex conjugation, and the contour of integration in \eqref{exact-form-Mk} is symmetric w.r.to the reflection $z\ra \bar z$, we have $b(k)=\bar b(k)=b(- k)$, i.e. $b(k)$ is even. 
 \end{proof}

Let 
$ \Kfib{0} = \Kfib{k=0}$ denote the $0$-fiber of $K$, acting on $L^2_{\rm per}(\R^3)$. We also let $\Pi_0$ denote the projection onto constant functions on $L^2_{\rm per}(\R^3)$ and $\bar \Pi_0 := 1- \Pi_0$. Finally, we define 
\begin{align}
	\barKfib{0} :=  \bar \Pi_0 \Kfib{0} \bar \Pi_0. \label{eqn:bar-K-0-def}
\end{align}
Recall  the abbreviation $\cbeta:=\beta e^{- \etal\beta}$. With this notation, we have

\begin{lemma}\label{lem:b-expan} Let  $m$  
be given in \eqref{m}.  
 The function $b(k)$ given in \eqref{b-expl} satisfies
  \begin{align} 
\notag		 b(k)  = |\Om|^{-1} (m +  O(\cbeta^2)) 
		&+ k \cdot \epsilon  k \\
\label{b-expan}& + k \cdot O(\cbeta)  k +O(|k|^4), \end{align}
   where $m$ is a scalar given by \eqref{m} and $ \epsilon$ is a real matrix given by \eqref{eps}-\eqref{eps''}.
 \end{lemma} 
 \begin{proof}[Proof of Lemma \ref{lem:b-expan}]  
  First, we use \eqref{b-expl} to write $b(k)$ as 
\begin{align}
\label{b=b1-b2}	&b (k)  = b_1(k) - b_2(k),\\
  \label{b1-def}	&b_1(k) := |k|^2+|\Om|^{-1} \left\lan 1, \Mfib{k} 1 \right\ran_{L^2_{\rm per}},\\ 
\label{b2-def}	&b_2(k) :=  |\Om|^{-1} \lan 1, \Mfib{k} \barKfib{k}^{-1} \Mfib{k} 1\ran_{L^2_{\rm per}}. 
\end{align}
We begin with $b_1(k)$. We claim that 
\begin{align}
	b_1(k) =& |\Om|^{-1} m 
	+ \epsilon'_1 |k|^2 + k \cdot (1+\epsilon') k + O(|k|^4), \label{eqn:bk-b0-explicit-formula}
\end{align}
 where $m$ and $\epsilon'$ are given in \eqref{m} and \eqref{eps'}, respectively, and $\epsilon'_1$  is a real, symmetric matrix satisfying 
	$\epsilon'_1 = O(\cbeta), $ 
contributing to  the third term on the r.h.s. of \eqref{b-expan}. 

Using definition of $b_1$ in \eqref{b1-def} and Proposition \ref{prop:M-BF-deco}, 
 we see that 
\begin{align}
	b_1(k) &= -  |\Om|^{-1}\lan 1, \den \oint \rpero(z) 1 r_{{\rm per}, k}(z) \ran_{L^2(\Om)}  \\
\label{b1-form'}		&= - |\Om|^{-1} \Tr_{\Om}  \oint \rpero(z)  r_{{\rm per}, k}(z),  
\end{align}
where $1$ is the constant function $1 \in \Lper$ and $\Om$ is an arbitrary fundamental cell of $\lat$. 
 To begin with, using the Cauchy-formula for derivatives, we obtain
\begin{align}\label{eqn:b10-explicit-formula}	b_1(0) =& - |\Om|^{-1} \Tr_{\Om}   \oint \rpero^2(z) =|\Om|^{-1} m.
\end{align}
Next, recall that $h_{{\rm per}, k} = (-i\n - k)^2 - \phi_{\rm per}$ (see \eqref{eqn:hperk-adhoc}). We have, by the resolvent identity, that 
\begin{align}\label{rper-res-form}\rperk(z) & - \rpero(z) =	\rperk(z)   [2(-i\nabla) \cdot k -   |k|^2] \rpero(z). \end{align}
Applying this identity to \eqref{b1-form'} and using that $b_1(k)$ is even, we obtain \eqref{eqn:bk-b0-explicit-formula}, with 
  $\epsilon'_1:= -  |\Om|^{-1} \Tr_{L^2_{\rm per}} \oint  \rpero^3(z)$ and $\epsilon'$ given by \eqref{eps'}.
Using the Cauchy-integral formula, we rewrite  $\epsilon'_1$  as
\begin{align}
	\epsilon'_1 =& -\frac12  \Tr_{L^2_{\rm per}}  f''_{T}(\hper - \mu).
\end{align}
Then following the proof of Lemma \ref{lem:Vp-uppper-bound} with $f'_{\rm FD}$ replaced by $f''_{\rm FD}$, we show that $\epsilon'_1 =  O(\cbeta)$. This proves \eqref{eqn:bk-b0-explicit-formula}.

Next, we prove the expansion
\begin{align}\label{eqn:b2-claim}	b_2(k) = k \cdot \epsilon'' k + k \cdot \epsilon_1'' k  +O(|k|^4) +  O(\cbeta^2), 
\end{align}
  where $\epsilon''$ is given in \eqref{eps''}, $\epsilon_1''$ is a real matrix satisfying $\epsilon_1'' = O(\cbeta)$ (contributing to  the third term on the r.h.s. of \eqref{b-expan}). 
First,  
we recall from \eqref{b2-def} 
\begin{align}
	b_2(k) =& |\Om|^{-1}\lan 1, (M \bar K^{-1} M)_k 1 \ran_{L^2_{\rm per}} \\
		=& |\Om|^{-1}\lan \Mfib{k} 1, \barKfib{k}^{-1} \Mfib{k} 1 \ran_{L^2_{\rm per}} , \label{eqn:b2-def}
\end{align}
where, recall, $\Mfib{k}, \Kfib{k}$ and $\barKfib{k}$ are the $k$-th Bloch-Floquet fiber of $M\equiv M_{\delta=1}, K\equiv K_{\delta=1}$ and $\bar K\equiv \bar K_{\delta=1}$.
Letting  
\begin{align}
\label{rhop}	\rho_k = (\bar P_a)_k \Mfib{k} 1 \in \Lper\, ,
\end{align}
where $\Lper$ is given in \eqref{L2-per}, 
 $a$ is given in \eqref{a-m}, and $P_a$ is defined in \eqref{eqn:choice-of-P}, we find
\begin{align}\label{b2-repr}	b_2(k) =  |\Om|^{-1} \lan \rho_k, \barKfib{k}^{-1} \rho_k \ran_{L^2_{\rm per}} \, .
\end{align}

Now, we expand $\rho_k$ in $k$. By \eqref{exact-form-Mk}, we have $M_{k=0} \one $ $=  -  \den\left[ \oint r_{\rm per, 0}^2(z)\right]$. Next, recall $\oint := \frac{1}{2\pi i} \int_\Gamma dz f_{T}(z-\mu)$  (see \eqref{eqn:oint-adhoc}) to obtain 
\begin{align} \label{M01}M_{k=0} \one =- \den  f_{T}'(h_{{\rm per}, 0}-\mu).\end{align} Since $f_{FD}'\le 0$, we have $M_{k=0} \one>0$. Introduce the function 
\begin{align} \label{V}	&V(x) = -  \den\left[f_{T}'(h_{{\rm per}, 0}-\mu) \right](x) \geq 0. 
\end{align}
By definition \eqref{rhop} and Eqs \eqref{exact-form-Mk}, \eqref{M01} and  \eqref{V}, we have 
\begin{align}
	&\rho_k = (\bar P_a)_k V + \rho'_k, \label{eqn:Ustar-rhop-expanded}  \\
	&\rho'_k :=   (\bar P_a)_k \den \oint \rpero^2(z)(2(-i\n) k + k^2)  r_{{\rm per}, k}(z).  \label{eqn:rho-prime-k}
\end{align}
Inserting the decomposition \eqref{eqn:Ustar-rhop-expanded} 
into \eqref{b2-repr} gives
\begin{align}
	|\Om|b_2(k) =   & \lan V, \barKfib{k}^{-1} V \ran + 2\Re\lan V, \barKfib{k}^{-1} \rho_{k}' \ran + \lan \rho_k', \barKfib{k}^{-1} \rho_{k}' \ran. \label{b2-lowest-order}
\end{align}
We expand the third term on the r.h.s. on \eqref{b2-lowest-order}. 
First, we give a rough bound.  
  For $z\in \Gamma$, we claim the estimates
\begin{align}\label{Del-r-est'}\|(1-\Delta)^{\al}  r_{{\rm per}, k}&(z)\|\notag\\
\le\ & \|(1-\Delta)^{\al}  r_{{\rm per}}(z)\|\ls\ d^{\al-1}\ls 1,\end{align}
for $\al=0, 1/2$, where 
$d \equiv d(z):=\dist (z, \s(h_{\rm per}))\ge \frac14 .$ 
The first estimate follows from \eqref{A-Ak-norm}. The second estimate is straightforward for $\al=0, 1$, which by interpolation, gives it for all $\al \in [0, 1]$. For $\al= 1/2$, it can be also proven directly as $\|(-\Delta)^{1/2}f\|^2=\lan f, (h_{{\rm per}}-z+\phi_{\rm per}+z) f\ran$. Taking $f = r_{{\rm per}}(z)u$, we arrive at the second estimate in \eqref{Del-r-est'} for $\al=1/2$.

By the second resolvent identity \eqref{rper-res-form} and estimates \eqref{Del-r-est'}, we have the expansion 
 \[r_{{\rm per}, k}(z)=r_{{\rm per}, 0}(z)+O(|k| d^{-3/2} +|k|^2 d^{-2}).\]
Using this expansion in \eqref{eqn:rho-prime-k}, we find 
\[\rho_{k}' =\rho' \cdot k +O(|k|^2),\]
where $\rho'$ is given in \eqref{rho'}.  
Using  the latter relation,  the relation $\barKfib{k}^{-1}=\barKfib{0}^{-1}+O(k)$ 
  and the fact that, since on $L^2_{\rm per}$ the spectrum of $-i\n$ is discrete,   $(\bar P_a)_{k=0} = (\bar P_{a=0})_{k=0}$ for $a$ is sufficiently small, we obtain 
\begin{align} \label{b2-lowest-order'}	|\Om|^{-1}\lan \rho_k', \barKfib{k}^{-1} \rho_k' \ran = - k \epsilon'' k + O(k^4),
\end{align}
 for 
 $\epsilon''$ is given in \eqref{eps''}, where the power of the remainder comes from the fact $b_2(k)$ is even which is shown by the same argument that was used in demonstration that  $b(k)$ is even.
Eqs. \eqref{b2-lowest-order} and \eqref{b2-lowest-order'} show that 
\begin{align}
	&b_2(k) = b_2(0) - 
	 k \epsilon'' k + O(k^4) + \text{Rem}, \label{eqn:b2-higher-order} \\
	&	\text{Rem} := \lan V, [\barKfib{k}^{-1} - \bar K^{-1}_0]V \ran + 2 \Re\lan V, \barKfib{k}^{-1} \rho_{ k}' \ran ,  \label{eqn:b2-higher-order-rem}
\end{align}
with $b_2(0):=|\Om|^{-1}\lan V, \barKfib{k}^{-1} V \ran$. To estimate $b_2(0)$ and the terms in \eqref{eqn:b2-higher-order-rem}, we use  Eq. \eqref{Kinv-est}  
   and the relation $\|\bar K^{-1}\|=\sup_k \|\barKfib{k}^{-1}\|$ (see \eqref{A-Ak-norm}) to obtain
\[\|\barKfib{k}^{-1}\|\ls 1.\] 
We use this bound,  
Lemma \ref{lem:Vp-uppper-bound}, \eqref{eqn:rho-prime-k} and the fact that $b_2(k)$ is even in $k$, to obtain 
\begin{align}
|\Om| |\text{Rem}| \ls & (\|V\|_{L^2_{\rm per}}^2 + \|V\|_{L^2_{\rm per}})|k|^2 
 \notag \\
 \label{eqn:b2-higher-order-def-est}	= &  O(\cbeta |k|^2),\\ 
  \label{b20-est}|\Om|  b_2(0) = & O(\|V\|_{L^2_{\rm per}}^2) = O(\cbeta^2).\end{align}
We identify  the first, third and fourth terms on the r.h.s. of \eqref{eqn:b2-higher-order} 
with the fourth,  third and second terms in \eqref{eqn:b2-claim}, respectively. 
 Equations \eqref{eqn:b2-higher-order} - \eqref{eqn:b2-higher-order-def-est} imply \eqref{eqn:b2-claim}.  

 Equations \eqref{b=b1-b2}, 
  \eqref{eqn:bk-b0-explicit-formula}, and \eqref{eqn:b2-claim} yield equation \eqref{b-expan}, with $
  \e_1'+\e_1''$ making up  the third term on the r.h.s. of \eqref{b-expan}. 
  This completes the proof of Lemma \ref{lem:b-expan}. 
  \end{proof}

\begin{lemma} \label{lem:dielectric-const-positive} The $3 \times 3$ matrix $\epsilon$ entering \eqref{b-expl} 
 is symmetric and satisfies 
\begin{align}\label{eps-ineq} 
\epsilon  \ge \one - O(\cbeta^2) .	\end{align}
 \end{lemma} 
\begin{proof} 
  We  prove this lemma  using the Feshbach-Schur map. 
  Let $P = P_s$ (see \eqref{eqn:choice-of-P}) for some real number $s > 0$, unrelated to $r$  and satisfying $B(\del s)\subset \Om^*$. 
  For any projection $P$ and operator $H$ on $L^2(\R^3)$, the Feshbach-Schur map $F_P(H)$ is defined as
\begin{align}
	 F_P(H) := P H P - PH \bar P \bar H^{-1}  \bar P H P. \label{eqn:def-of-Feshbach-map}
\end{align}
where $\bar P = 1-P$, $\bar H = \bar P H \bar P$, and $\bar H^{-1}$ is defined on the range of $\bar P$. The  Feshbach-Schur map has the property \cite{GS}
\begin{align}
	-\lambda \notin  \sigma(H) \iff -\lambda \notin \sigma(F_P(H+\lambda) - \lambda P). \label{eqn:feshbach-property}
\end{align}
for any $\lambda \geq 0$. That is, for all $\lam > 0$,
\begin{align}
	H \geq 0 \iff  F_P(H+\lambda) - \lambda P \geq 0. \label{eqn:feshbach-property-1}
\end{align}

With the Laplacian $\Delta$, we define
\begin{align}
	K_{c, \del} = K_\del + c \Delta. \label{eqn:K-e-del}
\end{align}
Since $M_\del > 0$ by Proposition \ref{pro:M-expl}, we have that $K_{c, \del} > 0$ for all $c \in [0,1)$. 
Consequently, \eqref{eqn:feshbach-property-1} shows that, for any $\lam > 0$, 
\begin{align} 
	F_P(K_{c, \del} + \lam) - \lam P \geq 0. \label{eqn:feshbach-spec-condition}
\end{align}

 Using definition \eqref{eqn:def-of-Feshbach-map} and the resolvent identity, we obtain 
\begin{align}
	 F_P&(K_{c, \del} + \lam) - \lam P  \label{eqn:FK-feshbach-map-term-1}\\
	&= PK_{c, \del} P  -  PM_\del (\bar K_{c, \del} + \lam \bar P)^{-1} M_\del P \\
	&= F_P(K_{c, \del}) +  \lam PM_\del \bar K_{c, \del}^{-1} (\bar K_{c, \del} + \lam \bar P)^{-1} M_\del P  \\ 
	&= F_P(K_{c, \del}) + \lam PM_\del \bar K_{c, \del}^{-2} M_\del P \notag \\
	&- \lam^2 PM_\del (\bar K_{c, \del})^{-2}(\bar K_{c, \del} + \lam \bar P)^{-1} M_\del P. \label{eqn:FK-feshbach-map-term-3}
\end{align}
By the choice of $P = P_s$ (see \eqref{eqn:choice-of-P}), we see that $\bar K_{c, \del} \gs s^2$. 
 Since $M_\delta \ls \delta^{-2}$, 
we see that the last term in \eqref{eqn:FK-feshbach-map-term-3} is bounded by $O(\lambda^2 \delta^{-4} s^{-6})$. Thus, \eqref{eqn:FK-feshbach-map-term-1} - \eqref{eqn:FK-feshbach-map-term-3} implies
\begin{align}
 \label{FK-exp1}	 F_P(K_{c, \del} + \lam) - \lam P&= F_P(K_{c, \del}) + \lam W+ O(\lambda^2 \delta^{-4} s^{-6})P, 
\end{align}	
  where $W:=PM_\del (\bar K_{c, \del})^{-2} M_\del P$. To estimate $W$,   we proceed as in the proof of Lemma \ref{lem:ell-b-rel}. 
 First,  since $M_\delta$ is $\latde$-periodic by  Proposition \ref{pro:Mdel-prop},   $W$ is $\latd$-periodic.
  Moreover, the definition  $W:=PM_\del (\bar K_{c, \del})^{-2} M_\del P$ and \eqref{Mdelta-M-relation'} yield 
	\begin{align}
			W =& P_s (\delta^{-2}U_\delta M_1 U_\delta^*)  \bar P_s [\delta^{-2}U_\delta (-\Delta + M_1) U_\delta^*]^{-2} \bar P_s 
			 (\delta^{-2}U_\delta M_1 U_\delta^*) P_s \ \notag \\				=& P_{\delta s} M_1 \bar P_{\delta s} \bar K_1^{-1} \bar P_{\delta s} M_1 P_{\delta s}, \notag 
		\end{align}  
which implies that
\begin{align} \label{W-W-del1}
	W  = \Udel W\big|_{\del = 1}\Udel^* .\end{align}
($(\bar K_{c, \del})^{-1}$ entering $W$ in the second power eats up $\delta^{-2}$ compared to \eqref{ell-ell-del1}.) Since $B(\del s)\subset \Om^*$, the last two properties and Lemma \ref{lem:PAP-repr-del} show that $W$ is a function of $-i\delta \n$ of the form   
\begin{align}\label{W-w-rel}	 W =& 
  w (-i\delta \nabla) P, 
\end{align}
where $w(k) = \lan W_k 1 \ran_\Om$, with  $W_k$ being the Bloch-Floquet fibers of $W $ and $1$ standing for the constant function, $1 \in L^2_{\rm per}(\R^3)$. 
 Using equation \eqref{AB-fib}, we find,  as in \eqref{ell-b-rel} - \eqref{b-expl},  the explicit form of $w(k)$:
\begin{align}  \label{w-expl}
	w (k) =& |\Om|^{-1} \big\lan 1, \Mfib{k} \barKfib{k}^{-2} \Mfib{k} 1 \big\ran_{L^2_{\rm per}}, \end{align}
where $\Mfib{k}$ and $\barKfib{k}$ are the $k$-th Bloch-Floquet fibres of $M_{\delta=1}$ and $\bar K_{\delta=1}$.

   Since the operator $\Mfib{k} \barKfib{k}^{-2} \Mfib{k}$ in \eqref{b-expl} is self-adjoint, the function $w(k)$ is real.  Arguing as with $b(k)$ in the proof of  Lemma \ref{lem:ell-b-rel}, 
     we conclude  that  $w(k)$ is even and smooth. 
    Furthermore, as  with $b_2(k)$ in the proof of Lemma \ref{lem:b-expan}, 
    we expand $w(k)$ in $k$ to the fourth order to obtain 
\begin{align}
	& W 
	=O(\cbeta^2) -\delta^2 \nabla \epsilon_3 \nabla P + O(\delta^4 (-i\nabla)^4P), \label{eqn:PMK-inv-KP-feshbach-part}
\end{align}
\begin{align}\label{eps3}
	\epsilon_3 := |\Om|^{-1} \lan \rho', {\bar K_{c, 0}}^{-2} \rho'\ran_{L^2_{\rm per}} > 0,
\end{align}
 where   $\rho'$ is given in \eqref{rho'}, $K_{c, 0}\equiv K_{c, \delta=1, k= 0}$ is the $0$-th fiber of $K_{c, \delta=1}$, and $ \bar K_{c, 0} = \bar \Pi_0  K_{c, 0} \bar \Pi_0$. Here $\bar \Pi_0 = 1 - \Pi_0$ and $\Pi_0$ is the projection in $L^2_{\rm per}$ onto constants. The inverse $\bar K_{c, 0}^{-2}$ is taken on the range of $\bar \Pi_0$. 
  Eqs \eqref{FK-exp1} and \eqref{eqn:PMK-inv-KP-feshbach-part} imply that 
  \begin{align}
	 F_P(K_{c,\del} + \lam) - \lam P &= F_P(K_{c, \del}) +O(\cbeta^2)- \lambda \delta^2 \nabla \epsilon_3 \nabla P\notag\\
	& + O(\delta^4 (-i\nabla)^4 P)+O(\lam \delta^{-4} s^{-6} P). 
	\label{eqn:FK-feshbach-map-term-5}
\end{align}

Now, we use definition \eqref{eqn:K-e-del} to expand the term $F_P(K_{c, \del})$ in \eqref{eqn:FK-feshbach-map-term-5} in $c$. A simple computation shows that
\begin{align}
	F_P(K_{c, \del}) =&  F_P(K_\del) + c \Delta P \\ 
		&- \sum_{n \geq 1} c^n PM_\del  (\bar K_{\del}^{-1} (-\Delta))^n \bar K_{\del}^{-1} M_\del P. \label{eqn:FPKepsilon-expand} 
\end{align}
Since $\bar K_{\del}\ge 0$, \eqref{eqn:FPKepsilon-expand} is negative, we conclude 
\begin{align} 
F_P(K_{c, \del}) 	&\leq F_P(K_\del) + c \Delta P \label{eqn:FK-feshbach-map-term-6-1}.
\end{align}
 Since $F_P(K_\del) = \ell$ for $r=s$ (see \eqref{ell-def}), we see, by  Lemma \ref{lem:b-expan}, that 
\begin{align}  \label{eqn:FK-feshbach-map-term-7} 
	F_P(K_\del) =&  -\nabla \epsilon  \nabla P 
		+ O(\delta^2 (-i\nabla)^4 P, 
\end{align}
with $\epsilon$ defined there. We use that 
 $O(\del^4 (-i\nabla)^4 P) = O(\tilde a^2 (-i\n)^2 P)$, where $\tilde a := \delta s$ (which is unrelated to the $a$ in \eqref{a-m}) 
and  \eqref{eqn:FK-feshbach-map-term-6-1} and \eqref{eqn:FK-feshbach-map-term-7} to obtain 
\begin{align} 
F_P(K_{c, \del}) 	&\leq -\n (\epsilon  - c + O(\tilde a^2))\n P. 
 \label{eqn:oneee-2-1}
\end{align}
Setting  $\epsilon_4:=O(\tilde a^2)+ \lambda \delta^2\epsilon_3$, we see that Eqs. \eqref{eqn:FK-feshbach-map-term-5}, \eqref{eqn:feshbach-spec-condition} and 
\eqref{eqn:oneee-2-1} imply
\begin{align}
	-\n (\epsilon + \epsilon_4 - c)\n P 
	&+O(\lam \delta^{-4} s^{-6}) P+O(\cbeta^2) P \notag\\& \ge  F_P(K_{c,\del} + \lam) - \lam P \ge 0. \label{eps-est'} 
\end{align}

Inequality \eqref{eps-est'} holds for all $s \in (0, \delta^{-1})$. Taking $s=\delta^{-3/4}$, we find \[-\n \epsilon\n P \ge [(c-O(\delta^{1/2}))\Delta   -O(\lam \delta^{1/2}) -O(\cbeta^2)] P_\del,\] 
where $P_\del:=P_{s=\del^{-3/4}}$. Since this holds for every $\del>0$, since $P\equiv P_s$ converges strongly to $\one$, as $s\ra \infty$, and since the expression for $\epsilon$ given in Lemma \ref{lem:b-expan}  is independent of $\delta$, we see that \[-\n \epsilon \n  \ge - c \Delta-  O(\cbeta^2),\] 
   for every $c \in [0,1)$. 
Passing to the Fourier transform gives $\xi \cdot \epsilon \xi  \ge  c |\xi|^2 -  O(\cbeta^2), \forall \xi\in \R^3$. For $ \xi\in \R^3,$ with $ |\xi|\ge 1$, this implies $\xi  \cdot\epsilon \xi  \ge  (c -  O(\cbeta^2)) |\xi|^2$, which is equivalent to \eqref{eps-ineq}.    \end{proof}

\section{Nonlinear estimates} \label{sec:nonlin-est-rough}

Let $N_\delta$ be given implicitly by \eqref{eqn:macro-eqn-varphi} and recall the definition of the $B_{s,\delta}$ norm from \eqref{s-norm-def}. Let $\dot{H}^{0}\equiv L^2$.  
In this section we prove estimates on $N_\delta$.


\begin{proposition} \label{prop:nonlinear-del-rough}
 Let Assumption \ref{A:diel} hold. 
 If $\|\varphi_1\|_{B_{s,\delta}}, \|\varphi_2\|_{B_{s,\delta}} = o(\delta^{-1/2})$, then we have the estimate
\begin{align}
 \label{Ndel-est-rough}	\|N_\delta & (\varphi_1) - N_\delta(\varphi_2)\|_{L^2} \notag \\
		& \ls  m^{ -1/3} \delta^{-1/2}(\|\varphi_1\|_{B_{s,\delta}} + \|\varphi_2\|_{B_{s, \delta}})\|\varphi_1 - \varphi_2\|_{\dot{H}^1} .
\end{align}
\end{proposition}
\DETAILS{\begin{proposition} \label{prop:nonlinear-rough}
 Let Assumption \ref{A:diel} hold. 
Then  we have the estimates  
\begin{align}
 \label{nonlinear-estm-rough}	 
  &\|N_\delta(\varphi_1) - N_\delta(\varphi_2)\|_{L^2}  \notag \\
& \ls  m^{-\frac13} \delta^{-1/2} (\|\varphi_1\|_{{ \delta}} + \|\varphi_2\|_{{ \delta}}) 
	\|\varphi_1 - \varphi_2\|_{{ \delta}}, 
\end{align} 
where, recall,  $m$ is defined in 
 \eqref{m} and, to simplify the notation, we let $m\ls 1$. 
\end{proposition}}
In Appendix \ref{sec:nonlin-est-improv}, we prove a more refined estimate. 
 We derive Proposition \ref{prop:nonlinear-del-rough}  from  its version  with $\delta = 1$ by rescaling. 
For  $\delta = 1$, we have the following result.

  \begin{proposition} \label{prop:nonlinear-rough}
 Let Assumption \ref{A:diel} hold 
 and either  $\| \psi\|_{L^2} = o(1)$ or  $\|\nabla \psi\|_{L^2} = o(1)$. 
 Then  $N:=N_{\delta=1}$ satisfies the estimate 
\begin{align}
	 \|N(\psi_1) - N(\psi_2)&\|_{L^2}\ls  \sum_{j=1}^2\bigg[ (\|\psi_j\|_{\dot{H}^1})\|\psi_1 - \psi_2\|_{\dot{H}^1} \notag \\&+  \big(\|\psi_j\|_{\dot{H}^1}^{1/3} \|\psi_j\|_{L^2}^{2/3} \|\psi_1 - \psi_2\|_{\dot{H}^{1}} \notag \\&  + \|\psi_j\|_{\dot{H}^{1}} \|\psi_1 - \psi_2\|_{\dot{H}^{1}}^{1/3}\|\psi_1 - \psi_2\|_{L^2}^{2/3}\big)\bigg]. \label{N-est-rough}\end{align}
\end{proposition}
We first derive Proposition \ref{prop:nonlinear-del-rough}  from  Proposition \ref{prop:nonlinear-rough} and then prove the latter statement. 
\begin{proof}[Proof of Proposition \ref{prop:nonlinear-del-rough}]
 By \eqref{eqn:F-delta-rescaling-rel}, $N_\delta$ and the unscaled nonlinearity $N = N_{\delta = 1}$ are related via
\begin{align}
	N_\delta(\varphi) = \delta^{-3/2} \Udel N( \psi),\ \qquad \psi = \delta^{-1/2} \Udel^* \varphi, \label{nonlinear-explicit-rescale}
\end{align}
where $\Udel$ is given in \eqref{rescU}. 
  Eqs \eqref{N-est-rough} and \eqref{nonlinear-explicit-rescale}   the relation $\|\Udel^* \varphi\|_{L^2}= \|\varphi\|_{L^2}$ and the notation $\psi_j = \delta^{-1/2} \Udel^* \varphi_j$ imply 
\begin{align}
	 \|N_\delta(\vphi_1) - N_\delta(\vphi_2)\|_{L^2}&\ls  	\delta^{-3/2}  \sum_{j=1}^2\bigg[\|\psi_j\|_{\dot{H}^1}\|\psi_1 - \psi_2\|_{\dot{H}^1} \notag \\ &+ \|\psi_j\|_{\dot{H}^1}^{1/3} \|\psi_j\|_{L^2}^{2/3} \|\psi_1 - \psi_2\|_{\dot{H}^{1}}  \notag \\& + \|\psi_j\|_{\dot{H}^{1}} \|\psi_1 - \psi_2\|_{\dot{H}^{1}}^{1/3}\|\psi_1 - \psi_2\|_{L^2}^{2/3} \bigg].\end{align}
	  Furthermore, using the relation $\|\psi_j\|_{\dot{H}^{k}}= \delta^{-1/2}\|\Udel^* \varphi_j\|_{\dot{H}^{k}}= \delta^{k-1/2}\| \varphi\|_{\dot{H}^{k}}$, we find
	  \begin{align}
	 \|N_\delta(\vphi_1) - N_\delta(\vphi_2)\|_{L^2}&\ls    \delta^{- 3/2}	  \sum_{j=1}^2\bigg[\delta \|\vphi_j\|_{\dot{H}^1} \|\vphi_1 - \vphi_2\|_{\dot{H}^1} \notag \\ &+ 
 	 \delta^{\frac13} \big(\|\vphi_j\|_{\dot{H}^1}^{1/3} \|\vphi_j\|_{L^2}^{2/3} \|\vphi_1 - \vphi_2\|_{\dot{H}^{1}} \notag \\&  + \|\vphi_j\|_{\dot{H}^{1}} \|\vphi_1 - \vphi_2\|_{\dot{H}^{1}}^{1/3}\|\vphi_1 - \vphi_2\|_{L^2}^{2/3}\big)\bigg]	. \label{Ndel-est'-rough} 
	\end{align}
To estimate the terms on the r.h.s. of \eqref{Ndel-est'-rough} we use the inequality $a^{1/3} b^{2/3}\le \frac23 (a+b)$, with $a:=\|\vphi\|_{\dot{H}^1}$ and $b:= m^{ 1/2}\del^{-1}\|\psi\|_{L^2}$, to obtain  
\[\|\vphi\|_{\dot{H}^1}^{1/3} \|\psi\|_{L^2}^{2/3}\le \frac23 (m^{ -1/2}\del)^{2/3}(\|\vphi\|_{\dot{H}^1}+m^{ 1/2}\del^{-1}\|\psi\|_{L^2}).\]
With  the definition of the norm $\|\cdot\|_{\del}$ in \eqref{s-norm-def}, this yields $\delta^{\frac13} \|\vphi\|_{\dot{H}^1}^{1/3} 
\|\vphi\|_{L^2}^{2/3} \|\chi\|_{\dot{H}^{1}}$\\ $\le\frac23 m^{ -1/3}\del \|\vphi\|_{\del} \|\chi\|_{\dot{H}^{1}}.$ Since $ \|\chi\|_{\dot{H}^{1}}\le \|\chi\|_{\del}$, this in turn implies \[\delta^{\frac13} \|\vphi\|_{\dot{H}^1}^{1/3} \|\vphi\|_{L^2}^{2/3} \|\chi\|_{\dot{H}^{1}}\le\frac23 m^{ -1/3}\del \|\vphi\|_{\del}\|\chi\|_{\del}.\]
 Applying this inequality to \eqref{Ndel-est'-rough}, 
  we arrive at \eqref{Ndel-est-rough}.  
 \end{proof}

\begin{proof}[Proof of Proposition \ref{prop:nonlinear-rough}]
Let $h_{\rm per}$ and $r_{\rm per}(z)$ be given in \eqref{rper-def}. First we observe that Eqs. \eqref{phi-del-eq}-\eqref{KdeltaMdelta-def}, with $\del=1$, read
\begin{align}
	\label{N} &N (\psi) =  F(\phi)- F(\phi_{\rm per})-d_\varphi F (\phi_{\rm per})\psi,\\ 
	\label{F-phi}&F (\phi) =   \den[ \ft(h^\phi - \mu)],	 
\end{align}
where $\psi := \phi - \phi_{\rm per}$ and, recall, $h^\phi:=- \Delta - \phi=h_{\rm per}- \psi$. 
Next, using Eqs \eqref{fFD-cauchy-formula} and \eqref{eqn:oint-adhoc} and expanding $(z-\hphi{\phi})^{-1}=(z-h_{\rm per}+ \psi)^{-1}$ to the second order, we find
\begin{align}\label{N-def}	N(\psi) :=& \den[\tilde{N}_2(\psi)],   
\end{align}
where
\begin{align}\label{tildeN-def} 
& \tilde{N}_k(\psi) := \oint (z-h_{\rm per}+ \psi)^{-1} [(-\psi)  r_{\rm per}(z) ]^k , 
\end{align}
with $\oint$ given by $\oint := \frac{1}{2\pi i} \int_\Gamma dz f_{T}(z-\mu)$, where $\Gamma$ is the contour given in Figure \ref{fig:cauchy-int-contour} (see \eqref{eqn:oint-adhoc}), equipped with the positive orientation.

We deform the contour $\Gamma$  given in Figure \ref{fig:cauchy-int-contour}  into the contour 
  indicated in Fig. \ref{fig:cauchy-int-contour'} by the blue dashed line and consisting of two separate contours traversed counter-clockwise.

By the formal resolvent expansion (without justifying the convergence)
\begin{align}
	(z-h_{\rm per}+\psi)^{-1} =&\sum_{k=2}^{\infty} r_{\rm per}(z) [(-\psi)  r_{\rm per}(z) ]^k, 
\end{align} 
we see that $N(\psi)$ can be written as the formal series
\begin{align} \label{N-series}	N(\psi) =& 
	\sum_{k=2}^{\infty} \den[N_k(\psi)],
\end{align}
 where 
\begin{align} \label{Nk-def}	& N_k(\psi) := \oint  r_{\rm per}(z)  [(-\psi)  r_{\rm per}(z) ]^k .
\end{align}

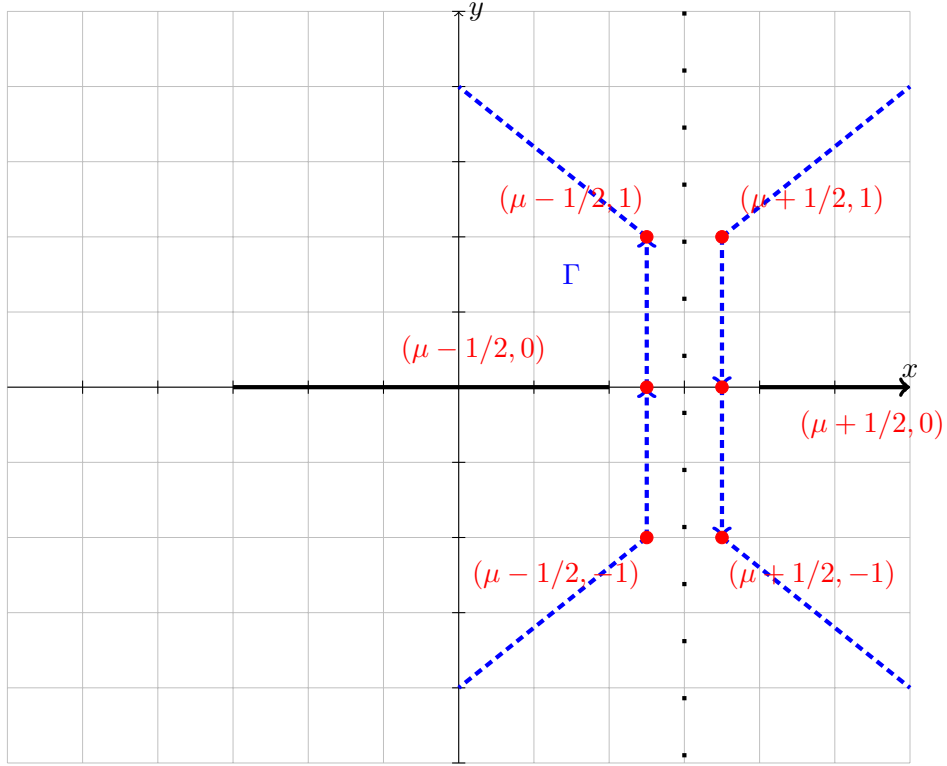
\begin{figure}[ht]  
\begin{center}
\begin{tikzpicture}
    \path [draw, help lines, opacity=.5]  (-6,-5) grid (6,5);
    \foreach \i in {1,...,5} \draw (\i,2.5pt) -- +(0,-5pt) node [anchor=north, font=\small] {} (-\i,2.5pt) -- +(0,-5pt) node [anchor=north, font=\small] {} (2.5pt,\i) -- +(-5pt,0) node [anchor=east, font=\small] {} (2.5pt,-\i) -- +(-5pt,0) node [anchor=east, font=\small] {};
    \draw [->] (-6,0) -- (6,0) node [anchor=south] {$x$};
    \draw [->] (0,-5) -- (0,5) node [anchor=west] {$y$};
		\draw [->, ultra thick, color=blue, densely dashed] (3.5,2) node {} -- (3.5,0) node {};
		\draw [->, ultra thick, color=blue, densely dashed] (2.5,0) node {} -- (2.5,2) node {};
		\draw [->, ultra thick, color=blue, densely dashed] (2.5, -2) node {} -- (2.5, 0) node {};
		\draw [->, ultra thick, color=blue, densely dashed] (3.5, 0) node {} -- (3.5, -2) node {};
		\draw [-, ultra thick, color=blue, densely dashed] (2.5, 2) node {} -- (0, 4) node {};
		\draw [-, ultra thick, color=blue, densely dashed] (2.5, -2) node {} -- (0, -4) node {};
		\draw [-, ultra thick, color=blue, densely dashed] (3.5, 2) node {} -- (6, 4) node {};
		\draw [-, ultra thick, color=blue, densely dashed] (3.5, -2) node {} -- (6, -4) node {};
		\draw [-, ultra thick, color=black] (-3, 0) node {} -- (2, 0) node {};
		\draw [->, ultra thick, color=black] (4, 0) node {} -- (6, 0) node {};
		\tikzstyle{loosely dotted}=[dash pattern=on \pgflinewidth off 20pt]
		\draw [-, ultra thick, color=black, loosely dotted] (3, 5) node {} -- (3, -5) node {};
		\draw[red, fill=red] (2.5,0) circle (.5ex);
		\node[color=red] at (0.2,0.5) {$(\mu - 1/2, 0)$};
		\draw[red, fill=red] (3.5,0) circle (.5ex);
		\node[color=red] at (5.5,-0.5) {$(\mu + 1/2, 0)$};
		\node[color=red] at (1.5,2.5) {$(\mu - 1/2, 1)$};
		\draw[red, fill=red] (2.5,2) circle (.5ex);
		\node[color=red] at (4.7,2.5) {$(\mu + 1/2, 1)$};
		\draw[red, fill=red] (3.5,2) circle (.5ex);
		\node[color=red] at (1.3,-2.5) {$(\mu - 1/2, -1)$};
		\draw[red, fill=red] (2.5,-2) circle (.5ex);
		\node[color=red] at (4.7,-2.5) {$(\mu + 1/2, -1)$};
		\draw[red, fill=red] (3.5,-2) circle (.5ex);
		\node[color=blue] at (1.5,1.5) {$\G$};
  \end{tikzpicture}
\end{center}
\caption{The deformation of the contour $\G$ is 
 indicated by the blue dashed line.   The spectrum of $h_{\rm per}$ is denoted by solid black line. The poles of $f_{T}(z-\mu)$ are denoted by the black dots. The number $c$ denotes the lower bound $h_{\rm per} > -c + 1$.} 
 \label{fig:cauchy-int-contour'}
\end{figure}

\begin{proposition} \label{prop:Nk-est}
 Let Assumption \ref{A:diel} hold 
  and let $N_2$ be given by \eqref{Nk-def}. Assume that $\|\nabla \psi\|_{L^2} = o(1)$, then we have the estimate
\begin{align}
	& \|\den[N_k(\psi)]\|_{L^2} 
	 \ls   \|\nabla \psi\|_{L^2}^{4/3}  
	  \|\psi\|_{L^2}^{2/3} \|\psi\|_{H^j}^{k-2},\ j=0, 1,\label{Nk-est-rough}
\end{align}
where the constants associated with $\ls$ are independent of $\beta$. 
\end{proposition}
\begin{proof}
Below, we use the notation $r =  r_{\rm per}(z)$, where $r_{\rm per}(z)$ is given in \eqref{rper-def}, and the estimate (see \eqref{Del-r-est'})
\begin{align}\label{Del-r-est''}\|(1-\Delta)^{\al} r\|\ls\ & d^{\al-1}\ls 1,\end{align}
for $\al\in [0, 1]$ and $z\in \Gam$, where \begin{align}\label{d-est-eta'}d \equiv d(z):=\dist (z, \s(h_{\rm per}))\ge \frac14 .\end{align}

We use the $L^2$-$L^2$ duality to estimate the $L^2$ norm of $\den[N_k(\psi)]$. We have, by  \eqref{den-def1} and definition \eqref{Nk-def},
 \begin{align}\notag	\|\den[N_k(\psi)]\|_{L^2} &=\sup_{ \|f\|_{L^2}=1}|\int f \den[N_k(\psi)]|\\ \notag&=\sup_{ \|f\|_{L^2}=1}|\Tr [f N_k(\psi)]|\\
 \label{NL2-dual'} &=\sup_{ \|f\|_{L^2}=1}|\oint \Tr (f  r (\psi r)^k)|.\end{align} 
(In the last two lines, $f$ is considered as a multiplication operator.)
 
 Let $f \in L^2$ and recall the Schatten norm $\|\cdot\|_{S^p}$ 
  defined in \eqref{Ip-norm'}. Using the non-abelian H\"{o}lder's inequality $1 = \frac{1}{2} + \frac{1}{6} + \frac{1}{3}+ \frac{1}{\infty}$, 
 we see that, for $k\ge 2$,
\begin{align}\label{fRphiRphiR}	|\Tr( f r (\psi r)^k)| \ls & \|f r\|_{S^2} \|\psi r\|_{S^6} \|\psi r\|_{S^3}\|\psi r\|^{k-2} \, . 
\end{align}
Next, we use 
 the operator trace-class estimate $\|A\|_{S^3}^3 = \Tr(|A|^3) \leq \|A\| \Tr(|A|^2) = \|A\| \|A\|_{S^2}^2\le \|A\|_{S^6} \|A\|_{S^2}^2$ to obtain 
\begin{align}\label{I3-I2-bnd'} 
 \|A\|_{S^3}\le \|A\|_{S^6}^{1/3} \|A\|_{S^2}^{2/3}.	 
\end{align}


Using this equality to estimate the third factor in \eqref{fRphiRphiR} and the standard relative bounds $\|\psi r\|\ls \|\psi \|_{L^2}$ and  $\|\psi r\|\ls \|\psi \|_{L^6}\ls \|\psi \|_{\dot H^1}$, we bound the r.h.s. of \eqref{fRphiRphiR} as 
\begin{align} \label{L2-est1'}	|\Tr( f r (\psi r)^k)| \ls &  
\|f r\|_{S^2} \| \psi r\|_{S^6}^{4/3} \|\psi r\|_{S^2}^{2/3}\|\psi \|_{\dot H^i}^{k-2},\ i=0, 1.
\end{align} 
For a typical term on the r.h.s., we have 
 $\|g r\|_{S^p}\le \|g (1-\Delta)^{-\al_p} \|_{S^p} \| (1-\Delta)^{\al_p} r\|$, with $3/(2p) <\al_p<1, p>3/2$, which, together with Kato-Seiler-Simon's inequality \eqref{eqn:KSS} and inequality \eqref{Del-r-est''}, gives \[\|g r\|_{S^p}\ls\|g\|_{L^p}d^{\al_p-1},\ 3/(2p) <\al_p<1, p>3/2.\] 
Applying this estimate to each of the first three factors on the r.h.s. of \eqref{L2-est1'} and using 
the Gagliardo-Nirenberg-Sobolev inequality \eqref{eqn:hardy-ineq}, we find  
\begin{align} \label{L2-est2'}	\big|\Tr( f r (\psi r)^k)\big|\ls & d^{-4/3}\|f\|_{L^2} \|\n \psi\|_{L^2}^{4/3} \|\psi \|_{L^2}^{2/3}\|\psi \|_{\dot H^j}^{k-2}\, , 
	\end{align}
for $j=0, 1$.  Recalling definition \eqref{d-est-eta'} of $d \equiv d(z)$, we see that  the integral on the r.h.s. of \eqref{NL2-dual'} converges absolutely. 
Eqs \eqref{L2-est2'}, \eqref{d-est-eta'}, \eqref{eqn:oint-adhoc} and \eqref{oint-decay} 
give
 \begin{align} \big| \oint \Tr( f r (\psi r)^k) \big| \ls & 
  \|f\|_{L^2} \|\n \psi\|_{L^2}^{4/3} \|\psi \|_{L^2}^{2/3}\|\psi \|_{\dot H^j}^{k-2}, \label{L2-est'} \end{align} 
for $j=0, 1$.   Eqs. \eqref{NL2-dual'} and \eqref{L2-est'} imply 
\eqref{Nk-est-rough}. 
\end{proof}


Now, we complete the proof of Proposition \ref{prop:nonlinear-rough}.  Proposition \ref{prop:Nk-est}  shows that 
  if 
  either $\|\psi\|_{L^2} = o(1)$ or $\|\psi\|_{L^2} < \infty$ and $\|\psi\|_{\dot{H}^{1}} = o(1)$, then series \eqref{N-series} converges absolutely in $L^2$. 

Now, using series \eqref{N-series}, we write  
\begin{align}
	N(\psi_1) - N(\psi_2)= \sum_{k \geq 2} \den[N_k(\psi_1) - N_k(\psi_2)] \, .
\end{align}
By definition \eqref{Nk-def}, $N_k(\psi)$ is an $k$-th degree monomial in $\phi$.
Hence, we can expand $N_k(\psi_1) - N_k(\psi_2)$ in the following telescoping form
\begin{align}
	x^k& - y^k = x^{k-1}(x-y) + x^{k-2}(x-y)y + \cdots + (x-y)y^k\, . \label{N-telescoping}
\end{align}
The proof of Proposition \ref{prop:nonlinear-rough} follows by applying appropriate and straightforward extension of  Proposition \ref{prop:Nk-est}  
to each term in the expansion of  $N_k(\psi_1) - N_k(\psi_2)$ given in \eqref{N-telescoping}.
 \end{proof}

\appendix
\section{$\epsilon(T) \rightarrow \epsilon(0)$ as $T \rightarrow 0$} \label{sec:epsilon-T-to-zero}
\begin{lemma} \label{lem:epsilon-T-to-zero}
Let $\ex = 0$. Then $\epsilon\equiv \epsilon(T) \rightarrow \epsilon(0)$ as $T \rightarrow 0$, where $\epsilon(0)$ is the dielectric constant for $T=0$ obtained in \cite{CL}.
\end{lemma}
\begin{proof}
 We see from \eqref{eps} 
  that $\epsilon(T), T=1/\beta,$ is of the form
\begin{align}
	\epsilon(T)  = \frac{1}{2\pi i} \int_\G \ft(z-\mu) X(z)
\end{align}
where $X(z)$ is some  holomorphic function on $\C \backslash \R$, independent of $\beta$, and remains holomorphic on the real axis where the gap of $h_{\rm per}$ occurs. On $\R$, we note that $f_{\rm FD}(\beta x)$ converges to the indicator function $\chi_{(-\infty, 0)}$ as $\beta \rightarrow \infty$. If we take $\beta \rightarrow \infty$, the integral
\begin{align}
\frac{1}{2\pi i} \int_\G \ft(z-\mu) X(z)
\end{align}
converges to $\frac{1}{2\pi i} \int_{G_1} X(z)$ where $G_1$ is any contour around the part of the spectrum of $h_{\rm per}$ that is less than $\mu_{\rm per}$. This is the same expression as in \cite{CL} after inserting $1 = \sum_i |\varphi_i \ran \lan \varphi_i|$ for each resolvent of $h_{\rm per}$ in $X(z)$ where the $\varphi_i$'s are eigenvectors of $h_{\rm per}$.
\end{proof}

\section{Bounds on  
$m$ and $V$} \label{app:bounds-on-V}
In this section, we prove bounds on 
$m$ and $V$  given   \eqref{m} and \eqref{V}. Note that $m=\|V\|_{L^1_{\rm per}}$. 
Since $f'_{T} < 0$ ($T= 1/\beta$),  \eqref{V}  implies that 
$V > 0$ 
and therefore, by \eqref{V},  $\|V\|_{L^1_{\rm per}} = \int_\Om V$, where $\Om$ is a fundamental domain of $\lat$ (see Section \ref{sec:results}), which yields
\begin{align} \label{m-intV-expr}m=\int_\Om V = - 
  \Tr_{L^2_{\rm per}}  f'_{T}(h_{{\rm per}, 0}-\mu),\ \qquad \mu=\mu_{\rm per}.\end{align}

\begin{lemma} \label{lem:V1-lower-bound}
Let Assumption \ref{A:diel} 
 hold and $\etal$ be given in \eqref{eta-Om}. Then
\begin{align} \label{V1-lower-bound}	m=\|V\|_{L^1_{\rm per}} \ge \frac14  \beta e^{-\beta \etal},
\end{align}
where 
$\etal$ is given in \eqref{eta-Om}. 
\end{lemma}
\begin{proof}
 Using that $\etal$ is the smallest distance between $\mu = \mu_{\rm per}$ and the spectrum of $h_{{\rm per}, 0}$ (see \eqref{eta-Om}) and Eq. \eqref{m-intV-expr} and replacing $\Tr_{L^2_{\rm per}}  f'_{T}(h_{{\rm per}, 0}-\mu)$ by the contribution of the eigenvalue of $h_{{\rm per}, 0}$ closest to $\mu$, we find 
\begin{align}
		m \geq & - 
		 f'_{T}(\etal) =  \beta \frac{e^{\beta \etal}}{(1+e^{\beta \etal})^2}  \ge \frac14  \beta e^{-\beta \etal}. \label{eqn:beta-f-prime-mu}
\end{align}
This gives  \eqref{V1-lower-bound}. \end{proof}

\begin{lemma} \label{lem:Vp-uppper-bound}
Let Assumption \ref{A:diel} 
  hold. Then, for $1 \leq p \leq \infty$, 
   \begin{align}
	\|V\|_{L^p_{\rm per}} \ls 
	\beta e^{-\etal\beta}. \end{align}
\end{lemma}
\begin{proof}
We do the case for $p=1$ and $p=\infty$, and conclude the lemma by interpolation. 
 By Assumption \ref{A:diel}, 
 the potential $\phi_{\rm per}$ is bounded. 
 Thus, $h_{{\rm per}, 0}$ has only discrete spectrum on $L^2_{\rm per}(\R^3)$ and
\begin{align}
	\frac{1}{\beta}\|V\|_{L^1_{\rm per}} = & \sum_{\lambda \in \sigma(h_{{\rm per}, 0})} \frac{ e^{\beta(\lambda-\mu)}}{(1+e^{\beta(\lambda-\mu)})^2}\\ 
 \label{V-L1-bnd}		\le & \sum_{\mu > \lambda \in \sigma(h_{{\rm per}, 0})}e^{\beta(\lambda-\mu)}+ \sum_{\mu < \lambda \in \sigma(h_{{\rm per}, 0})}e^{-\beta(\lambda-\mu)}\\ \label{V-L1-bnd'} = & \sum_{\lambda \in \sigma(h_{{\rm per}, 0})}e^{\beta |\lambda-\mu|} .\end{align}
  Again, we use that $\etal$ is the smallest distance between $\mu = \mu_{\rm per}$ and the spectrum of $h_{{\rm per}, 0}$ (see \eqref{eta-Om}).  Peeling the eigenvalue(s) closest to $\mu$ and letting $\etal +\xi$ stand for the distance between $\mu$ and the rest of the spectrum $\sigma(h_{{\rm per}, 0})$, 
  we find, for some constant $c$,
\begin{align} \label{V-L1-bnd''}
   \sum_{\lambda \in \sigma(h_{{\rm per}, 0})}& e^{\beta |\lambda-\mu|}  =ce^{\beta \etal}+  \sum_{\lambda \in \sigma(h_{{\rm per}, 0}), |\lambda-\mu|\ge\etal +\xi }e^{\beta |\lambda-\mu|}.\end{align}
 We estimate the sum on the r.h.s. by an integral as follows. 
Since the potential $\phi_{\rm per}$ is infinitesimally bounded with respect to $-\Delta$, the eigenvalues of $h_{{\rm per}, 0}$ go to infinity at a similar rate as those of $-\Delta$ (on $L^2_{\rm per}(\R^3)$), i.e. as $n^2$. Thus, assuming that for $\lam$ sufficiently large, the $n$th eigenvalue $\lam_n\approx n^2$ has the degeneracy of the order $O(n^k), k\ge 0$, we conclude that 
\begin{align} \label{2nd-sum-bnd} 
		& \sum_{\mu < \lambda \in \sigma(h_{{\rm per}, 0})}e^{-\beta(\lambda-\mu)}\ls \int_{x^2 \ge \mu +\etal+\xi} x^k e^{-\beta (x^2 -\mu)} d x \notag\\
		&\qquad =\frac12 \int_{y\ge\etal+\xi} (y+\mu)^{\frac{k-1}{2}} e^{-\beta y} d y  \ls\frac1\beta \mu^{\frac{k-1}{2}} e^{-\beta (\etal+\xi)}.\end{align} 
For the first sum in \eqref{V-L1-bnd}, we consider separately the cases $\mu\ls 1$ and $\mu\gg 1$ and, in the 2nd case, break the sum into the sums over $\lam \ls 1$ and  $\lam \gg 1$.  In the first three situations, the estimate is straightforward and in the last one, we proceed as in \eqref{2nd-sum-bnd} to obtain  
\[\sum_{\mu -\etal-\xi \ge \lambda \in \sigma(h_{{\rm per}, 0})}e^{\beta(\lambda-\mu)}\ls \frac1\beta \mu^{\frac{k-1}{2}} e^{-\beta \etal}.\] This proves the lemma for $p=1$.

Let $W^{4,1}_{\rm per}$ be the usual Sobolev space associated to $L^1_{\rm per}$ involving up to 4 derivatives. For the case $p=\infty$, we use the Sobolev inequality 
\begin{align}
	\|f\|_\infty \ls \|f\|_{W^{4,1}_{\rm per}}
\end{align}
for $f \in W^{4,1}_{\rm per}$.  Thus, it suffices for us to estimate $\|\nabla^j V\|_{L^1_{\rm per}}, j=0,...,4$. To this end, we note that 
\begin{align}
	\nabla \den(A) = \den( [\nabla, A] )
\end{align}
for an operator $A$ on $L^2_{\rm per}(\R^3)$. Thus, it suffices that we estimate the trace 1-norm of $\nabla^s f_{T}'(h_{{\rm per}, 0}-\mu) \nabla^{4-s}$ on $L^2_{\rm per}(\R^3)$ for $s=0,...,4$. 
 {Since the potential $\phi_{\rm per}$ is bounded together with all its derivatives, 
we have, for $s=0, \dots, 4$,  
   \begin{align} \label{nab-h-bnd}	\|\nabla^s h^{-s/2}\|\ls 1, \end{align}
where  $h:=h_{{\rm per}, 0}+c$, with $c>0$  s.t. $h_{{\rm per}, 0}+c>0$. Indeed, to fix ideas, consider one of the terms, say, $\|\nabla^3 h^{-3/2}\|$. We have  $\|\nabla^3f\|^2\le \|(-\Delta)^{3/2}f\|^2=\lan f, (h+\phi_{\rm per}-c)^3f\ran$. Taking $f = h^{-3/2}u$, expanding the binomial $(h+\phi_{\rm per}-c)^3$ and commuting the operator $h$ in  the resulting terms $h^2\phi_{\rm per}$ and $\phi_{\rm per} h^2$ to the right and left, respectively, and estimating the resulting commutators, $[h, \phi]$ and $[\phi_{\rm per}, h]=-[h, \phi_{\rm per}]$, we arrive at the estimate $\|\nabla^3 h^{-3/2}\|\ls 1$ as claimed.  \eqref{nab-h-bnd} implies also that $ \| h^{-2+s/2}\nabla^{4-s}\|\ls 1$, for $j=0, \dots, 4$. As the result, we have \[\|\nabla^s f_{T}'(h_{{\rm per}, 0}-\mu) \nabla^{4-s}\|_{S^1}\ls \|g(h_{{\rm per}, 0})\|_{S^1},\]  
where $g(x):=-(x+c)^{s/2} f_{T}'(x-\mu) (x+c)^{2-s/2}\ge 0$.} 
   Hence, it suffices  to estimate   $\|g(h_{{\rm per}, 0})\|_{S^1}=\tr [g(h_{{\rm per}, 0})]$. The latter can be done the same way as the case for $p=1$ by summing eigenvalues of $h_{{\rm per}, 0}$ and the lemma is proved.
\end{proof}

\section{Bound on $M_\del$} \label{sec:Mdel-bnd}
In an analogy to $\Lper\equiv L_{\rm per}^2(\R^3)$ given in \eqref{L2-per}, we let
\begin{align}\label{L2-per-delta}	
	L^2_{{\rm per}, \delta}\equiv  L^2_{{\rm per}, \delta}(\R^3)  := \{ f \in L^2_{\rm loc}(\R^3) : \text{$f$ is $\latd$-periodic } \}. 
\end{align} 
Moreover, we recall 
 $\n^{-1} := \n (-\Delta)^{-1}$ (see \eqref{nabla-inverse-def}). 
The main result of this appendix is the following  
\begin{proposition} \label{prop:Mdel-deco} Let  Assumption \ref{A:diel} 
  hold. Then the operator $M_\del$ can be decomposed as 
\begin{align} \label{Mdelta-deco}	M_\delta = M_\del' + M_\delta'' \, ,
\end{align}
with the operator $M_\del'$ and $M_\del''$ satisfying the estimates 
\begin{align} \label{M0-bnd}	&\| \bar \Proj  \nabla^{-1} M_\del' \Proj  \varphi \|_{L^2} \ls 
 \delta^{-1} \|V\|_{L^2_{\rm per}} \|\Proj \varphi\|_{L^2},\\
\label{Mdel''-bnd}	&\| \bar \Proj  \nabla^{-1} M_\delta'' \Proj  \nabla^{-1} \varphi \|_{L^2} \ls 
  \|\Proj  \varphi\|_{L^2}.
\end{align}
\end{proposition}
\begin{proof}[Proof of Proposition \ref{prop:Mdel-deco}]

Proposition \ref{prop:M-BF-deco} and the rescaling relation \eqref{Mdelta-M-relation'} imply the explicit form for the $k$-fibers of $M_\delta$:

\begin{lemma} \label{lem:for:pro:M-low-E}
 Then $M_\delta$ has a Bloch-Floquet decomposition \eqref{BF-deco-opr} with $\lat = \latde$, whose $k-$fiber $M_{\delta, k}$ acting on $\LperB$ is given by 
\begin{align} \label{exact-form-Mdeltak}
	M_{\delta,k} f =&  -  \delta \den\left[ \oint r^\del_{\rm per, 0}(z)f \rperkDelta(z)\right]
\end{align}
where 
 $f \in \LperB$ and, on $\LperB$,
\begin{align}
	&\rperkDelta = (z- \hperkDelta)^{-1}, \label{eqn:rperkDelta-adhoc} 
\quad	\hperkDelta = \delta^2 (-i\nabla-k)^2 + \delta \phi_{\rm per}^\delta. 
\end{align}
\end{lemma}

  We decompose  the operator $M_{\delta, k}$ acting on $\LperB$ as 
  \begin{align}
	M_{\delta, k} &=: M_{\delta, 0} + M_{\delta, k}'' \, , \label{eqn:Mprime-delta-k}
\end{align}
where $M_{\delta, 0} = M_{\delta, k=0}$ and $M_{\delta, k}'$ is defined by the expression \eqref{eqn:Mprime-delta-k}. We define operators $M_\del'$ and $M_\delta''$ on $L^2(\R^2)$ via
\begin{align}
	M_\del' &:=   \int_{\Omde^*}^\oplus d\hat k \, M_{\delta, 0} \varphi , \label{eqn:M0-def}\\
	M_\delta'' &:=  \int_{\Omde^*}^\oplus d\hat k \,  M''_{\delta, k} \varphi  \label{M''-def}\, ,
\end{align}
where $\Omde^*$ is a fundamental cell of the reciprocal lattice to  $\latd$ and  $d\hat k = |\Omde^*|^{-1} dk$. By Lemma \ref{lem:for:pro:M-low-E} and definition \eqref{eqn:Mprime-delta-k}, the latter operators satisfy \eqref{Mdelta-deco}.

\begin{lemma} \label{lem:M0-on-P} 
$M_\del'$ (see \eqref{eqn:M0-def}) restricted to the range of $\Proj $ is a multiplication operator given by 
\begin{align}
\label{M0-Pr}	(M_\del' \Proj  \varphi)(x)  = V_\delta(x)  (\Proj \varphi)(x),
\end{align}
where
\begin{align}
	V_\delta(x) = -\delta^{-2}  \den\left[f_{T}'(h_{{\rm per}, 0}-\mu) \right](\delta^{-1}x) \label{Vdelta-def}\, ,
\end{align} 
with $h_{{\rm per}, 0}$ given in \eqref{eqn:hperk-adhoc} (with $k = 0$). 
\end{lemma}
\begin{proof}
By \eqref{exact-form-Mdeltak} and definition of $M_\del'$ in \eqref{eqn:M0-def}, we see that
\begin{align}
	M_\del' & P_r \varphi 
	= - \int^\oplus_{\Omde^*} d\hat k \, \delta \den\left[ \oint r^\del_{\rm per, 0}(z)(P_r \varphi)_k r^\del_{\rm per, 0}(z)\right].
\end{align}
By Corollary \ref{cor:Pf-computation} and the Cauchy integral formula, 
\begin{align}
	M_\del' \Proj  \varphi =& - \int^\oplus_{\Omde^*} d\hat k \, \delta \den\left[ \oint (r^\del_{\rm per, 0}(z))^2 \right] |\Omde|^{-1} \hat \varphi(k) \\
		=&  - \int^\oplus_{\Omde^*} d\hat k \, \den[f_{T}'(h^\delta_{{\rm per}, 0} - \mu)] |\Omde|^{-1} \hat \varphi(k) . \label{eqn:after-app-of-cauchy-to-M-0}
\end{align}
where $r^\del_{\rm per, 0}(z)$ and $h^\delta_{{\rm per}, 0}$ are given in \eqref{eqn:rperkDelta-adhoc}. 
Applying the inverse Bloch-Floquet transform \eqref{eqn:inverse-bloch}, \eqref{eqn:after-app-of-cauchy-to-M-0} implies
\begin{align}
	M_\del' &\Proj  \varphi 
		= -\delta   \den[f_{T}'(h^\delta_{{\rm per}, 0}-\mu) ]  \int_{\Omde^*} d\hat k \, e^{-ikx} |\Omde|^{-1}\hat \varphi(k). \label{eqn:eqn:after-app-of-cauchy-to-M-1}
\end{align}
Since $d\hat k$ is normalized by the volume $|\Omde^*|$ 
 (which is independent of the choice of the cell), \eqref{eqn:eqn:after-app-of-cauchy-to-M-1} shows
\begin{align}
	 \label{eqn:M0Pvarphi-bloch-form}	M_\del' \Proj  \varphi	=& - \delta \den\left[f_{T}'(h^\delta_{{\rm per}, 0}-\mu) \right]  \Proj \varphi .
\end{align}
By Lemma \ref{lem:U-rho-UAUstar} and recalling the definition of $\Udel$ from \eqref{rescU}, we see that
\begin{align}
	 \delta \den\left[ f_{T}'(h^\delta_{{\rm per}, 0}-\mu) \right] &= \delta \den\left[\Udel f_{T}'(h_{{\rm per}, 0}-\mu) \Udel^* \right] \\
	&= \delta^{-2} \den\left[f_{T}'(h_{{\rm per}, 0}-\mu) \right](\delta^{-1}x) \, ,
\end{align}
where $h_{{\rm per}, 0} = h_{{\rm per}, 0}^{\delta = 1}$, which together with \eqref{eqn:M0Pvarphi-bloch-form} gives \eqref{M0-Pr}-\eqref{Vdelta-def}.
\end{proof}
 
 \begin{proof}[Proof of \eqref{M0-bnd}]
Let $V_\delta$ be given in \eqref{Vdelta-def}.  Since the Bloch-Floquet decomposition is unitary,  we see, by Lemma \ref{lem:M0-on-P} and Corollary \ref{cor:Pf-computation}, that
\begin{align}
	\|M_\del' \Proj  \varphi\|_{L^2}^2 = & \|V_\delta \Proj  \varphi\|_{L^2}^2 =  \int_{\Om_\delta^*} d\hat k   \| V_\delta |\hat \varphi(k)| |\Om_\delta|^{-1}  \|_{\LperB}^2, \label{eqn:lem:M0-term-eqn}
\end{align}
where $\LperB$ is given in \eqref{L2-per-delta}.
Using the fact that $d\hat k = |\Omde^*|^{-1} dk$ and $|\Omde| = \delta^3 |\Om|$, \eqref{eqn:lem:M0-term-eqn} implies
\begin{align}
	\| M_\del' \Proj \|_{L^2}^2 = & \delta^{-3} |\Om|\| V_\delta \|_{\LperB}^2 \|\Proj \varphi\|_{L^2}^2.  \label{eqn:PM0Pvarphi-pre}
\end{align}
By a change of variable, we see that
	$\|V_\delta \|_{\LperB} = \delta^{-1/2}\|V \|_{L^2_{\rm per}},$ 
where $V$ is given by \eqref{V}. Combing with \eqref{eqn:PM0Pvarphi-pre}, the fact $\bar \Proj  (-i\nabla)^{-1}  \ls r^{-1}$ (where $\n^{-1}$ is given in \eqref{nabla-inverse-def}) and $r^{-1}= a^{-1}\delta \ls \delta$ (see \eqref{a-m}) yields Eq. \eqref{M0-bnd}. \end{proof}

\begin{proof}[Proof of \eqref{Mdel''-bnd}]
 Let $M_\delta''$ be given by \eqref{M''-def} 
  and 
 $k^{-1} := k/|k|^2$. Let $\varphi \in L^2(\R^3)$. 
 By Corollary \ref{cor:Pf-computation}, we have 
\begin{align}
(\Proj  \nabla^{-1} \varphi)_k=	k^{-1} \hat \varphi(k) |\Omde|^{-1} \chi_{\Br}(k). \label{eqn:p-inver-Proj-varphi}
\end{align}
This gives $M_\delta'' \Proj  \nabla^{-1} \varphi= |\Omde|^{-1}\int^\oplus_{\Br} d\hat k M''_{\delta, k} k^{-1} \hat \varphi(k)$. Since the Bloch-Floquet decomposition is unitary, 
 we see, using \eqref{eqn:p-inver-Proj-varphi}, that 
\begin{align}
		 \| M_\delta'' \Proj  \nabla^{-1} &\varphi \|_{L^2}^2 \notag\\
		&= |\Om_\delta|^{-2}\int_{B_r} d\hat k \, \|M''_{\delta, k} 1\|_{\LperB}^2 |k|^{-2} |\hat \varphi(k)|^2. \label{eqn:Mprime-two-scale-1}
\end{align}
Since $d \hat k = |\Omd^*|^{-1} dk = |\Omd| dk$ and $|\Om_\delta| = \delta^3|\Om|$,  \eqref{eqn:Mprime-two-scale-1} is bounded as
\begin{align}
	\| M_\delta'' \Proj  \nabla^{-1} &\varphi \|_{L^2}^2 \notag\\
		& \ls  \delta^{-3} \sup_{k \in B_r} \left( \|M''_{\delta, k} 1\|_{\LperB}^2 |k|^{-2} \right) \|\Proj  \varphi\|_{L^2}^2,  \label{eqn:Mprime-two-scale} 
\end{align}
where $1 \in \LperB$ is the constant function $1$ and $\LperB$ is given in \eqref{L2-per-delta}.

By \eqref{exact-form-Mdeltak} and \eqref{eqn:Mprime-delta-k}, we have, for $M''_{\delta, k}$  given in \eqref{eqn:Mprime-delta-k},  that 
\begin{align}
	M''_{\delta, k} \varphi =& -  \delta \den \left[ \oint r_{{\rm per}, 0}(z)  \varphi  (\rperkDelta(z)-r^\delta_{{\rm per}, 0}(z)) \right].
\end{align}
Since $\rperkDelta(z)-r_{{\rm per}, 0}(z)=r^\delta_{{\rm per}, 0}(z)A_k \rperkDelta(z)$, where 
$A_k:=-2 (-i\n) \delta k+\delta^2 |k|^2$, this gives 
\begin{align} 
	M''_{\delta, k} \varphi =& - \delta  \den \oint \left[\rperDelta(z)  \varphi \rperDelta(z) A_k \rperkDelta(z) \right]. \label{eqn:Mprimek-form}
\end{align}
By the rescaling relation \eqref{Mdelta-M-relation'} and \eqref{eqn:Mprimek-form}, we see that
\begin{align}
	 \|M''_{\delta, k}& 1\|_{\LperB} = \|U_\del^* M''_{\delta, k}U_\del \cdot U_\del^*  1\|_{L^2_{\rm per}} 
		=  \delta^{3/2-2}\|M''_{1, k} 1\|_{L^2_{\rm per}}\notag \\
		& =  \delta^{-1/2} \big\| \den \big[ \oint  r_{\rm per}^2(z) A_k\rperDeltak(z)  \big] \big\|_{L^2_{\rm per}}. \label{eqn:delta-3-2-M-prime-explicit}
\end{align}
By \eqref{eqn:delta-3-2-M-prime-explicit}, notation $A_k:=-2 (-i\n) \delta k+\delta^2 |k|^2$ and inequality \eqref{Del-r-est'}, with $\al=0,1/2$, we obtain, for $|k| \leq r$,
\begin{align}
	\|M''_{\delta, k}1\|_{L^2_{\rm per}} |k|^{-1} \ls & \delta^{-1/2+1}  + \delta^{-1/2+2} r  \notag\\
		=& \delta^{1/2} + \delta^{3/2} r \, . \label{eqn:d-12-d12r}
\end{align}
By \eqref{a-m} and \eqref{eqn:Mprime-two-scale}, equation \eqref{eqn:d-12-d12r} shows that
\begin{align}
	\| M_\delta'' \Proj  \nabla^{-1} \varphi \|_{L^2} \ls \delta^{-1} .
\end{align}
This bound, the observation that  $\|\bar \Proj \nabla^{-1}\|_\infty \ls r^{-1}$ (see \eqref{eqn:choice-of-P}) and the definition $r =a/\delta \gs 1/\del$ imply Eq. \eqref{Mdel''-bnd}. 
\end{proof}
This completes the proof of Proposition \ref{prop:Mdel-deco}.
\end{proof}
We use Proposition \ref{prop:Mdel-deco} to prove the following
 \begin{proposition} \label{prop:Mdel-bnd}  Let  Assumption \ref{A:diel} 
  hold and let $\beta e^{-\etal\beta}\ls 1$ (which is weaker than  Assumption \ref{A:scaling}).  Then the operator $M_\del$ is bounded as
\begin{align} \label{Mdel-bnd}	\|\n^{-1}\bar \Proj & M_\delta f\|_{\dot{H}^1} \ls  \| f\|_{\del}. 
\end{align}
\end{proposition}
\begin{proof}[Proof of Proposition \ref{prop:Mdel-bnd}]
 Decomposing $M_\delta$ according to   \eqref{Mdelta-deco} and using bounds  Eqs \eqref{M0-bnd} and \eqref{Mdel''-bnd} of Proposition \ref{prop:Mdel-deco}, we see that 
\begin{align} \notag 
	\|\n^{-1}\bar \Proj  M_\delta \Proj & \varphi\|_{L^2} \leq  \|\n^{-1}\bar \Proj  M_\del' \Proj  \varphi\|_{L^2} + \|\n^{-1}\bar \Proj  M_\delta'' \Proj  \varphi\|_{L^2}\\
	&\ls  \delta^{-1} \|V \|_{L^2_{\rm per}} \| f\|_{L^2} + \|\nabla f\|_{L^2}, \label{eqn:estimate-M-decompose-M-0-M-prime-lem-applied}
\end{align}
where $V$ is given in \eqref{V}. Using $\|V\|_{L^2_{\rm per}} \leq \|V\|_{L^\infty}^{1/2} \|V\|_{L_{\rm per}^1}^{1/2}$ and the definition $m:=\|V\|_{L^1_{\rm per}}$ in \eqref{eqn:estimate-M-decompose-M-0-M-prime-lem-applied} gives
\begin{align}
	\|\n^{-1}\bar \Proj & M_\delta f\|_{\dot{H}^1} 
	\ls  \|V \|_{L^\infty}^{1/2} (\delta^{-1} m^{1/2}\| f\|_{L^2}) +  \|\nabla f\|_{L^2} \, .
\end{align}
  Lemma \ref{lem:Vp-uppper-bound} and definition \eqref{s-norm-def} imply \eqref{Mdel-bnd}. 
\end{proof}

\section{Refined nonlinear estimates} \label{sec:nonlin-est-improv}

Let $N_\delta$ be given implicitly by \eqref{eqn:macro-eqn-varphi} and recall the definition of the $B_{s,\delta}$ norm from \eqref{s-norm-def}. Let $\dot{H}^{0}\equiv L^2$.  
In this section we prove estimates on $N_\delta$.


\begin{proposition} \label{prop:nonlinear}
 Let Assumption \ref{A:diel} hold. 
 If $\|\varphi_1\|_{B_{s,\delta}}, \|\varphi_2\|_{B_{s,\delta}} = o(\delta^{-1/2})$, then we have the estimate
\begin{align}
 \label{Ndel-est}	\|N_\delta & (\varphi_1) - N_\delta(\varphi_2)\|_{L^2} \notag \\
		& \ls  e^{-\beta}m^{ -1/3} \delta^{-1/2}(\|\varphi_1\|_{B_{s,\delta}} + \|\varphi_2\|_{B_{s, \delta}})\|\varphi_1 - \varphi_2\|_{\dot{H}^1} .
\end{align}
\end{proposition}
\DETAILS{\begin{proposition} \label{prop:nonlinear-rough}
 Let Assumption \ref{A:diel} hold. 
Then  we have the estimates  
\begin{align}
 \label{nonlinear-estm-rough}	 
  &\|N_\delta(\varphi_1) - N_\delta(\varphi_2)\|_{L^2}  \notag \\
& \ls  m^{-\frac13} \delta^{-1/2} (\|\varphi_1\|_{{ \delta}} + \|\varphi_2\|_{{ \delta}}) 
	\|\varphi_1 - \varphi_2\|_{{ \delta}}, 
\end{align} 
where, recall,  $m$ is defined in 
 \eqref{m} and, to simplify the notation, we let $m\ls 1$. 
\end{proposition}}

 We derive Proposition \ref{prop:nonlinear}  from  its version  with $\delta = 1$ by rescaling. 
For  $\delta = 1$, we have the following result.

  \begin{proposition} \label{prop:nonlinear-micro}
 Let Assumption \ref{A:diel} hold 
  $\| \psi\|_{L^2} = o(1)$. 
 Then  $N:=N_{\delta=1}$ satisfies the estimate 
\begin{align}
	 \|N(\psi_1) - N(\psi_2)&\|_{L^2}\ls  \sum_{j=1}^2\bigg[ (\|\psi_j\|_{\dot{H}^1})\|\psi_1 - \psi_2\|_{\dot{H}^1} \notag \\&+ 
	 e^{-\beta} \big(\|\psi_j\|_{\dot{H}^1}^{1/3} \|\psi_j\|_{L^2}^{2/3} \|\psi_1 - \psi_2\|_{\dot{H}^{1}} \notag \\&  + \|\psi_j\|_{\dot{H}^{1}} \|\psi_1 - \psi_2\|_{\dot{H}^{1}}^{1/3}\|\psi_1 - \psi_2\|_{L^2}^{2/3}\big)\bigg]. \label{N-est}\end{align}
\end{proposition}
The derivation of Proposition \ref{prop:nonlinear}  from  Proposition \ref{prop:nonlinear-micro} is same as that of  Proposition \ref{prop:nonlinear-del-rough}  from  Proposition \ref{prop:nonlinear-rough} and we omit it here. 
\begin{proof}[Proof of Proposition \ref{prop:nonlinear-micro}]
Let $h_{\rm per}$ and $r_{\rm per}(z)$ be given in \eqref{rper-def}. We use the relations \eqref{N} - \eqref{Nk-def} in the proof of Proposition \ref{prop:nonlinear-rough}. Following the latter proof we see that it suffices to improve the estimate of $N_k(\psi)$
\DETAILS{ which we reproduce here,  

\begin{align} \label{Nk-def'}	& N_k(\psi) := \oint  r_{\rm per}(z)  [(-\psi)  r_{\rm per}(z) ]^k ,
\end{align}}
in Proposition \ref{prop:Nk-est}, to which we proceed.

\DETAILS{We define explicit constants appearing in estimates below:
\begin{align} 
	& C_{1,\beta} = \sup_{f,g,h; P \in \{ P_-, P_+ \}} \frac{\oint \Tr [fR P g R P h P R]}{\|f\|_{L^2} \|\nabla g\|_{L^2}^{4/3} \|h\|_{L^2}^{2/3}}, \label{C1-beta} 
\end{align}
where the sup is taken over all possible $f,g,h,u, v$ such that the denominators are finite and
 $R =  r_{\rm per}(z)$ is given in \eqref{rper-def}.}



\begin{proposition} \label{prop:N2}
 Let Assumption \ref{A:diel} hold 
  and let $N_k$ be given by \eqref{Nk-def}. Assume that $\|\nabla \psi\|_{L^2} = o(1)$, then, for any $k\ge 2$,  we have the estimate
\begin{align} \label{Nk-est}	& \|\den[N_k(\psi)]\|_{L^2} 
	\ls \|\nabla \psi\|_{L^2}^k + 
	 e^{-\beta}  \|\nabla \psi\|_{L^2}^{4/3}  \|\psi\|_{L^2}^{2/3}\del_{k, 2},
\end{align}
where the constants associated with $\ls$ are independent of $\beta$ and $\del_{k, 2}$ is the Kronecker  delta. 
\end{proposition}
\begin{proof}
We begin with $k=2$. To improve upon estimate \eqref{Nk-est-rough}, we, following \cite{CLS}, use the partition of unity  
\begin{align}
 \label{P-+}P_1+ P_2= \one, \text{ with } 	P_1 := \chi_{h_{\rm per} < \mu} \text{ and } P_2 := \chi_{h_{\rm per} \geq \mu}. 
\end{align}
  Let $R_i \equiv R_i (z)=  r_{\rm per}(z)  P_i$ where $i=1, 2,$ $P_i$ and $r_{\rm per}(z)$  are given in \eqref{P-+} and \eqref{rper-def}. Recalling definition \eqref{Nk-def} of $N_k(\psi)$ and inserting the partition of unity,  $P_1 + P_2 = \one$, after each $R$ in the integrand of \eqref{Nk-def}, we arrive at 
\begin{align}
	N_2(\psi) = \sum_{a,b,c = 1,2}N^{abc}_2(\psi) \, ,
\end{align}
 where the  $N_2^{abc}(\psi)$, for $a,b,c =1, 2$, denote the operators
\begin{align}
	N^{abc}_2(\psi) =&  \oint R_a \phi R_b \psi R_c \, . \label{eqn:N-abc-form-def}
\end{align}
We estimate the terms individually.  Below, we use 
  the estimate (see \eqref{Del-r-est'})
\begin{align}\label{Del-r-est}\|(1-\Delta)^{\al} R_i(z)\|\ls\ & d^{\al-1}\ls 1,\ i=1, 2,\end{align}
for $\al\in [0, 1]$ and $z\in \Gam$, where 
\begin{align}\label{d-est-eta}d \equiv d(z):=\dist (z, \s(h_{\rm per}))\ge \frac14 .\end{align}

\textbf{Case 1: $(121)$ and $(212)$.} We estimate the case for $(121)$, the other case is done similarly. Since $P_1P_2 = 0$, we write
\begin{align}
	N^{(121)}_2(\psi) =&  \oint R_1 \psi R_2 P_2 \psi R_1 \\
		=&  \oint R_1 [P_1, \psi ]P_2 R_2 P_2 [\psi, P_1] R_1 \, . \label{eqn:N--+--form-def}
\end{align}
Applying Lemma \ref{lem:rhoA-by-tracial-norm} and equation 
\eqref{Del-r-est} to the r.h.s. and using that the operator norm is bounded by the $I^2$ norm, 
 we find 
\begin{align}
 \label{N2-est121'}	\|\den[N^{(121)}_2(\psi)] \|_{L^2} &\ls  \|(1-\Delta)^{3/4+\epsilon} N^{(121)}_2(\phi)\|_{S^2} \\
 \label{N2-est1'}		&\ls \absoint d^{-3} \| [P_1, \psi ]P_2 \|_{S^2}^2. 
	\end{align} 
where, recall,
\begin{align}
	\absoint :=& \frac{1}{2\pi} \int_\Gamma dz |f_{T} (z-\mu)|. \label{eqn:absoint-def}
\end{align}
A key observation allowing us to obtain an improved estimate is that the commutators lead to gradient estimates: 
\begin{lemma} \label{lem:commutator-P}
 Let Assumption \ref{A:diel} hold, 
 we have the estimate
\begin{align}
	\|[P_i, \psi] \|_{S^2} \ls \|\nabla \psi\|_{L^2},\ i=1, 2 .
\end{align}
\end{lemma}
\begin{proof}[Proof of Lemma \ref{lem:commutator-P}]
Since the identity commutes with any operator and $P_2 = \one - P_1$ (see \eqref{P-+}), we prove the lemma for $P_1$ only. Since $h_{\rm per}$ (see \eqref{rper-def}) has a gap at $\mu$, 
the Cauchy integral formula implies
\begin{align}
	P_1 = \frac{1}{2\pi i} \int_{\Gamma_1} (z-h_{\rm per})^{-1} = \frac{1}{2\pi i} \int_{\Gamma_1} r_{\rm per}(z) \label{eqn:P---explicit-form}
\end{align}
where $\Gamma_1$ is the contour $\{t + i ;  -c\le t < \mu\} \cup \{t - i ;   -c\le t < \mu\}
 \cup \{ -c -it + (1-t)i  : t \in [0,1] \} \cup  \{ \mu -it  + (1-t)i  : t \in [0,1] \}$, where $c > 0$ is any constant such that $h_{\rm per} > -c +1$, and the contour is traversed counter-clockwise. We see that
\begin{align}
	& [P_1, \psi] = \frac{1}{2\pi i} \int_{\Gamma_1} [ r_{\rm per}(z) ,\psi] \\
		=& \frac{1}{2\pi i} \int_{\Gamma_1}  r_{\rm per}(z) [\nabla \cdot, \nabla \psi] r_{\rm per}(z)  \notag \\
			& +  \frac{1}{2\pi i} \int_{\Gamma_1} r_{\rm per}(z) (2\nabla \psi \cdot \nabla) r_{\rm per}(z).  \, \label{eqn:P---commute-varphi}
\end{align}
Lemma \ref{lem:commutator-P} is now proved by an application of the Kato-Seiler-Simon inequality (\eqref{eqn:KSS}) to \eqref{eqn:P---commute-varphi} and noting that $\Gamma_1$ is compact and has length $O(1)$.
\end{proof}
Using Lemma \ref{lem:commutator-P} and estimates  \eqref{oint-decay} and \eqref{d-est-eta} in \eqref{N2-est1'} yields that
\begin{align}
	\|\den[N^{(121)}_2(\psi)] \|_{L^2} 
	 \ls \|\nabla \psi\|_{L^2}^2 \, .\label{eqn:final-nonlin-N-2-est-1}
\end{align}

\textbf{Case 2: $(112), (211), (122),$ $(221)$.} We estimate the case for $(112)$, the other cases are done similarly. Again, since $P_1P_2 = 0$, we write
\begin{align}
	N^{(112)}_2(\psi) =&  \oint R_1 \psi R_1  \psi  R_2 \\
		=&  \oint R_1 \psi R_1  [\psi, P_1] R_2 \, . \label{eqn:N--+-eqn1}
\end{align}
Using Lemma \ref{lem:rhoA-by-tracial-norm} as in with $N^{(121)}_2(\psi)$ in \eqref{N2-est121'}, we estimate \eqref{eqn:N--+-eqn1} as
\begin{align} \label{N112-est1}	\|\den[N^{(112)}_2(\psi)] \|_{L^2} &\ls  \absoint  d^{-1} \|\psi R_1 \| \| [\psi, P_1]R_2 \|_{S^2} .
\end{align}
where $\absoint$ is defined in \eqref{eqn:absoint-def}. By the inequality $\|A\|  \leq \|A\|_{I^p}$, for any $p < \infty$ for any operator $A$ on $L^2(\R^3)$, and the Kato-Seiler-Simon inequality \eqref{eqn:KSS}, we find  $\|\psi R_1 \|\le\|\psi R_1 \|_{S^6}\ls \|\psi\|_{L^6}$. Using this, together with Lemma \ref{lem:commutator-P}, in \eqref{N112-est1}, we obtain
\begin{align}
		\|\den[N^{(112)}_2(\psi)] \|_{L^2} \ls & \absoint  d^{-2} \|\psi\|_{L^6} \|\nabla \psi\|_{L^2} .\label{eqn:nonlin-N-2-mmp-2}
\end{align}
Combining this with  \eqref{oint-decay}, \eqref{d-est-eta} and Hardy-Littlewood's inequality \eqref{eqn:hardy-ineq} gives 
\begin{align}
	\|\den[N^{(112)}_2(\psi)] \|_{L^2} \ls & \|\nabla \psi\|_{L^2}^2. \label{eqn:final-nonlin-N-2-est-3}
\end{align}

\textbf{Case 3: $(111)$ and $(222)$.} 
 We use the $L^2$-$L^2$ duality to estimate the $L^2$ norm of $\den[N^{(qqq)}_2(\psi)], q=1, 2$. We have, by  \eqref{den-def1} and definition \eqref{Nk-def},
 \begin{align}\notag	\|\den[N^{(qqq)}_2(\psi)]\|_{L^2} &=\sup_{ \|f\|_{L^2}=1}|\int f \den[N^{(qqq)}_2(\psi)]|\\ \notag&=\sup_{ \|f\|_{L^2}=1}|\Tr [f N^{(qqq)}_2(\psi)]|\\
 \label{NL2-dual} &=\sup_{ \|f\|_{L^2}=1}|\oint \Tr (f R_q \psi R_q \psi R_q)|.\end{align} 
(In the last two lines, $f$ is considered as a multiplication operator.) 
To show that the integral on the r.h.s. converges absolutely, we follow the arguments in \eqref{fRphiRphiR}-\eqref{L2-est2'} to prove, 
\DETAILS{Let $R = r_{\rm per}(z)$ and $P_\pm$  be given in \eqref{rper-def} and \eqref{P-+}.  By definition \eqref{C1-beta}, we have to estimate
\begin{align}
	\oint \Tr [f R P_k g R P_k h P_k], \quad k=1,2, 
\end{align}
where $f, h \in L^2(\R^3)$, $g \in \dot{H}^1(\R^3)$.} 
 for $q=0, 1$,  
\begin{align} \label{L2-est2}	\big|\Tr( f R_q (\psi R_q)^2)\big|\ls & d^{-4/3}\|f\|_{L^2} \|\n \psi\|_{L^2}^{4/3} \|\psi \|_{L^2}^{2/3}. 
	\end{align}
Due to definition \eqref{d-est-eta} of $d \equiv d(z)$, this shows that  the integral on the r.h.s. of \eqref{NL2-dual} converges absolutely.
\DETAILS{$***$
Using 
that for any operator $A$ on $L^2(\R^3)$, we have
\begin{align}
	\|A\|_{S^3}^3 = \Tr(|A|^3) \leq \|A\| \Tr(|A|^2) = \|A\| \|A\|_{S^2}^2, \label{eqn:3-norm-bound-by-infty-and-2-norm}
\end{align}
Kato-Seiler-Simon's inequality \eqref{eqn:KSS}, and Hardy-Littlewood's inequality \eqref{eqn:hardy-ineq}, equation \eqref{eqn:fR-phiR-phiR-} shows 
\begin{align}
	|\Tr( f R_1 \psi R_1 \psi R_1)| \ls & \|f\|_{L^2} \|\n \psi\|_{L^2}^{4/3} \|\psi R_1\|_{L^2}^{2/3} \, . 
\end{align}
$***$}

		\begin{lemma} \label{lem:int-reduct}
For $q=1, 2$, we have
 \begin{align}
	\oint & \Tr [f R_q g R_q h R_q] \notag \\
		&= \frac{1}{2\pi i}\int_{\G_q} (f_{T}(z) - 1) \Tr [f R_q g R_q h R_q]. \label{int1}
\end{align}
\end{lemma}
\begin{proof}
Note that  the contour $\G$ in Figure \ref{fig:cauchy-int-contour'} is the union of two disjoint contours, $\G=\G_1 \cup \G_2$, with  $\G_1$ being the closed contour  and $\G_2$ unbounded one (i.e. the parts of $\G$ with $\Re \, z < \mu$ and $\Re \, z > \mu$). Hence, since $P_1/P_2$ projects to the spectrum of $h_{\rm per}$ on the left/right of $\mu$, the Cauchy theorem implies that 
\begin{align} \label{Rq-exp1}
	\int_{\G} dz \, &\Tr [f R_q g R_q h R_q] =\int_{\G_q} dz \, \Tr [f R_q g R_q h R_q]  .
\end{align}
 Next, we  note that, by Bloch's theory, we have
\begin{align}
	\int_{\G_q} dz \, \Tr [f R_q g R_q h R_q]& =\int_{\G_q} dz\, \int_{(\Om^*)^3} d\hat k d\hat k_1 d \hat k_2 \\
	&\times \Tr_{L^2_{\rm per}} f_{k-k_1} (R_q)_{k_1} g_{k_1-k_2} (R_q)_{k_2} h_{k_2-k} (R_q)_k . \label{eqn:RPminus-exp1}
\end{align}
Computing the trace in the complete orthonormal basis  of eigenvectors $\varphi_{m,k}$ of  $(R_q)_{k}$ (with eigevalues $\lambda_{m,k}$)  and inserting the complete orthonormal bases  of eigenvectors $\varphi_{n,k_1}$ and $\varphi_{r, k_2}$ of  $(R_q)_{k_1}$ and $(R_q)_{k_2}$ (with eigevalues $\lambda_{n,k_1}$ and $\lambda_{r,k_2}$) into \eqref{eqn:RPminus-exp1}, we see that 
\begin{align}
	 \int_{\G_q} \Tr [f R_q g & R_q h R_q] = \sum_{m,n,r}  \int_{(\Om^*)^3} d\hat k d\hat k_1 d \hat k_2\\
	&\times  \lan \varphi_{m,k}, f_{k-k_1} \varphi_{n, k_1} \ran \lan \varphi_{n, k_1}, g_{k_1-k_2} \varphi_{r, k_2} \ran \lan \varphi_{r, k_2}, h_{k_2-k} \varphi_{m, k} \ran \\
		&\times \int_{\G_q} dz\, \frac{1}{(z-\lambda_{m,k})(z-\lambda_{n,k_1})(z-\lambda_{r, k_2})} . \label{eqn:RPminus-exp2}
\end{align}
For $q=1$, since $P_1$ projects to the spectrum of $h_{\rm per}$ on the left of $\mu$, we see that $\lambda_{m,k}, \lambda_{n,k_1}, \lambda_{r, k_2} < \mu$. In particular, these eigenvalues are in the left closed contour in Figure \ref{fig:cauchy-int-contour'}. Consequently, Cauchy's integral formula shows that the 
integal in \eqref{eqn:RPminus-exp2} is identically zero. Similar argument applies to  $q=2$. This shows that

\begin{align}
 \frac{1}{2\pi i}\int_{\G_q} \Tr [f R_q g R_q h R_q]  = 0.
\end{align}
Thus 
\eqref{int1} follows.
\end{proof}

Using the explicit form of the Fermi-Diract distribution $f_{T}$ in \eqref{fT}, 
  we see that
\begin{align}
	|f_{T}(z-\mu) - 1| = \frac{e^{\beta(\Re \, z-\mu)}}{|1+e^{\beta(z-\mu)}|} \, .\label{eqn:fFD-1-exp1}
\end{align}
By condition \eqref{eta-R3} and by the choice of the contour, $\G$, in Figure \ref{fig:cauchy-int-contour'}, we see that the if $z \in \G$, then $\Re \, z$ is at least 
at the distance $\ge 1$ from $\mu$. 
Hence, for $z$ in a contour $\G$,  \eqref{eqn:fFD-1-exp1} implies that
\begin{align}\label{fT-est2}
	|f_{T}(z-\mu) - 1|  \ls  e^{-\beta}.
\end{align}

 Applying estimates \eqref{fT-est2} and \eqref{L2-est2} to the r.h.s. of \eqref{int1} and recalling the definition \eqref{d-est-eta} of $d\equiv d(z)\ge \frac14$, 
 we arrive at the inequality 
\begin{align} \label{int1-est}	|\oint & \Tr [f R_q g R_q h R_q]| \notag \\
		&\ls  e^{-\beta}\|f\|_{L^2} \|\n \psi\|_{L^2}^{4/3} \|\psi \|_{L^2}^{2/3}.
\end{align}
This inequality, together with the relation \eqref{NL2-dual}, gives 
\begin{align}
	\|\den[N^{(qqq)}_2(\phi)]\|_{L^2} \ls & e^{-\beta} \|\nabla \phi\|_{L^2}^{4/3} \|\phi\|_{L^2}^{2/3}. \label{eqn:final-nonlin-N-2-est-5}\end{align}
Inequalities \eqref{eqn:final-nonlin-N-2-est-1}, \eqref{eqn:final-nonlin-N-2-est-3} and \eqref{eqn:final-nonlin-N-2-est-5} imply estimate \eqref{Nk-est} for $k=2$. 

Now we estimate $N_k$ for $k>2$.  
By \eqref{oint-decay} and \eqref{eqn:resolvent-bound-on-G}, it suffices to estimate $\den[R(\phi R)^k]$ where $R =  r_{\rm per}(z)$ is given in \eqref{rper-def}. Using Lemma \ref{lem:rhoA-by-tracial-norm}, we see that
\begin{align}
	\|\den[R(\phi R)^k]\|_{L^2} \ls & \|(1-\Delta)^{4/3+\epsilon}R(\phi R)^k\|_{S^2} \\
		\ls & \|(\phi R)^k\|_{S^2}. \label{eqn:lem:L2kgeq3-eqn-1}
\end{align}
Using H\"{o}lder's inequality with $\frac{1}{2} = \frac{1}{6} + \frac{1}{6} + \frac{1}{6} + $ another $k$ terms of $\frac{1}{\infty}$, \eqref{eqn:lem:L2kgeq3-eqn-1} becomes 
\begin{align}
	\|\den[R(\phi R)^k]\|_{L^2}  \leq & \|\phi R\|_{S^6}^3 \|\phi R\|^{k-3} \\
		\leq & \|\phi R\|_{S^6}^k , \label{eqn:lem:L2kgeq3-eqn-2}
\end{align}
where the last line follows since $\| \cdot \| \leq \| \cdot \|_{S^p}$ for $p < \infty$. Combining with Kato-Seiler-Simon's inequality \eqref{eqn:KSS} and Hardy-Littlewood's inequality \eqref{eqn:hardy-ineq}, \eqref{eqn:lem:L2kgeq3-eqn-2} implies \eqref{Nk-est} for $k\ge 3$. 
\end{proof}



The rest of the proof of  Proposition \ref{prop:nonlinear} proceeds as the proof of in Proposition \ref{prop:Nk-est}.
\DETAILS{Now, we complete the proof of Proposition \ref{prop:nonlinear}. Lemma \ref{lem:L2kgeq3}  shows that 
  if $\|\psi\|_{L^2} < \infty$ and either $\|\nabla \psi\|_{L^2} = o(1)$ or $\|\nabla \psi\|_{\dot{H}^{1}} = o(1)$, then series \eqref{N-series} converges absolutely in $L^2$. Moreover this lemma and  Proposition \ref{prop:N2} imply the estimate \[\|\den[N (\psi)]\|_{L^2} \ls  \|\nabla \psi\|_{L^2}^{4/3}\|\psi\|_{L^2}^{2/3} \|\psi\|_{H^j}^{k-2}.\]  

Now, using series \eqref{N-series}, we write  
\begin{align}
	N(\psi_1) - N(\psi_2)= \sum_{k \geq 2} \den[N_k(\psi_1) - N_k(\psi_2)] \, .
\end{align}
By definition \eqref{Nk-def}, $N_k(\psi)$ is an $k$-th degree monomial in $\phi$.
Hence, we can expand $N_k(\psi_1) - N_k(\psi_2)$ in the following telescoping form
\begin{align}
	x^k& - y^k = x^{k-1}(x-y) + x^{k-2}(x-y)y + \cdots + (x-y)y^k\, . \label{N-telescoping}
\end{align}
The proof of Proposition \ref{prop:nonlinear-micro} follows by applying appropriate and straightforward extension of  Proposition \ref{prop:N2}  and Lemma \ref{lem:L2kgeq3} to each term in the expansion of  $N_k(\psi_1) - N_k(\psi_2)$ given in \eqref{N-telescoping}.} 
 \end{proof}

\end{document}